\documentclass[aps,prx,onecolumn,superscriptaddress,nofootinbib,10pt,longbibliography]{revtex4-2}

\usepackage[T1]{fontenc}
\usepackage[utf8]{inputenc}
\usepackage{lmodern}
\usepackage{csquotes}
\usepackage{amsmath,amssymb,amsfonts,amsthm}
\usepackage{bbm}
\usepackage{mathrsfs}
\usepackage{physics}
\usepackage{centernot}
\usepackage{comment}
\usepackage{graphicx}
\usepackage{subcaption}
 \usepackage{booktabs}
\usepackage{multirow}
\usepackage{enumitem}
\usepackage{listings}
\usepackage{tikz}
\usetikzlibrary{positioning,calc}
\usepackage{tikz-cd}
\usepackage{pst-node}
\usepackage{mathtools}

\usepackage[pdfusetitle,
  bookmarks=true,bookmarksnumbered=false,bookmarksopen=false,
  breaklinks=false,pdfborder={0 0 0},pdfborderstyle={},
  colorlinks=true, linkcolor=magenta, urlcolor=blue, citecolor=blue]{hyperref}
\usepackage{cleveref}

\usepackage{newunicodechar}
\newunicodechar{ρ}{\rho}
\newunicodechar{σ}{\sigma}
\newunicodechar{λ}{\lambda}
\newunicodechar{Λ}{\Lambda}
\newunicodechar{μ}{\mu}
\newunicodechar{ν}{\nu}
\newunicodechar{ψ}{\psi}
\newunicodechar{ϕ}{\phi}
\newunicodechar{φ}{\varphi}
\newunicodechar{π}{\pi}
\newunicodechar{α}{\alpha}
\newunicodechar{ϵ}{\epsilon}
\newunicodechar{ε}{\varepsilon}
\newunicodechar{δ}{\delta}
\newunicodechar{ω}{\omega}
\newunicodechar{β}{\beta}
\newunicodechar{ξ}{\xi}
\newunicodechar{κ}{\kappa}
\newunicodechar{∗}{\textasteriskcentered}

\newtheorem{theorem}{Theorem}
\numberwithin{theorem}{section}
\newtheorem{proposition}[theorem]{Proposition} 
\newtheorem{lemma}[theorem]{Lemma}    
\newtheorem{corollary}[theorem]{Corollary} 

\theoremstyle{definition}
\newtheorem{definition}[theorem]{Definition}   
\newtheorem{example}[theorem]{Example}         

\newtheoremstyle{remarkbold} 
  {6pt} 
  {6pt} 
  {\itshape} 
  {} 
  {\bfseries} 
  {.} 
  { } 
  {} 

\theoremstyle{remarkbold}
\newtheorem*{remark}{Remark}

\let\oldremark\remark
\let\endoldremark\endremark
\renewenvironment{remark}
  {\oldremark\itshape}
  {\endoldremark}
  
\newcommand{\EndS}[1]{\operatorname{End}^{\mathfrak{S}_{#1}}}
\newcommand{\esn}{\EndS{n}}
\newcommand{\HS}{\mathcal{H}}
\newcommand{\n}{^{\otimes n}}
\newcommand{\A}{\mathcal{A}}
\newcommand{\CPTP}{\operatorname{CPTP}}

\begin{document}
\title{Permutation Invariant Optimization Problems in Quantum Information Theory: A Framework for Channel Fidelity and Beyond}
\author{Bjarne Bergh}
\email{bb536@cam.ac.uk}
\affiliation{Department of Applied Mathematics and Theoretical Physics,\\
University of Cambridge, Cambridge CB3 0WA, United Kingdom}

\author{Marco Parentin}
\email{marco.parentin@gmail.com}
\affiliation{Politecnico di Torino, 10129 Torino, Italy}

\begin{abstract}
Exploiting permutation invariance to reduce the exponential scaling of semidefinite programs in quantum information has emerged as a powerful computational technique. In this work, we develop a systematic framework for using this reduction via Schur--Weyl duality for optimization problems, and establish methods that allow one to work fully inside the permutation invariant subspace while performing operations such as (partially) applying channels and taking (partial) traces, or computing expressions like the quantum relative entropy. We then apply our techniques to the problem of computing efficient lower bounds on the channel fidelity over $n$ parallel uses of a quantum channel. The algorithm, which we call symmetric seesaw method, exploits permutation-invariant codes to yield improved lower bounds on the channel fidelity over $n$ uses of the depolarizing and amplitude-damping channel in the regime of tens of channel uses, and was used in~\cite{Parentin2026} to demonstrate non-asymptotic superactivation of quantum capacity for $n = 17$. An implementation of our methods, aimed at being suitable for various quantum information theoretic optimization problems, is also available as an open-source Python package.
\end{abstract}

\maketitle

\tableofcontents

\section{Introduction}
\label{sec:introduction}

Semidefinite programming (SDP) has become a central computational tool in quantum information theory~\cite{Vandenberghe1996, Wang2018SDP, Skrzypczyk2023}. The main objects of quantum mechanics---states, measurements, and channels---are described by positive semidefinite operators satisfying linear constraints, making many fundamental problems naturally amenable to SDP formulations. The resulting programs can be solved efficiently by interior-point algorithms, while SDP duality often reveals deep operational connections between seemingly disparate quantities. Recent years have seen a much progress in quantum information theory in which SDP techniques were used to resolve long-standing open problems: the irreversibility of entanglement theory~\cite{lami2023nosecondlaw}, the computability of zero-error entanglement cost under positive-partial-transpose operations~\cite{lami2024computablecost}, and efficient converse bounds for quantum channel capacities~\cite{Wang2016, Wang2019}, to name a few. Despite this progress, the tractability of SDPs is inherently limited by the dimension of the underlying Hilbert space, which grows exponentially when one considers $n$ independent and identically distributed (i.i.d.) copies of a channel $\mathcal{N}$ or state $\rho$, denoted $\mathcal{N}^{\otimes n}$ or $\rho^{\otimes n}$. 

This exponential growth is the central computational obstruction which we try to address in this paper. It arises, for instance, when an SDP-computable quantity $f$ is non-additive under tensor products, i.e. $f(\rho^{\otimes n}) \neq n f(\rho)$ or $f(\mathcal{N}^{\otimes n}) \neq nf(\mathcal{N})$, so one is interested in investigating its behavior for increasing values of $n$. Such non-additivity is a pervasive phenomenon in quantum information \cite{Wilde_book}, occurring both in entanglement theory and quantum channel coding, and can be a beneficial feature: for instance, one may take multiple copies of a channel to achieve strictly higher transmission rates than are possible with a single use~\cite{Smith2008, Hastings2009, Smith2011}, or take multiple copies of a state to obtain a strictly tighter upper bound on distillable entanglement than the single-copy value of the relevant monotone~\cite{Audenaert2000, Audenaert2002, Wang2017_Rains, Rubboli2022}.
This motivates to find techniques to tackle this exponential scaling. The key structural observation in our work enabling an escape from this curse of dimensionality is \textit{permutation-invariance}: in many cases, the SDPs are invariant under permutations of the $n$ copies, i.e., for any solution there exists an equally feasible and optimal solution that is symmetric with respect to the $n$ systems. This observation, combined with the representation theory of the symmetric group $\mathfrak{S}_n$ via Schur--Weyl duality~\cite{Vallentin2009, Barman2017, Fawzi2022}, allows one to block-diagonalize the SDP variables and reduce the problem size from $\mathcal{O}(\exp(n))$ to $\mathcal{O}(\mathrm{poly}(n))$.

\

\paragraph{Prior work.}
The exploitation of symmetry to simplify optimization problems has a long history~\cite{ Vallentin2009,Klerk2010, Handbook_semidefinite2012}, with applications ranging from classical coding theory~\cite{Schrijver2005, Gijswijt2006, Laurent2007, Litjens2016, Polak2020} to graph theory~\cite{DeKlerk2007, Brosch2024} and polynomial optimization~\cite{Gatermann2004, Riener2013, Raymond2018}. In quantum information, the symmetric group acts naturally on $n$-fold tensor products $\mathcal{H}^{\otimes n}$ by permuting the tensor factors, and the algebra of permutation-invariant operators $\mathrm{End}^{\mathfrak{S}_n}(\mathcal{H}^{\otimes n})$ plays a distinguished role. Its structure is governed by Schur--Weyl duality, which yields a block-diagonalization of $\mathrm{End}^{\mathfrak{S}_n}(\mathcal{H}^{\otimes n})$ into irreducible blocks indexed by Young diagrams. The number of such blocks, and their sizes, grow only polynomially in $n$, which is the source of the complexity reduction (see Section \ref{sec:perm-inv} for formal definitions).

This observation has recently found applications in several areas of quantum information. Specifically, it was used in~\cite{Fawzi2022} to obtain efficiently computable bounds on quantum channel capacities, and in~\cite{Bergh2025}, to efficiently compute $n$-shot error exponents in quantum channel discrimination. Very recently, \cite{Fang2024, Fang2025} used it to efficiently approximate regularized relative entropies.
Alongside these developments, the permutation-invariant structure underlies several foundational tools in quantum information theory, notably the quantum de Finetti theorem~\cite{Caves2002, Christandl2007} and the theory of $k$-extendible states~\cite{Doherty2002, Doherty2004, Berta2019}, which provide a hierarchy of SDP-computable relaxations of the separability problem converging to the set of separable states as $k \to \infty$. This was used by \cite{Berta2022, Chee2023, Kossmann2025} to derive efficient SDP outer approximations to the channel fidelity.

\

\paragraph{This work.}
In this paper, we develop the theory of permutation-invariance in quantum information further, and establish techniques to compute operations such as partial traces and channel concatenations wholly within the permutation invariant subspaces. We also provide a complete open-source Python package that implements the Schur–Weyl symmetry reduction explicitly without ever constructing exponentially large matrices, for use in semidefinite programming and other optimization problems in quantum information theory (Section \ref{sec:Python_package}).

As main application, we study efficiently computable \textit{lower bounds} to the channel fidelity over $n$ copies of a quantum channel, denoted $F_c(\mathcal{N}^{\otimes n}, d)$ and defined as the maximum fidelity with which a maximally entangled state of Schmidt rank $d$ can be transmitted through $n$ parallel uses of a channel $\mathcal{N}_{A \to B}$, optimized over all encoding and decoding operations $(\mathcal{E}_n, \mathcal{D}_n)$ (see Section \ref{sec:symmetric_seesaw} for a formal definition). Using a \textit{seesaw method}~\cite{Reimpell2005, Reimpell2008, Fletcher2007} and phrasing it in terms of the maximal singlet fraction, we restrict the optimization to permutation-invariant (PI) codes \cite{Leung1997}, and use Schur--Weyl block-diagonalization to find SDP-computable poly-($n$) lower bounds to $F_c(\mathcal{N}^{\otimes n}, d)$ in polynomial time, yielding a \textit{symmetric seesaw method} (Theorem \ref{thm:seesaw_polytime}), which also provides an explicit coding protocol. These results are complementary to those of \cite{Kossmann2025}, which focuses on upper bounds, and allow to easily reach number of channel uses beyond $n =20$ for a generic qubit channel. 

\

The method was used in~\cite{Parentin2026} to demonstrate, for the first time, superactivation of quantum capacity~\cite{Smith2008} in the non-asymptotic regime, showing that $n \leq 17$ joint channel uses suffice for the original Smith--Yard channel pair to transmit at a fidelity which is impossible to achieve with either channel alone (even if unlimited copies are allowed). In this paper, we derive the theoretical foundations of this method in detail, highlighting how it fits into a more general framework, also applying it to different problems, for example showing how it gives rise to improved lower bounds on channel fidelity for $n \leq 20$ uses of the amplitude-damping channel and the depolarizing channel (Section \ref{sec:numerical_examples}). Specifically, for the amplitude damping channel, the symmetric seesaw method achieves error probabilities below $1\%$ for damping probability $\gamma < 0.2$.

Additionally, as a separate application, we demonstrate the effectiveness of the method to efficiently solve relative entropy programs (see also \cite{Fang2025}), yielding an efficiently computable explicit algorithm to obtain the Rains relative entropy~\cite{Rains2001SemidefiniteDistillableEntanglement} on $n$ copies of a bipartite state $\rho_{AB}^{\otimes n}$ (Theorem \ref{thm:rains_polytime}).

\

On the theoretical level, in Section \ref{sec:perm-inv}, we provide general results on how to efficiently compute partial traces and concatenations of permutation covariant channels in the permutation-invariant subspace (Propositions \ref{prop:serial_concatenation} and \ref{prop:general_serial_contenation}) and generalize the results to matrix $*$-algebras in Section \ref{sec:cq-algebras}. These results are extensively used in the coding setting of the symmetric seesaw method (Section \ref{sec:symmetric_seesaw}), but are expected to be of independent interest and useful more broadly in other problems of quantum information. Besides the two explicit examples provided here, our framework applies to several other SDPs in quantum information, such as all SDPs based on $k$-extendibility studied in \cite{Kaur2021,Singh2025} or those related to PPT-entanglement manipulation (e.g., those introduced in \cite{lami2024computablecost, Wang2016_distillable}). 

We leave as an open problem to combine the permutation-invariance investigated in this work with additional group symmetries (e.g., those investigated in \cite{Singh2021}), to further simplify certain optimization problems problem (see also \cite{Polak2020} or \cite[Section 4.2]{Fawzi2022}).

\subsection{Organization}
\label{subsec:organization}

This paper is organized as follows. Section~\ref{sec:preliminaries} introduces notation, quantum channels, the Choi representation, semidefinite programming, and the theory of permutation-invariant operators. Section~\ref{sec:perm-inv} develops the algebraic machinery for efficiently working with permutation-invariant operators: orbit bases, count matrices, the Schur--Weyl $*$-isomorphism, and polynomial-time algorithms for serial concatenation and partial traces. Section~\ref{sec:seesaw} reviews the seesaw method and its connection with the maximal singlet fraction, then presents the symmetric seesaw method and the symmetric power iteration algorithm, and giving examples for two important examples of channels. Section~\ref{sec:other_application_to_REE_programs} applies the toolbox to efficiently approximate the Rains relative entropy on $n$ copies of a state. Finally, Section~\ref{sec:Python_package} describes briefly the implementation all of these techniques in our open-source python package \href{https://github.com/bbbergh/permqit}{\texttt{permqit}}.

\section{Preliminaries}
\label{sec:preliminaries}

\subsection{Basic Notation}
\label{subsec:notation}

Let $\mathcal{H}$ be a finite-dimensional complex Hilbert space of dimension $d_{\mathcal{H}}$, with standard basis $\{\ket{i}_{\mathcal{H}}\}_{i=1}^{d_{\mathcal{H}}}$. We write $[d] = \{1, \ldots, d\}$. We denote by $\mathcal{L}(\mathcal{H})$ the set of linear operators on $\mathcal{H}$, by $\mathcal{P}(\mathcal{H})$ the set of positive semidefinite operators, and by $\mathcal{D}(\mathcal{H}) = \{\rho \in \mathcal{P}(\mathcal{H}) : \Tr(\rho) = 1\}$ the set of quantum states (density operators). Pure states are density operators of unit rank: $\rho = \psi \equiv \ket{\psi}\!\bra{\psi}$ for some unit vector $\ket{\psi} \in \mathcal{H}$. Given two density operators $\rho, \sigma$, their (squared) fidelity~\cite{Uhlmann1976} is a measure of their closeness and is defined as $F(\rho, \sigma) \coloneqq \bigl(\Tr\!\bigl[\sqrt{\sqrt{\sigma}\, \rho\, \sqrt{\sigma}}\,\bigr]\bigr)^{\!2}$. When one state is pure, it simplifies to the overlap $F(\rho, \psi) = \bra{\psi}\rho\ket{\psi}$.

We denote by $\mathbbm{1}_{\mathcal{H}}$ the identity operator and by $\pi_{\mathcal{H}} = \mathbbm{1}_{\mathcal{H}}/d_{\mathcal{H}}$ the maximally mixed state. For Hermitian operators $\rho, \sigma \in \mathcal{L}(\mathcal{H})$, we write $\sigma \geq \rho$ if $\sigma - \rho \in \mathcal{P}(\mathcal{H})$. We label quantum systems by capital Roman letters: in a communication setting, $A$ denotes Hilbert spaces belonging to the sender (Alice) and $B$ those belonging to the receiver (Bob). A bipartite operator is written as $X_{RA} \in \mathcal{P}(R \otimes A)$, where the subscripts indicate the local subsystems. The marginal on $A$ is $X_A = \Tr_R X_{RA} = \sum_{i} (\mathbbm{1}_A \otimes \bra{i}_R) X_{RA} (\mathbbm{1}_A \otimes \ket{i}_R)$. We write $R \simeq A$ if $d_R = d_A$.

The maximally entangled state (MES) of dimension $d$ associated with the standard basis is denoted as $\Phi^d_{RA} = \ket{\Phi}\!\bra{\Phi}_{RA}$, where:\begin{equation}
\label{eq:MES}
\ket{\Phi}_{RA} \coloneqq \frac{1}{\sqrt{d}} \sum_{i=1}^{d} \ket{i}_R \ket{i}_A.
\end{equation} We write the corresponding unnormalized MES as $\Gamma^d_{RA} \coloneqq d\,\Phi^d_{RA}$. Note that $\Tr_A \Phi^d_{RA} = \pi_R$ and $\Tr_R \Phi^d_{RA} = \pi_A$. Given a bipartite state $\rho_{RA}$, its fidelity with respect to the MES is called its \textit{singlet fraction} \cite{Horodecki1999}\begin{equation}
    F(\rho, \Phi^d) = \bra{\Phi^d}\rho\ket{\Phi^d},
\end{equation} and can be thought of as the probability of passing an \textit{entanglement test}, i.e. of getting a positive answer after a measurement of the form $\{\Phi_{RA}^d, \mathbbm{1}_{RA} - \Phi_{RA}^d\}$.

A quantum state $\rho_{AB}$ is \textit{separable}, denoted $\rho_{AB}\in \mathrm{SEP}(A;B)$, if it can be written as a convex combination
\begin{equation}
\label{eq:separable_state}
    \rho_{AB} = \sum_{x} p_X(x)\, \rho^x_A \otimes \sigma^x_B,
\end{equation}
for some probability distribution $p_X(x)$ and sets of states $\{\rho^x_A\}$, $\{\sigma^x_B\}$. A state that is not separable is called \textit{entangled}.

Given $n$ copies of $\mathcal{H}$, we write $\mathcal{H}^{\otimes n} = \mathcal{H} \otimes \cdots \otimes \mathcal{H}$ with standard basis $\{\ket{\underline{i}}\}_{\underline{i} \in [d_{\mathcal{H}}]^n}$, where $\underline{i} = (i_1, \ldots, i_n)$ with $i_k \in [d_{\mathcal{H}}]$ for all $k = 1, \ldots, n$.

We denote by $\mathrm{CP}(A \to B)$ the set of completely positive maps from $\mathcal{L}(A)$ to $\mathcal{L}(B)$. A quantum channel $\mathcal{N}_{A \to B} \in \CPTP(A \to B)$ is a completely positive and trace-preserving (CPTP) map. Quantum channels describe the most general physical operation on quantum systems, including local processing (encoding $\mathcal{E}$, decoding $\mathcal{D}$) and transmission operations (denoted as $\mathcal{N}$). The ideal communication channel is the identity $\mathrm{id}_{A \to B}$, which transmits all inputs unaltered.

Given a channel $\mathcal{N}_{A \to B}$, its adjoint (with respect to the Hilbert--Schmidt inner product) is the map $\mathcal{N}^*_{B \to A}$ satisfying $\Tr[Y\, \mathcal{N}(X)] = \Tr[\mathcal{N}^*(Y)\, X]$ for all $X \in \mathcal{L}(A)$, $Y \in \mathcal{L}(B)$. We denote by $\mathrm{CPU}(A \to B)$ the set of completely positive and unital maps. One has the correspondence \begin{equation}
    \label{eq:unitality-trace-preservation-are-dual}
    \mathcal{N} \in \CPTP(A \to B) \iff \mathcal{N}^* \in \mathrm{CPU}(B \to A),
\end{equation} i.e., unitality and trace-preservation are dual notions. Working with CPU maps instead of CPTP maps corresponds to working in the Heisenberg picture of quantum mechanics (linear transformation of observables) rather than the Schr\"odinger picture (linear transformation of states)~\cite{Wolf2012}.

Among the several known characterizations of CPTP maps, in this work we will extensively use the Choi representation~\cite{Choi1975}. Every channel $\mathcal{N} \in \CPTP(A \to B)$ is uniquely described by its Choi state, defined for $R \simeq A$ as $\Phi^{\mathcal{N}}_{RB} = (\mathrm{id}_R \otimes \mathcal{N}_{A \to B})(\Phi^{d_A}_{RA})$. The unnormalized Choi matrix is $\Gamma^{\mathcal{N}}_{RB} = d_A\, \Phi^{\mathcal{N}}_{RB} = (\mathrm{id}_R \otimes \mathcal{N}_{A \to B})(\Gamma^{d_A}_{RA})$. Notice that with $d_A = d_B = d$, one has $\Phi^{\mathrm{id}}_{RB} = \Phi^d_{RB}$ and $\Gamma^{\mathrm{id}}_{RB} = \Gamma^d_{RB}$. By the Choi--Jamio\l kowski isomorphism, the Choi state contains all the information about the map $\mathcal{N}$. In particular, the CPTP and CPU conditions become linear constraints:
\begin{equation}
\label{eq:choi_cptp}
\Gamma^{\mathcal{N}}_{RB} \geq 0, \quad \Tr_B \Gamma^{\mathcal{N}}_{RB} = \mathbbm{1}_R  \quad \iff \quad \mathcal{N} \in \CPTP(A \to B),
\end{equation}
\begin{equation}
\label{eq:choi_cpu}
\Gamma^{\mathcal{N}}_{RB} \geq 0, \quad \Tr_R \Gamma^{\mathcal{N}}_{RB} = \mathbbm{1}_B \quad \iff \quad \mathcal{N} \in \mathrm{CPU}(A \to B).
\end{equation}
We denote by $\mathcal{C}(A;B)$ the set of bipartite operators that are valid Choi matrices of $\CPTP(A \to B)$ maps and by $\mathcal{C}^*(A;B)$ the set of valid Choi matrices of $\mathrm{CPU}(A \to B)$ maps, i.e.:\begin{align}
\label{eq:choi_adj_set}
    \mathcal{C}^*(A:B) &\coloneqq \left\{X_{AB}: X_{AB} \geq 0, \ \Tr_{A}X_{AB} = \mathbbm{1}_B\right\},\\
\label{eq:choi_set}
    \mathcal{C}(A:B) &\coloneqq \left\{X_{AB}: X_{AB} \geq 0, \ \Tr_{B}X_{AB} = \mathbbm{1}_A\right\}.
\end{align} 
Obviously, we have $\mathcal{C}^*(A:B) = \mathcal{C}(B:A)$. The linearity of these constraints makes the Choi representation ideally suited for convex optimization.

\

In communication settings, we will often consider serial concatenations, parallel concatenations and adjoint of quantum channels in the Choi representation. 
\

Given $\mathcal{N}_{A \to B}$ and $\mathcal{D}_{B \to C}$, their serial concatenation is denoted by $\mathcal{M}_{A \to C} = \mathcal{D}_{B \to C}\circ \mathcal{N}_{A \to B}$ and its Choi representation is:\begin{equation}
\label{eq:concatenation_choi}
   \Gamma^{\mathcal{M}}_{RC}  = \Tr_B[(\Gamma^{\mathcal{N}}_{RB})^{T_B} \cdot \Gamma^{\mathcal{D}}_{BC}]
\end{equation}
where we denoted as $T_B$ the partial transpose map over $B$, acting on $X_{AB}\in \mathcal{L}(AB)$ as $X_{AB}^{T_B} = \sum_{j,j'=1}^{d_B} (\mathbbm{1}_A \otimes \ket{j}_{B}\bra{j'}_B) X_{AB} (\mathbbm{1}_A \otimes \ket{j}_{B}\bra{j'}_B)$, and there are implicit identities on the $R$ and $C$ systems.
\

Given $\mathcal{M}_{A_1 \to B_1}$ and $\mathcal{N}_{A_2 \to B_2}$, the tensor product channel $\mathcal{O}_{A_1A_2 \to B_1B_2} =\mathcal{M}_{A_1 \to B_1}\otimes \mathcal{N}_{A_2 \to B_2} $ has Choi representation 
\begin{equation}
  \Gamma^{\mathcal{O}}_{A_1A_2B_1B_2} =\Gamma^{\mathcal{M}}_{A_1B_1}\otimes \Gamma^{\mathcal{N}}_{A_2B_2}
\end{equation}
where strictly speaking there is also a reordering of subsystems (from $A_1B_1A_2B_2$ to $A_1A_2B_1B_2$) taking place, which we don't make explicit. Of particular interest is the case $\mathcal{N}^{\otimes n}_{A^n \to B^n}$, i.e. the $n$-fold tensor product of $\mathcal{N}_{A \to B}$, where one has:
\begin{equation}
\label{eq:choi_tensor_product}
    \Gamma^{\mathcal{N}^{\otimes n}}_{A^n B^n} =(\Gamma^{\mathcal{N}}_{AB})^{\otimes n}.
\end{equation}
Finally, if $\mathcal{N}^*_{B \to A}$ denotes the adjoint map of $\mathcal{N}_{A \to B}$, its Choi representation has the form:\begin{equation}
\label{eq:adjoint_choi}
    \Gamma^{\mathcal{N}^*}_{BA} = (\Gamma^{\mathcal{N}}_{AB})^{T}
\end{equation}
where the transpose is taken in the basis chosen for the maximally entangled state $\Phi$ (or $\Gamma$) and again we don't make changes in subsystem ordering explicit.

\subsection{Semidefinite Programming and Symmetry Reduction}
\label{subsec:sdp}

Semidefinite programs (SDPs)~\cite{Vandenberghe1996} are a subclass of convex optimization problems that have proven very useful in quantum information theory (see e.g.~\cite{Wang2018SDP, Skrzypczyk2023} for reviews). 
An SDP is a convex optimization program, where one optimizes a linear objective function over an affine subspace of the convex cone of positive semi-definite matrices. More formally, we use the notation of \cite{Watrous2018} and define a complex SDP as a triple $(C, D, \Psi),$ where $C, D$ are Hermitian and and $\Psi$ is an Hermiticity-preserving map:\begin{equation}\label{eq:SDP_definition}
\begin{aligned}
& \text{Primal} \ \ \ \ \ \ \ \ \ \ \ \ \ \ \ \ \ \  &\text{Dual} \ \ \ \ \ \  \\
     \max \quad & \Tr [CX] \ \ \ \ \ \ \ \ \ \ \ \ \ \ \ \ \ \  &\min  \Tr[DY]  \\
    \text{subject to:} \quad  & X \geq 0  &\text{subject to:}\ \ \ \ Y\geq 0\\
    & \Psi(X) \leq D &\ \ \ \ \Psi^*(X) \geq C 
\end{aligned}
\end{equation}
where $\Psi^*$ is the dual map to $\Psi$. Either problem is called feasible if there exists a variable satisfying the corresponding constraint. For these two problems, we define their optimal attained values:
\begin{align}
    \alpha \coloneqq \sup\{\Tr[CX]: \Psi(X) \leq D, X \geq 0\}\\
    \beta \coloneqq \inf \{\Tr[DY]: \Psi^*(Y) \geq C, Y \geq 0\}
\end{align}
where $\alpha = - \infty$ if the primal problem is not feasible and $\beta = + \infty $ if the dual problem is not feasible.
Every SDP satisfies weak duality, which states that $\alpha \leq \beta$ and implies that any valid variable of the dual problem upper bounds the optimal value of the primal problem. 

If there exists a $X \geq 0$ (resp. $Y \geq 0$) such that $D - \Psi(X)$ (resp. $\Psi^*(Y)-C$) is positive definite, then the primal (resp. dual) problem is said to be strictly feasible. The equality $\alpha = \beta$ is called strong duality and holds if the SDP satisfies the Slater condition, which is that both the primal and dual problems are strictly feasible (there are other weaker conditions, see e.g. \cite{Watrous2018}). SDPs are efficiently solvable via interior-point methods, using for instance CVX~\cite{Grant2014} with solvers such as Mosek~\cite{Mosek2024} and SeDuMi~\cite{Sturm1999}, or the Python packages CVXPY~\cite{Diamond2016} and PICOS~\cite{Sagnol2022}.
\

In many cases of practical interest, an SDP is invariant under some group action and can be simplified by exploiting this symmetry~\cite{Vallentin2009}. 
Let $G$ be a finite group, and let $g \mapsto U(g)$ be a unitary representation of $G$ on a Hilbert space $\mathcal{H}$. A quantum state $\rho \in \mathcal{D}(\mathcal{H})$ is called \textit{$G$-invariant} if $U(g)\, \rho\, U(g)^\dagger = \rho$ for all $g \in G$, and we write $\rho \in \mathrm{End}^{G}(\mathcal{H})$. The \textit{twirl channel} with respect to $G$ is the CPTP map: \begin{equation}
\label{eq:twirl_channel}
    \mathcal{T}_G(\rho) \coloneqq \frac{1}{|G|}
    \sum_{g \in G} U(g)\, \rho\, U(g)^\dagger
\end{equation}
which projects any state onto the $G$-invariant subspace. By construction, $\mathcal{T}_G(\rho) \in \mathrm{End}^{G}(\mathcal{H})$ for every $\rho$.
We say that an SDP in primal form as in \eqref{eq:SDP_definition} is \textit{$G$-invariant} if $C \in \operatorname{End}^G(\mathcal{H})$ and for every feasible solution $X\in \mathcal{P}(\mathcal{H})$ and for every $g \in G$, $U(g) X U^\dag(g)$ remains a feasible solution. By convexity, one can then find an optimal solution in $\text{End}^{G}(\mathcal{H})$. In fact, if $X$ is an optimal solution of \eqref{eq:SDP_definition}, so is its group average \eqref{eq:twirl_channel}. We write $G_1 \leq G_2$ to denote that $G_1$ is a subgroup of $G_2$. Note that\begin{equation}
    G_1 \leq G_2 \Longrightarrow \text{End}^{G_2}(\mathcal{H}) \subseteq \text{End}^{G_1}(\mathcal{H}).
\end{equation}
A quantum channel $\mathcal{N}_{A \to B}$ is called \textit{$G$-covariant} if there exist unitary representations $g \mapsto U_A(g)$ on $\mathcal{H}_A$ and
$g \mapsto V_B(g)$ on $\mathcal{H}_B$ such that:
\begin{equation}
\label{eq:G_covariant_channel}
    \mathcal{N}_{A \to B}(U_A(g)\, \rho_A\, {U_A(g)}^\dagger)
    = V_B(g)\, \mathcal{N}_{A \to B}(\rho_A)\, {V_B(g)}^\dagger
    \qquad \forall\, g \in G,\; \forall\, \rho_A \in \mathcal{D}(\mathcal{H}_A).
\end{equation}
Given an arbitrary channel $\mathcal{N}_{A \to B}$, its \textit{$G$-twirl} is defined as:
\begin{equation}
\label{eq:channel_twirl}
    \overline{\mathcal{N}}(\rho) \coloneqq \frac{1}{|G|}
    \sum_{g \in G} {V_B(g)}\, \mathcal{N}(
    {U_A(g)}^\dagger\, \rho\, U_A(g))\, {V_B(g)}^\dagger,
\end{equation} and it is $G$-covariant by construction. The Choi matrix of a $G-$covariant channel \eqref{eq:G_covariant_channel} is invariant under the action of $U_A^T(g)\otimes V_B^\dag (g)$, i.e.:\begin{equation}
\label{eq:choi_group_covariant}
    \Gamma_{AB}^\mathcal{N} =(U_A^T(g)\otimes V_B^\dag (g))  \Gamma_{AB}^\mathcal{N}(U_A^T(g)\otimes V_B^\dag (g))^\dag  \qquad \forall\, g \in G.
\end{equation}
A channel $\mathcal{N}_{A\to B}$ is called $G-$invariant on its input (resp. output) if it satisfies the condition $ \mathcal{N}_{A \to B}(\rho_A) = V_B(g)\, \mathcal{N}_{A \to B}(\rho_A)\, {V_B(g)}^\dagger$ (resp. $  \mathcal{N}_{A \to B}(U_A(g)\, \rho_A\, {U_A(g)}^\dagger)
    = \mathcal{N}_{A \to B}(\rho_A)$) for all $g \in G$ and $\rho_A \in \mathcal{D}(\mathcal{H}_A)$. Its Choi state satisfies an invariance of the form \eqref{eq:choi_group_covariant} where the identity replaces $U_A$ (resp. $V_B$).  By \eqref{eq:concatenation_choi}, one immediately has:
    \begin{itemize}
        \item If $\mathcal{N}_1$ is $G$-invariant on its input and $\mathcal{N}_2$ is $G$-covariant, then  $\mathcal{N}_2 \circ \mathcal{N}_1$ is $G$-invariant on its  input;
        \item If $\mathcal{N}_1$ is $G$-covariant and $\mathcal{N}_2$ is $G$-invariant on its output, then $\mathcal{N}_2 \circ \mathcal{N}_1$ is $G$-invariant on its output.
    \end{itemize}

The main group of interest in this work is the \textit{symmetric group} on $n$ symbols, denoted by $\mathfrak{S}_n$, which acts on $\mathcal{H}^{\otimes n}$ by permuting tensor factors:
\begin{equation}
\label{eq:symmetric_group_action}
\pi \cdot (\ket{v_1} \otimes \cdots \otimes \ket{v_n}) \coloneqq \ket{v_{\pi^{-1}(1)}} \otimes \cdots \otimes \ket{v_{\pi^{-1}(n)}} \qquad \forall\, \ket{v_i} \in \mathcal{H},\ \forall\, \pi \in \mathfrak{S}_n.
\end{equation}
We denote by $P_{\mathcal{H}^{\otimes n}}(\pi)$ the permutation matrix, i.e., the unitary representation of $\pi \in \mathfrak{S}_n$ on $\mathcal{H}^{\otimes n}$. A state $\rho \in \mathcal{D}(\mathcal{H}^{\otimes n})$ is called \textit{permutation-invariant} if:
\begin{equation}
    P(\pi)\, \rho\, P(\pi)^\dagger = \rho \qquad \forall\, \pi \in \mathfrak{S}_n.
\end{equation}
Note that permutation matrices are real and $P(\pi)^\dagger = P(\pi)^T$ for all $\pi \in \mathfrak{S}_n$. We denote the symmetrizer channel as $\mathcal{S}_n \equiv \mathcal{T}_{\mathfrak{S}_n}$, i.e. the twirling channel \eqref{eq:channel_twirl} for $G = \mathfrak{S}_n$.

For any $\rho$, we denote its symmetrization as\begin{equation}
    \label{eq:symmetrization}
    \bar{\rho} \coloneqq \mathcal{S}_n(\rho) = \frac{1}{n!} \sum_{\pi \in \mathfrak{S}_n} P(\pi)\, \rho\, P(\pi)^\dagger,
\end{equation}which is permutation-invariant by construction. The subspace of permutation-invariant operators is denoted as:
\begin{equation}
\label{eq:def_End_Sn}
    \mathrm{End}^{\mathfrak{S}_n}(\mathcal{H}^{\otimes n}) \equiv \esn(\mathcal{L}(\HS\n)) \coloneqq \bigl\{\rho \in \mathcal{L}(\mathcal{H}^{\otimes n}) : P(\pi)\, \rho\, P(\pi)^\dagger = \rho \quad \forall\, \pi \in \mathfrak{S}_n\bigr\}.
\end{equation}
The simplest example of a permutation-invariant state is a tensor product state $\rho^{\otimes n}$. We will often consider multipartite operators $\rho_{RS^n} \in \mathcal{L}(R \otimes S^{\otimes n})$ for some system $S$ and a fixed reference $R$. We say that $\rho_{RS^n}$ is symmetric with respect to the reference $R$ if it is invariant under permutation of the $S$ systems while keeping $R$ fixed:
\begin{equation}
\label{eq:perm_inv_reference_system}
    \rho_{RS^n} = (\mathbbm{1}_R \otimes P_{S^n}(\pi))\, \rho_{RS^n}\, (\mathbbm{1}_R \otimes P_{S^n}(\pi))^\dagger \qquad \forall\, \pi \in \mathfrak{S}_n.
\end{equation}
The corresponding invariant subspace is denoted as $\mathrm{End}^{\mathfrak{S}_n}(R \otimes S^{\otimes n})$. We also consider operators $\rho_{S_1^n S_2^n} \in \mathcal{L}(S_1^{\otimes n} \otimes S_2^{\otimes n})$ that are invariant under simultaneous local permutations:
\begin{equation}
\label{eq:perm_inv_two_systems}
    \rho_{A^n B^n} = (P_{A^n}(\pi) \otimes P_{B^n}(\pi))\, \rho_{A^n B^n}\, (P_{A^n}(\pi) \otimes P_{B^n}(\pi))^\dagger \qquad \forall\, \pi \in \mathfrak{S}_n,
\end{equation}
with invariant subspace $\mathrm{End}^{\mathfrak{S}_n}(A^{\otimes n} \otimes B^{\otimes n})$. A quantum channel $\mathcal{N}_n \in \CPTP(A^n \to B^n)$ is \textit{permutation-covariant} if:
\begin{equation}
    \mathcal{N}_n(P_{A^n}(\pi)\, \rho\, P_{A^n}(\pi)^\dagger) = P_{B^n}(\pi)\, \mathcal{N}_n(\rho)\, P_{B^n}(\pi)^\dagger \qquad \forall\, \pi \in \mathfrak{S}_n
\end{equation}
or more compactly, $\mathcal{N}_n = \mathcal{U}_{B^n}^\dagger(\pi) \circ \mathcal{N}_n \circ \mathcal{U}_{A^n}(\pi)$ for all $\pi$, where $\mathcal{U}(\pi)(\cdot) = P(\pi)(\cdot)P(\pi)^\dagger$. The simplest example is the $n$-fold tensor product $\mathcal{N}^{\otimes n}$. It follows from~\eqref{eq:choi_tensor_product} that the Choi matrix of a permutation-covariant channel satisfies \eqref{eq:perm_inv_two_systems}, i.e. $\Gamma^{\mathcal{N}_n}_{A^n B^n} \in \mathrm{End}^{\mathfrak{S}_n}(A^{\otimes n} \otimes B^{\otimes n})$. While later on in this work we will primarily focus on $n$-fold tensor product channels, we stress that the construction applies to any permutation-covariant (possibly non-i.i.d.) channel.
\

A channel $\mathcal{E}_n : R \to S^n$ (resp. $\mathcal{D}_n : S^n \to R$) is symmetric on its output (resp. input) if it is invariant under permutations of the $S$ systems:
\begin{align}
\label{eq:condition_symmetry_encoder}
\mathcal{E}_n(\rho_R) = P_{S^n}(\pi)\mathcal{E}_n(\rho_R)P_{S^n}(\pi)^\dagger, \\
\label{eq:condition_symmetry_decoder}
\mathcal{D}_n(P_{S^n}(\pi)\rho_{S^n}P_{S^n}(\pi)^\dagger) = \mathcal{D}_n(\rho_{S^n}),
\end{align}
for all $\pi\in\mathfrak{S}_n$. Equivalently, $\mathcal{E}_{R \to S^n}= \mathcal{U}_{S^n}(\pi) \circ \mathcal{E}_{R \to S^n} \ \ \forall \pi \in \mathfrak{S}_n$ and $\mathcal{D}_{S^n \to R}=  \mathcal{D}_{S^n \to R} \circ \mathcal{U}_{S^n}(\pi) \ \ \forall \pi \in \mathfrak{S}_n$, or in terms of the Choi representation $\Gamma^{\mathcal{E}_n}_{RS^n}\in\mathrm{End}^{\mathfrak{S}_n}(R\otimes S^{\otimes n}), \Gamma^{\mathcal{D}_n}_{S^nR}\in\mathrm{End}^{\mathfrak{S}_n}(S^{\otimes n}\otimes R)$ in the sense of \eqref{eq:perm_inv_reference_system}.

A bipartite state $\rho_{RS} \in \mathcal{D}(R \otimes S)$ is called \textit{$k$-extendible} with respect to $S$~\cite{Werner1989, Doherty2002, Doherty2004} if there exists a permutation-invariant extension $\sigma_{RS^k} \in \mathrm{End}^{\mathfrak{S}_k}(R \otimes S^{\otimes k})$ with $\sigma_{RS^k} \geq 0$ and $\mathrm{Tr}_{S^{k-1}}[\sigma_{RS^k}] = \rho_{RS}$. The set of such states is denoted $\mathrm{EXT}_k(R; S)$. These sets form a nested hierarchy $\mathrm{SEP}(R : S) \subset \cdots \subset \mathrm{EXT}_{k+1} \subset \mathrm{EXT}_k \subset \cdots \subset \mathrm{EXT}_2$, converging to the set of separable states as $k \to \infty$ by the quantum de Finetti theorem~\cite{Caves2002, Christandl2007}. A quantum channel $\mathcal{N}_{A \to B}$ is called \textit{$k$-extendible}~\cite{Pankowski2013, Kaur2019} if its Choi state $\Phi^{\mathcal{N}}_{AB}$ is $k$-extendible with respect to $B$; equivalently, for every input $\rho_{RA}$ the output $\omega_{RB} \equiv (\mathrm{id}_R \otimes \mathcal{N}_{A\to B})(\rho_{RA})$ is $k$-extendible.

If an SDP is $G$-invariant with $G = \mathfrak{S}_n$, we say that it is \textit{permutation-invariant}. This observation, combined with a suitable basis for $\mathrm{End}^{\mathfrak{S}_n}$, is the starting point for the symmetry reduction developed in Section~\ref{sec:perm-inv}. 
\subsection{Maximal Singlet Fraction}
\label{subsec:maximal_singlet_fraction}

One crucial example of an SDP that we study in this work is the \textit{maximal singlet fraction}\footnote{This quantity was named \textit{quantum correlation} in \cite{Koenig2009} and  \textit{one-sided optimized singlet fidelity} in \cite{Skrzypczyk2023}.}, originally studied in \cite{Horodecki1999} and reviewed in detail e.g. in \cite[Chapter 7.4]{Skrzypczyk2023}, which measures the closeness of a bipartite state to a MES.
\begin{definition}
   Let $\Phi_{A\hat{A}}$ be a MES of dimension $d_A$, and let $\rho_{AB}\in \mathcal{D}(\mathcal{H}_{AB})$. The maximal singlet fraction of $\rho_{AB}$ is defined as:
\begin{equation}
\label{eq:maximal_singlet_fraction}
F_{\Phi}^{(B)}(\rho_{AB}) \coloneqq \max_{\mathcal{N}_{B \to \hat{A}}\in \CPTP(B\to \hat{A})} F\bigl(\Phi_{A\hat{A}}, (\mathrm{id}_{A} \otimes \mathcal{N}_{B \to \hat{A}})(\rho_{AB})\bigr).
\end{equation}
Note that in general we allow $d_{B} \neq d_A$, and the superscript $B$ in \eqref{eq:maximal_singlet_fraction} indicates that we maximize over all local channels on system $B$.
The maximal singlet fraction admits an SDP formulation in terms of the Choi 'state' $\Phi^{\mathcal{\mathcal{N}^*}}$ of the adjoint of the optimizing channel (which need not be normalized, as $\mathcal{N}^*$ need not be TP):\begin{equation}
\label{eq:SDP_Maximal_Singlet_Fraction}
\boxed{
\begin{aligned}
\textbf{Primal:}\quad & \\
\max_{\Phi^{\mathcal{\mathcal{N}^*}}_{AB}} \quad & \Tr\!\big(\rho_{AB}\, \Phi^{\mathcal{\mathcal{N}^*}}_{AB}\big) \\
\textrm{s.t.} \quad & \Tr_{A}[\Phi^{\mathcal{\mathcal{N}^*}}_{AB}] = \frac{\mathbbm{1}_{B}}{d_A} \\
& \Phi^{\mathcal{\mathcal{N}^*}}_{AB} \geq 0
\end{aligned}
\qquad\qquad
\begin{aligned}
\textbf{Dual:}\quad & \\
\min_{Y_{B}} \quad & \frac{1}{d_A}\,\Tr(Y_{B}) \\
\textrm{s.t.} \quad & \rho_{AB} \leq \mathbbm{1}_{A} \otimes Y_{B} \\
&
\end{aligned}
}
\end{equation}
\end{definition}
The situation is depicted in Figure \ref{fig:maximal_singlet_fraction_fixed}.
\begin{figure}
    \centering
\begin{tikzpicture}[
    every node/.style={font=\Large},
    block/.style={draw, line width=0.7pt, minimum width=2.5cm, minimum height=1.2cm, align=center},
    arrow/.style={line width=0.7pt, -{Stealth[scale=1.0]}},
    line/.style={line width=0.7pt}
]

\filldraw[black] (-1.5,1.2) circle (1.3pt) node[left] {$\rho_{AB}$};

\draw[line] (-1.5,1.2) -- (0,2.2);
\draw[line] (-1.5,1.2) -- (0,0.2);

\draw[line] (0,2.2) -- (5.5,2.2) node[midway, above] {$A$};

\draw[line] (0,0.2) -- (1.5,0.2) node[midway, above] {$B$};

\node[block] (decoder) at (3, 0.2) {$\mathcal{N}_{B \to \hat{A}}$};

\draw[arrow] (1.5,0.2) -- (decoder.west);

\draw[line] (decoder.east) -- (5.5,0.2) node[midway, above] {$\hat{A}$};

\coordinate (merge) at (7,1.2);
\draw[line] (5.5,2.2) -- (merge);
\draw[line] (5.5,0.2) -- (merge);

\node[right] at (merge) {$\approx \Phi^{d_A}_{A \hat{A}}$};

\end{tikzpicture}
  \caption{Maximal singlet fraction: the bipartite state $\rho_{AB}$ is processed locally on $B$ via $\mathcal{N}_{B \to \hat{A}}$ and compared to the maximally entangled state, using the fidelity in \eqref{eq:maximal_singlet_fraction}.}
\label{fig:maximal_singlet_fraction_fixed}
\end{figure}
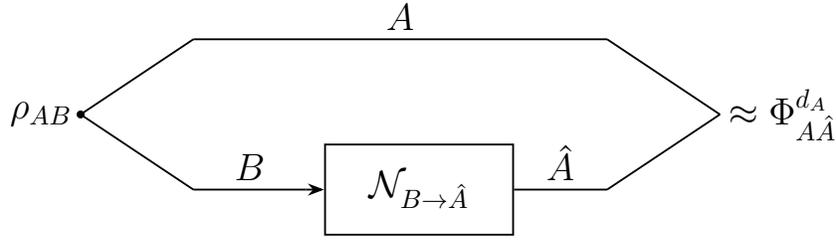 

The maximal singlet fraction finds important applications in entanglement theory and satisfies many useful properties (see \cite[Section 4.1.2]{Parentin2025Thesis} or \cite[Chapter 7.4]{Skrzypczyk2023} for a review). In particular, it satisfies strong duality, and the dual problem of \eqref{eq:SDP_Maximal_Singlet_Fraction} can be interpreted as a \textit{generalized robustness} \cite{VidalTarrach1999, Steiner2003}
 \begin{equation}
 \label{eq:maximal_singlet_fraction_as_generalized_roboustness}
 \begin{aligned}
F_{\Phi}^{(B)}(\rho) = \quad &  \min_{r} \frac{1+r}{d^2}  \\
& \text{subject to:}   \ \ \frac{\rho_{AB} + r\omega_{AB}}{1+r} = \pi_{A} \otimes \sigma_{B}\\
 & \ \ \ \ \ \ \ \ \ \ \ \ \ \ \ \omega_{AB} \in \mathcal{D}(\mathcal{H}_{AB}), \sigma_{B}\in \mathcal{D}(\mathcal{H}_{B}),
 \end{aligned}
\end{equation}
where $\omega_{AB}$ can be seen as the noise which must be added to the state $\rho_{AB}$ in order to bring it to a form where it is a product state such that the $A$ system is maximally mixed.
Moreover, \cite{Koenig2009} showed that this operational quantity is related to the \textit{conditional min-entropy} \cite{Renner2005, Datta2009}:\begin{align}
\label{eq:conditionalMinEntropy}
    H_{\min}(A|B)_{\rho_{AB}} &\coloneqq -\inf_{\sigma_{B} \in \mathcal{D}(B)}\log_2 \inf_{\lambda \geq 0}\{\lambda: \rho_{A B} \leq \lambda \mathbbm{1}_{A}\otimes \sigma_{B} \}) \\
    \label{eq:conditional_min_entropy_from_singlet_fraction}
    &= -\log\left(d_A \cdot F^{(B)}_{\Phi}(\rho_{AB})\right).
\end{align}

The following lemma shows that the maximal singlet fraction of a $G$-invariant state is always achieved by a local recovery channel that is $G-$invariant on its input.

\begin{lemma}
\label{lem:covariant_singlet_fraction}
Let $G$ be a finite group with unitary representation $g \mapsto V(g)$ on $B$, and let $\rho_{AB} \in \mathcal{D}(AB)$ be
invariant under the action of $G$ on $B$, i.e.:
\begin{equation}
    (\mathbbm{1}_{A} \otimes V_B(g))\, \rho_{AB}\,
    (\mathbbm{1}_{A} \otimes V_B(g))^\dagger = \rho_{AB}
    \qquad \forall\, g \in G.
\end{equation}
Then the maximal singlet fraction~\eqref{eq:maximal_singlet_fraction} is always achieved by a channel $\mathcal{N}_{B \to \hat{A}}$ that is $G$-invariant on its input:
\begin{equation}
    \mathcal{N}_{B \to \hat{A}}(V_B(g)\, \sigma_B\, V_B(g)^\dagger)
    = \mathcal{N}_{B \to \hat{A}}(\sigma_B)
    \qquad \forall\, g \in G,\; \forall\, \sigma \in \mathcal{D}(B).
\end{equation}
\end{lemma}

\begin{proof}
Let $\mathcal{N}_{B \to \hat{A}}$ be any feasible channel for the
SDP~\eqref{eq:SDP_Maximal_Singlet_Fraction}, and define its $G$-input twirl as per \eqref{eq:twirl_channel}:
\begin{equation}
    \overline{\mathcal{N}}_{B \to \hat{A}}(\sigma_B) \coloneqq \frac{1}{|G|}
    \sum_{g \in G} \mathcal{N}(V_B(g)^\dagger\, \sigma_B\, V_B(g)).
\end{equation}
Combining \eqref{eq:choi_group_covariant} and \eqref{eq:adjoint_choi}, the Choi state of $ \overline{\mathcal{N}}^*_{\hat{A} \to B}$ satisfies for all $g \in G$:\begin{equation}
\label{eq:G-twirl-choi-adj}
    \Phi^{\overline{\mathcal{N}}^*}_{AB} = (\mathbbm{1}_{A} \otimes V_B(g))      \Phi^{\overline{\mathcal{N}}^*}_{AB}   (\mathbbm{1}_{A} \otimes V^\dag_B(g))  
\end{equation}
The operator in \eqref{eq:G-twirl-choi-adj} satisfies both constraints in \eqref{eq:SDP_Maximal_Singlet_Fraction}. Indeed, $\Phi^{\overline{\mathcal{N}}^*}_{AB}  \geq 0$ (convex combination of PSD operators) and:\begin{equation}
    \mathrm{Tr}_{A}[\overline{\Phi^{\mathcal{N}^*}_{AB}}] =
\frac{\mathbbm{1}_B}{d_A} \frac{1}{|G|} \sum_g V^\dag(g) V(g) =
\frac{\mathbbm{1}_B}{d_A}.
\end{equation}
Moreover:
\begin{align}
    \mathrm{Tr}\!\left[\rho_{AB}\,
    \overline{\Phi^{\mathcal{N}^*}_{AB}}\right]
    &= \frac{1}{|G|} \sum_{g \in G}
    \mathrm{Tr}\!\left[\rho_{AB}\,
    (\mathbbm{1}_{A} \otimes V_B(g))\,
    \Phi^{\mathcal{N}^*}_{AB}\,
    (\mathbbm{1}_{A} \otimes V_B(g))^\dagger\right] \\
    &= \Tr[\overline{\rho_{AB}}   \Phi^{\mathcal{N}^*}_{AB}] \\
    &= \mathrm{Tr}\!\left[\rho_{AB}\,
    \Phi^{\mathcal{N}^*}_{AB} \right]
\end{align}
where the last step follows from the $G-$invariance of $\rho_{AB}$. Hence, $\overline{\mathcal{N}}$ achieves the same objective as $\mathcal{N}$ and is $G$-invariant on its input by construction.
\end{proof}
This lemma in particular applies to the case where the second system is of the form $B^n$ for some $n \in \mathbb{N}$ and $\rho_{AB^n} \in \mathrm{End}^{\mathfrak{S}_n}(A \otimes B^{\otimes n})$: in this case, the map $\mathcal{N}_{B^n \to \hat{A}}$ can be restricted to be symmetric on its input as in \eqref{eq:condition_symmetry_decoder}.
\medskip

In the following, we will also make use of the maximal singlet fraction with optimization over the first system $A$, defined symmetrically to~\eqref{eq:maximal_singlet_fraction}:
\begin{equation}
\label{eq:maximal_singlet_fraction_A}
    F^{(A)}_{\Phi}(\rho_{AB}) \coloneqq \max_{\mathcal{N}_{A \to \hat{B}}\in \mathrm{CPTP}(A\to \hat{B})} F\bigl(\Phi_{B\hat{B}},\, (\mathcal{N}_{A \to \hat{B}} \otimes \mathrm{id}_{B})(\rho_{AB})\bigr).
\end{equation}

This quantity admits an analogous SDP formulation to \eqref{eq:SDP_Maximal_Singlet_Fraction}, with the roles of systems $A$ and $B$ swapped, and shares analogous properties to it. For instance, Lemma~\ref{lem:covariant_singlet_fraction} holds verbatim with $A$ and $B$ exchanged, and one has:\begin{align}
\label{eq:conditionalMinEntropy_BA}
    H_{\min}(B|A)_{\rho_{AB}} &= -\log\left(d_B \cdot F^{(A)}_{\Phi}(\rho_{AB})\right).
\end{align}

\section{Permutation Invariance: Efficient Operations and Isomorphism}
\label{sec:perm-inv}

\subsection{Permutation-Invariant Operators}
\label{sec:orbit_basis_operations}
Often in quantum information theory one considers independent and identically-distributed (i.i.d.) copies of a quantum state $\rho_A$ or a quantum channel $\mathcal{N}_{A \to B}^{\otimes n}$, which by definition satisfy:\begin{equation}
    \rho^{\otimes n}_{A^n} \in \esn(A\n) \qquad \Gamma^{\mathcal{N}^{\otimes n}}_{A^n B^n} \in \esn(A\n \otimes B\n).
\end{equation}
This permutation invariance can be exploited to avoid the exponential scaling of SDPs, because $\esn(\HS\n)\subset \mathcal{L}(\HS\n)$ is a subspace of \textit{polynomial dimension} in $n$, scaling as $O((n+1)^{d_\HS^2})$ instead of $O(d_\HS^{2n})$ (see \eqref{eq:dim_space_perm_inv}). In order to exploit this, we require a framework for dealing with permutation-invariant operators and efficiently phrase optimization problems in $\esn(\HS\n)$. To do so, the idea is to construct a suitable basis of $\esn(\HS\n)$ and rephrase SDPs as ones in which we minimize over now $O(\text{poly}(n))$ many basis coefficients, avoiding to ever construct matrices of exponential size in $n$.
\subsubsection{Orbit Basis and Count Matrices}
The space $\mathrm{End}^{\mathfrak{S}_n}(\mathcal{H}^{\otimes n})$ defined in \eqref{eq:def_End_Sn} admits a natural basis indexed by orbits of $\mathfrak{S}_n$ acting diagonally on pairs of multi-indices. Formally, we define the \textit{group orbit} of a pair of multi-indices $(\underline{i} \in [d_\HS]^n,\underline{j} \in [d_\HS]^n)$ under the action of $\mathfrak{S}_n$ as\begin{equation}
    O(\underline{i},\underline{j}) \coloneqq \{(\pi(\underline{i}),\pi(\underline{j})): \pi \in \mathfrak{S}_n\}.
\end{equation} 
We label these orbits as $O_r^{\mathcal{H}}$, for $r = 1, \ldots, m_{\mathcal{H}}$, where we denoted $m_{\mathcal{H}} \coloneqq \dim \mathrm{End}^{\mathfrak{S}_n}(\mathcal{H}^{\otimes n})$. Given $A \in \mathrm{End}^{\mathfrak{S}_n}(\mathcal{H}^{\otimes n})$, one clearly has by permutation invariance $ (A)_{(\underline{i},\underline{j})} =  (A)_{(\pi(\underline{i}),\pi(\underline{j}))}$ for all $\pi \in \mathfrak{S}_n.$
We define the orbit matrix as the following incidence matrix:
\begin{equation}
    (C_r^{\mathcal{H}})_{(\underline{i},\underline{j})} \coloneqq \begin{cases}
        1 & \text{if} \ \ {(\underline{i},\underline{j})} \in O_r^{\mathcal{H}}\\
        0 & \text{otherwise}\, .
    \end{cases}
\end{equation}
The \textit{orbit basis} is the set of orbit matrices $\{C_r^{\mathcal{H}}\}_{r=1}^{m_{\mathcal{H}}}$, and they form an orthogonal basis of $\mathrm{End}^{\mathfrak{S}_n}(\mathcal{H}^{\otimes n})$ with respect to the Hilbert--Schmidt inner product (this orthogonality is very easy to see, and one also easily sees that the number of orbits matches the dimension of $\esn(\HS\n)$). Every orbit $r$ can be identified by a representative multi-index pair $(\underline{i_r}, \underline{j_r})$, which can be efficiently computed using the one-to-one correspondence between orbits and count matrices: A \textit{count matrix} is a matrix of natural numbers $E \in \mathbb{N}^{d_{\mathcal{H}} \times d_{\mathcal{H}}}$ satisfying\begin{equation}
\label{eq:one_to_one_correspondence}
    \sum_{a,b =1}^{d_{\mathcal{H}}} E_{(a,b)} = n.
\end{equation}
For a given multi-index pair $(\underline{i},\underline{j})$, we can associate a corresponding count matrix, by counting the number of times the symbols $(a,b)$ appear in the multi-index pair $(\underline{i}, \underline{j})$, i.e.:
\begin{equation}
\label{eq:count_matrix}
    (E^{(\underline{i},\underline{j})})_{(a,b)} = |\{k \in \{1,\ldots, n\}: i_k = a,\, j_k = b\}| \qquad a,b \in \{1,\ldots, d_{\mathcal{H}}\}.
\end{equation}
By this construction, for every pair $(\underline{i},\underline{j})$, $(\underline{i}',\underline{j}') \in \{1,..., d_{\mathcal{H}}\}^n\times\{1,..., d_{\mathcal{H}}\}^n$, we have:\begin{equation}
    (\underline{i}',\underline{j}') = (\pi(\underline{i}),\pi(\underline{j})) \ \ \ \ \text{for some} \ \pi \in \mathfrak{S}_n  \iff E^{(\underline{i},\underline{j})} = E^{ (\underline{i}',\underline{j}')}\,.
\end{equation}
Additionally, given a count matrix $E$ we can always construct a representative pair of multi-indices $(\underline{i}, \underline{j})$ such that $E = E^{(\underline{i}, \underline{j})}$ by just going through all pairs $(a,b)$ and adding $E_{(a,b)}$ copies $(a,b)$ to the string. 
Therefore, each orbit $O_r^{\mathcal{H}}$ is uniquely associated to a non-negative integer solution of~\eqref{eq:one_to_one_correspondence}. Thanks to this correspondence, we will sometimes use the notation $C_E$ to indicate the matrix with element:
\begin{equation}
\label{eq:orbit_count_relation}
    (C_E)_{(\underline{i}, \underline{j})} \coloneqq \begin{cases}
        1 &\text{if} \ \ E^{(\underline{i}, \underline{j})} = E\\
        0 &\text{otherwise}\, .
    \end{cases}
\end{equation} 
Since two pairs of symbols are in the same orbit if and only if they give the same count matrix $E_{(a,b)}$, orbits correspond one-to-one with \textit{weak compositions of $n$ into $d^2$}, one part for each pair $(a,b)$, whose number is:
\begin{equation}
\label{eq:dim_space_perm_inv}
    m_{\mathcal{H}} = \binom{n + d_{\mathcal{H}}^2 -1}{d_{\mathcal{H}}^2 -1}\leq (n+1)^{d_\mathcal{H}^2}.
\end{equation} Given an operator $X\in \mathcal{L}(\mathcal{H})$, the coefficients of its tensor product $X^{\otimes n} \in \mathrm{End}^{\mathfrak{S}_n}(\mathcal{H}^{\otimes n})$ with respect to the orbit basis can be computed straightforwardly from $X$ by picking a representative of orbit $r$, call it $(\underline{i_r}, \underline{j_r})$, and computing the product of the different $n$ elements of the original matrix in the given position:
\begin{equation}
\label{eq:coeff_tensor_product}
    c_r = \prod_{k=1}^n X_{(i_r)_k, (j_r)_k}.
\end{equation}
such that then
\begin{equation}
\label{eq:tensor_product_orbit}
    X^{\otimes n} = \sum_{r = 1}^{m_{\mathcal{H}}}c_r\, C_r^{\mathcal{H}}.
\end{equation}
Since only \textit{one} representative of each orbit is required, these coefficients are computable in $O(\mathrm{poly}(n))$ time. 

\begin{lemma}
\label{lem:support_and_counts}
Let $X \in \mathcal{L}(\mathcal{H})$ and define $\mathrm{nz}(X) \coloneqq \{(a,b) \in [d_{\mathcal{H}}]^2 : X_{(a,b)} \neq 0\}$. Then, if $X^{\otimes n}$ is expressed as in \eqref{eq:tensor_product_orbit}, one has\begin{equation}
\label{eq:condition_sparsity}
    c_r \neq 0 \iff (E_r)_{(a,b)} = 0 \quad \forall (a,b) \notin \mathrm{nz}(X).
\end{equation}
    The number of such orbits is:
\begin{equation}
\label{eq:sparse_orbits}
 \binom{n + |\mathrm{nz}(X)| - 1}{n}
    = O(n^{|\mathrm{nz}(X)| - 1}).
\end{equation}
\end{lemma}
\begin{proof}
By~\eqref{eq:coeff_tensor_product}, the coefficient of orbit $r$ with representative $(\underline{i_r}, \underline{j_r})$ is $c_r = \prod_{k=1}^n X_{(i_r)_k, (j_r)_k}$. This product vanishes if any factor $X_{(i_r)_k, (j_r)_k} = 0$, i.e.\ if any position $k$ carries a pair type $(a,b) \notin \mathrm{nz}(X)$, which is equivalent to \eqref{eq:condition_sparsity}. The number of valid count matrices is then equal to the number of ways to distribute $n$ among $|\mathrm{nz}(X)|$ allowed pair types, giving~\eqref{eq:sparse_orbits}.
\end{proof} 

Two fundamental operations on orbit matrices --- the trace and the Hilbert--Schmidt inner product --- can be computed directly from the count matrices, without constructing the (exponentially large) corresponding orbit matrices.

\begin{lemma}
\label{lemma:trace_HS_orbit}
Let $\{C_r^{\mathcal{H}}\}_{r=1}^{m_{\mathcal{H}}}$ be the orbit basis of $\mathrm{End}^{\mathfrak{S}_n}(\mathcal{H}^{\otimes n})$, with corresponding count matrices $\{E_r \in \mathbb{N}^{d_{\mathcal{H}} \times d_{\mathcal{H}}}\}_r$. Then:
\begin{enumerate}
    \item \textbf{Orthogonality.} The orbit matrices are pairwise orthogonal:
    \begin{equation}
        \mathrm{Tr}[(C_r^{\mathcal{H}})^\dagger C_{r'}^{\mathcal{H}}] = \delta_{r,r'}\, |O_r|,
    \end{equation}
    where the squared Hilbert--Schmidt norm equals the orbit size:
    \begin{equation}
    \label{eq:inner_product_HS}
        \|C_r^{\mathcal{H}}\|_{\mathrm{HS}}^2 = |O_r| = \binom{n}{(E_r)_{(1,1)},\, (E_r)_{(1,2)},\, \ldots,\, (E_r)_{(d_{\mathcal{H}}, d_{\mathcal{H}})}}.
    \end{equation}
    \item \textbf{Trace.} The trace of an orbit matrix is nonzero only for diagonal count matrices:
    \begin{equation}
    \label{eq:formula_trace_orbit_matrices}
        \mathrm{Tr}[C_r^{\mathcal{H}}] = \begin{cases} \displaystyle\frac{n!}{\prod_{a=1}^{d_{\mathcal{H}}} m_a!} & \text{if } E_r \text{ is diagonal, with } m_a \equiv (E_r)_{aa} \\[6pt] 0 & \text{otherwise} \, . \end{cases}
    \end{equation}

\end{enumerate}
All quantities are computable in $O(m_{\mathcal{H}} \cdot d_{\mathcal{H}}^2)$ time from the coefficients and count matrices alone.
\end{lemma}

\begin{proof}
\begin{enumerate}
    \item 
Since $(C_r^{\mathcal{H}})_{\underline{i},\underline{j}} \in \{0,1\}$ with value $1$ iff $(\underline{i},\underline{j}) \in O_r^{\mathcal{H}}$, we have:
\begin{equation*}
    \|C_r^{\mathcal{H}}\|_{\mathrm{HS}}^2 = \sum_{\underline{i},\underline{j}} |(C_r^{\mathcal{H}})_{(\underline{i},\underline{j})}|^2 = |O_r|.
\end{equation*}
The orbit size counts the distinct permutations of $n$ positions among $d_{\mathcal{H}}^2$ pair types with multiplicities $(E_r)_{ab}$, and is hence given by the multinomial coefficient stated above. For $r \neq r'$, the two orbits $O_r$ and $O_{r'}$ are disjoint (since all orbits partition the set of all multi-index pairs), and so $(C_r^{\mathcal{H}})^\dagger C_{r'}^{\mathcal{H}}$ has all zero diagonal entries, hence zero trace.

\item A diagonal entry $(\underline{i},\underline{i})$ belongs to orbit $r$ iff $E^{(\underline{i},\underline{i})} = E_r$. Since $(E^{(\underline{i},\underline{i})})_{ab} = \delta_{ab}\, |\{k : i_k = a\}|$, the count matrix of any diagonal entry is always diagonal. Thus $\mathrm{Tr}[C_r^{\mathcal{H}}] = 0$ unless $E_r$ is diagonal. When it is, $\mathrm{Tr}[C_r^{\mathcal{H}}]$ counts the distinct multi-indices $\underline{i}$ with $m_a = (E_r)_{(a,a)}$ occurrences of symbol $a$, giving the multinomial $n!/\prod_a m_a!$.
\end{enumerate}
\end{proof}
As a direct consequence of the above lemma and of linearity of the trace, we also have that for all $X_{S^n},Y_{S^n} \in \mathrm{End}^{\mathfrak{S}_n}(S^{\otimes n})$, written as $X_{S^n} = \sum_{r\in [m_S]} x_r\, C_r^{S}$ and $Y_{S^n} = \sum_{r\in [m_S]} y_r\, C_r^{S}$:
    \begin{align}
               \mathrm{Tr}[X_{S^n}] &= \sum_{r=1}^{m_{S}} x_r\, \mathrm{Tr}[C_r^{S}],\\
        \mathrm{Tr}[X_{S^n}^\dagger Y_{S^n}] &= \sum_{r=1}^{m_{S}} \overline{x_r}\, y_r\, |O_r|.
    \end{align} Given $X_{S^n}\in  \mathrm{End}^{\mathfrak{S}_n}(S^{\otimes n})$, we'll often be interested in computing its full transpose $X^T_{S^n}$ (when seen as an element of $\mathcal{L}(\HS_S\n)$), which is also in $\mathrm{End}^{\mathfrak{S}_n}(S^{\otimes n})$. To do so, we rely on the following simple result.
\begin{lemma}
\label{lem:transpose_orbit_basis}
Let $C_r$ be an orbit matrix corresponding to count matrix $E_r$. Then, its transpose $C_r^T$ is the orbit matrix corresponding to transposed count matrix of $(E_r)^T$, i.e., using~\eqref{eq:orbit_count_relation}:
\begin{equation}
\label{eq:count-orbit-transpose-relation}
C_r^T =C_{r^T}  \coloneqq C_{E_r^T}.
\end{equation}
\end{lemma}
\begin{proof}
Let $C_r \coloneqq C_E$ where $E \coloneqq E^{(\underline{i_r}, \underline{j_r})}$ for some representative multi-index pair $(\underline{i_r}, \underline{j_r})$. By definition of the transpose, $(C_r^T)_{(\underline{i},\underline{j})} = (C_r)_{(\underline{j},\underline{i})}, \quad \forall \underline{i}, \underline{j} \in [d_S]^{\times n}$.
The new representative multi-index pair is obtained from $(\underline{i_r}, \underline{j_r})$ by swapping the roles of $a$ and $b$ in~\eqref{eq:count_matrix}. Hence, it generates an orbit whose count matrix is precisely the transpose of $E$.
\end{proof}
By linearity of the transpose, we conclude that for $X_{S^n} = \sum_{r\in [m_S]} x_r\, C_r^{S}$:
\begin{align}
(X_{S^n})^T &=  \sum_{r=1}^{m_S} x_{r}\, C_{r^T}^S = \sum_{r=1}^{m_S} x'_{r}\,  C_r^S,
\end{align}
where $x'_{r} \coloneqq x_{r^T}$ for all $r \in [m_S]$. In other words, transposition in the orbit basis corresponds to a simple reshuffle of the orbit basis coefficients $\{x_r\}_{r \in [m_S]}$. If $X_{S^n}^\dag = X_{S^n},$ (as it is often the case), then $x'_{r} = \overline{x}_r$ for all $r \in [m_s]$.

In many settings, it is also interesting to consider partial transposes and partial traces of bipartite operators of the form $X_{A^{n} B^{n}}\in \mathrm{End}^{\mathfrak{S}_n}(A^{\otimes n}\otimes B^{\otimes n})$, i.e. permutation-invariant with respect to joint permutations of systems $A$ and $B$.
In this case, every canonical basis vector $\ket{i} \in \HS_A \otimes \HS_B$ splits as $\ket{i} = \ket{i_A i_B}$, and hence also every multi-index pair $(\underline{i_s}, \underline{j_s})$ splits into two pairs of pairs $(\underline{i^{(A)}_s}\, \underline{i^{(B)}_s}, \underline{j^{(A)}_s}\,\underline{j^{(B)}_s})$. The pairing $(\underline{i^{(A)}_s}, \underline{j^{(A)}_s})$ then corresponds to an orbit $O(\underline{i^{(A)}_s}, \underline{j^{(A)}_s}) = O^A_r$ in $\esn(\HS_A\n)$ for some $r \in [m_A]$, and we will also write $r \equiv r(s)$ to emphasize that this is a \enquote{part} of the original joint orbit $s$. Analogously also the orbit $O(\underline{i^{(B)}_s}, \underline{j^{(B)}_s}) = O^B_t$ corresponds to an orbit $t(s) = t \in [m_B]$ in $\esn(\HS_B\n)$. In this way, the count matrix $E_{AB}^{(\underline{i}_{s}, \underline{j}_{s})} \in \mathbb{N}^{(d_A d_B) \times (d_A d_B)}$ can be thought of as a bipartite matrix with generic element $\left(E_{AB}^{(\underline{i}_{s}, \underline{j}_{s})}\right)_{(a_A, a_B), (b_A, b_B)}$, for $a_A, b_A \in [d_A]$ and $a_B, b_B\in [d_B]$, and the corresponding marginal count matrices on $A$ and $B$ are\begin{align}
 \label{eq:A_marginal}
    (E_A)_{(a_A, b_A)}  &\coloneqq (E_{r(s)})_{(a_A, b_A)} = \sum_{a_B, b_B = 1}^{d_B} (E_{AB})_{(a_A,a_B),(b_A,b_B)},\\
    (E_B)_{(a_B, b_B)}  &\coloneqq (E_{t(s)})_{(a_B, b_B)} = \sum_{a_A, b_A = 1}^{d_A} (E_{AB})_{(a_A,a_B),(b_A,b_B)}.
     \label{eq:B_marginal}
\end{align}
Then, following the same proof as in Lemma \ref{lem:transpose_orbit_basis}, the partial transpose of the corresponding orbit matrix with respect to system $B$ can be computed as\begin{equation}
\label{eq:count-orbit-partial-transpose-relation}
    C_s^{T_B} = C_{s^{T_B}}  \coloneqq C_{E^{T_B}},
\end{equation}
where $E^{T_B}_{AB}$ denotes the partial-transposed count matrix. An analogous statement holds for $T_{A}$. This allows one to efficiently compute partial transposes of operators $X_{A^nB^n} = \sum_{r\in [m_{AB}]} x_r\, C_r$ via \begin{equation}\label{eq:coefficient_partial_transpose_relation}
    X^{T_{B}}_{A^n B^n} = \sum_{s=1}^{m_{AB}} x'_s\, C_s,
\end{equation}
where $x'_s \coloneqq x_{s^{T_{B}}}$ for all $s \in [m_{AB}]$.

Besides partial transposes we can also compute the partial traces $\Tr_{B^n}[C_s]$ and $\Tr_{A^n}[C_s]$ efficiently, as shown in the following Lemma.

\begin{lemma}
\label{lemma:partial_trace_orbit}
Let $s \in [m_{AB}]$ be a joint orbit with count matrix $E_s$, and let $r(s) \in [m_A]$ denote its marginal orbit on $A$, with count matrix 
$E_{r(s)}$ defined in~\eqref{eq:A_marginal}. Similarly, let $t(s) \in [m_B]$ denote its marginal orbit on $B$, with count matrix $E_{t(s)}$ defined in \eqref{eq:B_marginal}. Then
\begin{equation}
\label{eq:partial_trace_B_orbit}
    \Tr_{B^n}[C_s^{AB}] = \tau_{s}   \kappa_s^A \, C_{r(s)}^A,
\end{equation}
where
\begin{align}
\label{eq:kappa_A}
    \kappa_s^A &\coloneqq 
    \prod_{(a_A,b_A) \in [d_A]^2}
    \binom{(E_{r(s)})_{(a_A,b_A)}}{
        (E_s)_{(a_A,1),(b_A,1)},\,
        (E_s)_{(a_A,1),(b_A,2)},\,
        \ldots,\,
        (E_s)_{(a_A,d_B),(b_A,d_B)}
    },\\
    \label{eq:tau_r_s}
    \tau_{s} &\coloneqq \begin{cases}1  &\text{if} \   E_{t(s)} \text{diagonal}\\
    0   &\text{otherwise},
    \end{cases}
\end{align} 
\end{lemma}

\begin{proof}
Since $(C_s^{AB})_{\underline{i},\underline{j}} \in \{0,1\}$ with value $1$ iff $(\underline{i},\underline{j}) \in O_s^{AB}$, we have
\begin{equation}
\label{eq:partial_trace_B_derivation}
    \bigl(\Tr_{B^n}[C_s^{AB}]\bigr)_{(\underline{i_A},\underline{j_A})}
    =
    \sum_{\underline{k_B} \in [d_B]^n}
    (C_s^{AB})_{(\underline{i_A}\,\underline{k_B},\; \underline{j_A}\,\underline{k_B})}
    =
    \bigl|
    \{ \underline{k_B} :
    E^{(\underline{i_A}\,\underline{k_B},\; \underline{j_A}\,\underline{k_B})}_{AB} = E_s
    \}
    \bigr|.
\end{equation}

We distinguish two cases.

\smallskip
\noindent
\textit{Non-diagonal case.}
If there exist $a_B \neq b_B$ such that $(E_s)_{(a_A,a_B),(b_A,b_B)} \neq 0$, or equivalently $E_{t(s)}$ defined in \eqref{eq:B_marginal} is non-diagonal, then the constraint in~\eqref{eq:partial_trace_B_derivation} cannot be satisfied, since in $(\underline{i_A}\,\underline{k_B},\underline{j_A}\,\underline{k_B})$ the $B$-indices coincide. Hence the set is empty and
\[
\bigl(\Tr_{B^n}[C_s^{AB}]\bigr)_{(\underline{i_A},\underline{j_A})} = 0
\]
for all $(\underline{i_A},\underline{j_A})$, so $\Tr_{B^n}[C_s^{AB}] = 0$.

\smallskip
\noindent
\textit{Diagonal case.} Assume now that $E_{t(s)}$ is diagonal. Fix $(\underline{i_A},\underline{j_A})$. A term in~\eqref{eq:partial_trace_B_derivation} is nonzero only if $(\underline{i_A},\underline{j_A}) \in O_{r(s)}^A$, i.e.\ $E^{(\underline{i_A},\underline{j_A})} = E_{r(s)}$. The multiplicity is then given by the number of strings $\underline{k_B}$ such that when paired up with the strings $\underline{i_A}$ and $\underline{j_A}$, the joint orbit of $(\underline{i_A}\,\underline{k_B}, \underline{j_A}\,\underline{k_B})$ is equal to the orbit $s$ (i.e.\ has count matrix $E_s$). All such $\underline{k_B}$ are thus permutations of each other, but not all permutations are allowed, since we can only take such permutations of $\underline{k_B}$ that -- while keeping $\underline{i_A}$ and $\underline{j_B}$ constant -- change the joint string $(\underline{i_A}\,\underline{k_B}, \underline{j_A}\,\underline{k_B})$ only up to a permutation. These are exactly the permutations within sets where $(\underline{i_A}, \underline{j_A})$ is constant.
Precisely, for each pair of values $(a_A, b_A)$ appearing $(E_{r(s)})_{(a_A,b_A)}$ times, for \eqref{eq:partial_trace_B_derivation} to hold the entries of $\underline{k_B}$ must be distributed among values $b_B \in [d_B]$ with multiplicities $(E_s)_{(a_A,b_B),(b_A,b_B)}$. Looking only at positions in the string where $(\underline{i_A}, \underline{j_A})$ takes values $(a_A, b_A)\in [d_A]^2$, the number of choices for $\underline{k_B}$ at these positions is thus given by the multinomial coefficient \begin{equation}
    \binom{(E_{r(s)})_{(a_A,b_A)}}{
        (E_s)_{(a_A,1),(b_A,1)},\,
        (E_s)_{(a_A,1),(b_A,2)},\,
        \ldots,\,
        (E_s)_{(a_A,d_B),(b_A,d_B)}
    }.
\end{equation} 
The number of all possible choices is then given by taking the product over all $(a_A,b_A)$, which yields $\kappa_s^A$, and also is independent of the representative $(\underline{i_A},\underline{j_A})$ we chose. This proves~\eqref{eq:partial_trace_B_orbit}.
\end{proof}A completely analogous argument holds for $\Tr_{A^n}[C_s]$, simply swapping the roles of $A$ and $B$.

\begin{remark}
Note that in the special case where the system $A$ is trivial ($d_A = 1$), then $\Tr_{B^n}[C_s^{AB}]$ becomes the full trace over $B^n$. In that case, $E_{r(s)} = n$ and the multinomial coefficient in \eqref{eq:kappa_A} reduces to \eqref{eq:formula_trace_orbit_matrices}, recovering item 2 of Lemma \ref{lemma:trace_HS_orbit}. 
\end{remark}
By linearity of the partial trace and by Lemma \ref{lemma:partial_trace_orbit}, we thus have for $X_{A^nB^n} = \sum_{r\in [m_{AB}]} x_r\, C_r^{AB}$:
\begin{equation}
    \Tr_{B^n}[X_{A^n B^n}] = \sum_{r =1}^{m_A}\left(\sum_{s=1}^{m_{AB}} x_s \kappa_s^A \tau_{r, s}\right) C_r^A.
\end{equation}
\subsubsection{Efficient Channel Concatenations}
Channel concatenations are closely related to partial traces by the link-product formula for their Choi matrices. Specifically, recall from \eqref{eq:concatenation_choi} that if channel $\mathcal{N} \in \CPTP(A \to B)$ has Choi matrix $\Gamma^{\mathcal{N}}_{AB}$ and channel $\mathcal{D} \in \CPTP(B \to C)$ has Choi matrix $\Gamma^{\mathcal{D}}_{BC}$, then the concatenated channel $\mathcal{M} = \mathcal{D} \circ \mathcal{N}$ has Choi matrix\begin{equation*}
    \Gamma_{AC}^{\mathcal{M}} = \Tr_B[(\Gamma^{\mathcal{N}}_{AB})^{T_B} \Gamma^{\mathcal{D}}_{BC}],
\end{equation*} 
where there are implicit identities on the $A$ and $C$ system, and $T_B$ is the partial transpose on $B$.

We first show how to efficiently compute the serial concatenation between channels sharing permutation-invariance with respect to a common system $B$. Specifically, let $\mathcal{N}_n \in \CPTP(A^n \to B^n)$ and $\mathcal{D}_n \in \CPTP(B^n \to R)$ such that $\Gamma^{\mathcal{N}_n}_{A^n B^n} \in \mathrm{End}^{\mathfrak{S}_n}(A^n \otimes B^n)$, i.e. permutation covariant, and $\Gamma^{\mathcal{D}_n}_{B^n R} \in \mathrm{End}^{\mathfrak{S}_n}(B^n \otimes R)$, i.e.\ permutation invariant on its input. We now want to efficiently compute the Choi matrix of $\mathcal{M}_{n} \coloneqq \mathcal{D}_n \circ \mathcal{N}_n$. This is possible with the following proposition, which establishes the functional relation between the orbit-basis coefficients of $\Gamma^{\mathcal{M}_n}_{A^n R}$ and those of $\Gamma^{\mathcal{N}_n}_{A^n B^n}$ and $\Gamma^{\mathcal{D}_n}_{B^n R} $.

\begin{proposition}
\label{prop:serial_concatenation}
Let $t \in [m_B]$ be any orbit on $\mathrm{End}^{\mathfrak{S}_n}(B^n)$ and let $s \in [m_{AB}]$ be a joint orbit in $\esn((AB)^n)$ with count matrix $E_s$ and corresponding $A$-marginal $E_{r(s)}$ and $B$-marginal $E_{t(s)}$. Then, the action of the partial-trace map
\begin{equation}
\label{eq:partial_trace_map}
\begin{aligned}
  \mathcal{T}_B\colon \mathrm{End}^{\mathfrak{S}_n}(A^{\otimes n} \otimes B^{\otimes n}) \times \mathrm{End}^{\mathfrak{S}_n}(B^{\otimes n}) &\to \mathrm{End}^{\mathfrak{S}_n}(A^{\otimes n})\\
   (C_s^{AB}, C_t^B) &\mapsto \mathrm{Tr}_{B^n}\!\left[C_s^{AB} \cdot (\mathbbm{1}_{A^n} \otimes C_{t^T}^B)\right]
\end{aligned}
\end{equation}
can be efficiently computed as
\begin{equation}
\label{eq:partial_trace_identity}
 \mathcal{T}_B(C_s^{AB}, C_t^B) = \kappa_s^A \delta_{t, t(s)} \cdot C_{r(s)}^A,
\end{equation}
where $\kappa_s^A$ was defined in \eqref{eq:kappa_A} and we denoted\begin{equation}
\label{eq:delta_r_r(s)}
    \delta_{t, t(s)} \coloneqq \begin{cases}1\quad \text{if} \   t(s) = t\\
    0  \quad\text{otherwise}.
    \end{cases}
\end{equation}
\end{proposition}

\begin{proof}
It is easy to see that the partial trace of a jointly permutation-invariant matrix is a permutation invariant matrix, and hence the action of \eqref{eq:partial_trace_map} on a generic basis element can be expressed as $\mathcal{T}_B(C_s^{AB}, C_t^B) = \sum_{r=1}^{m_A} c_r^{s,t} C_r^A$, for some coefficients $\{ c_r^{s,t}\}_{r \in [m_A]}$ to be determined. For fixed $s \in [m_{AB}]$ and $t \in [m_B]$, the action can be expressed using multi-indices by explicitly writing out the matrix multiplication and the partial trace:
\begin{align}
\left(\mathrm{Tr}_{B^n}\!\left[C_s^{AB} \cdot (\mathbbm{1}_{A^n} \otimes C_{t^T}^B)\right]\right)_{(\underline{i_A},\,\underline{j_A})}
&= \Tr_{B^n}\left(\sum_{\underline{i_B} \in [d_B]^{\times n}} (C_s^{AB})_{(\underline{i_A},\,\underline{i_B}),\,(\underline{j_A},\,\underline{j_B})} \cdot (C_{t^T}^B)_{(\underline{j_B},\,\underline{i_B})} \right) \\
\label{eq:partial_trace_elementwise}
&= \sum_{\underline{j_B} \in [d_B]^{\times n}} \sum_{\underline{i_B} \in [d_B]^{\times n}} (C_s^{AB})_{(\underline{i_A},\,\underline{i_B}),\,(\underline{j_A},\,\underline{j_B})} \cdot (C_{t}^B)_{(\underline{i_B}, \underline{j_B})},
\end{align}
for all $\underline{i_A},\,\underline{j_A} \in [d_A]^{\times n}$.
Both factors in \eqref{eq:partial_trace_elementwise} are $0$--$1$ valued, so a nonzero contribution requires simultaneously:
\begin{equation}
\label{eq:orbit_matrices_condition}
    (C_s^{AB})_{(\underline{i_A},\,\underline{i_B}),\,(\underline{j_A},\,\underline{j_B})} = 1, \qquad (C_{t}^B)_{(\underline{i_B}, \underline{j_B})} = 1.
\end{equation}
These two conditions impose $(\underline{i_A}\,\underline{i_B}, \underline{j_A}\,\underline{j_B})$ to lie in orbit $s$ and $(\underline{i_B}, \underline{j_B})$ to lie in orbit $t$, that is:
\begin{align}
\label{eq:condition_marginal}
 E^{(\underline{i_A}\,\underline{i_B},\, \underline{j_A}\,\underline{j_B})} = E_s, \qquad 
 E^{(\underline{i_B},\,\underline{j_B})} = E_t.
\end{align}
Combining the two conditions, $E_t$ must equal the $B$-marginal of $E_s$ as defined in \eqref{eq:B_marginal}; in other words, the map \eqref{eq:partial_trace_map} acts non-trivially on $\mathrm{End}^{\mathfrak{S}_n}(B^{\otimes n})$ only for the $B$-marginal orbit $t(s)$, i.e., its action is completely determined by the joint orbit $s \in [m_{AB}]$. In the same way, by \eqref{eq:condition_marginal}, for fixed $s \in [m_{AB}]$ the only multi-indices $(\underline{i_A}, \underline{j_A})$ that give a nonzero contribution in \eqref{eq:partial_trace_elementwise} are those belonging to its unique $A$-marginal orbit $r(s) \in [m_A]$. The proportionality constant $\kappa_s^A$ counts, for a fixed $(\underline{i_A}, \underline{j_A})$ in orbit $r(s)$, the number of pairs $(\underline{i_B}, \underline{j_B})$ satisfying \eqref{eq:orbit_matrices_condition}. By the same argument of Lemma \ref{lemma:partial_trace_orbit}, this is the multinomial coefficient $\kappa_s^A$ in \eqref{eq:kappa_A}. We thus obtained
\begin{equation}
\label{eq:partial_trace_identity_final}
 \mathcal{T}_B(C_s^{AB}, C_t^B) = \mathrm{Tr}_{B^n}\!\left[C_s^{AB} \cdot (\mathbbm{1}_{A^n} \otimes C_{t^T}^B)\right] = \begin{cases} \kappa_s^A \cdot C_{r(s)}^A & \text{if } t = t(s), \\ 0 & \text{otherwise.} \end{cases}
\end{equation}
The computation is efficient because $m_{AB} = O(n^{d_A^2 d_B^2-1})$ and computing $r(s), t(s), \kappa_s^A$ from $E_s$ requires a constant (in $n$) number of operations per orbit.
\end{proof}
\begin{remark}\mbox{}
\begin{itemize}
    \item Item 1 of Lemma \ref{lemma:trace_HS_orbit} is a special case of Proposition \ref{prop:serial_concatenation} when the $A$ system is trivial ($d_A = 1$). In that case, the map \eqref{eq:partial_trace_map} becomes\begin{equation*}
          \mathcal{T}_B
   (C_s^{AB}, C_t^B)= \mathrm{Tr}_{B^n}\!\left[C_{t'}^{B} \cdot C_{t^T}^B\right] = \Tr[( C_{t}^B)^\dag C_{t'}^{B}],
    \end{equation*}
    where $m_{AB} = m_B$ and we identified the $B$-marginal as $t'\equiv t(s) = s$. The $A$-marginal is the unique trivial orbit, so $C_{r(s)}^A = 1$, and the multinomial prefactor \eqref{eq:kappa_A} reduces to the single multinomial coefficient of the full count matrix, which is precisely the orbit size:
\begin{equation}
    \kappa_{t'}^B = \binom{n}{(E_{t'})_{(1,1)}, (E_{t'})_{(1,2)}, \dots, (E_{t'})_{(d_A, d_A)}} = |O_{t'}|.
\end{equation}
The identity~\eqref{eq:partial_trace_identity} then reads \begin{equation}
    \Tr[( C_{t}^B)^\dag C_{t'}^{B}] = \delta_{t',t}\, |O_{t'}|,
\end{equation}which is precisely \eqref{eq:inner_product_HS}.

    \item Lemma \ref{lemma:partial_trace_orbit} is a special case of Proposition \ref{prop:serial_concatenation} when $t = t(s)$ corresponds to a diagonal orbit. Specifically, since\begin{align}
        \Tr_{B^n}[C_s^{AB}] &= \Tr_{B^n}[C_s^{AB} \cdot (\mathbbm{1}_{A^n} \otimes \mathbbm{1}_{B^n})] \\
        &= \sum_{t \in [m_B], E_t \ \mathrm{diagonal}}  \Tr_{B^n}[C_s^{AB} \cdot (\mathbbm{1}_{A^n} \otimes C_{t^T}^B)] \\
        &= \sum_{t \in [m_B], E_t \ \mathrm{diagonal}}  \kappa_s^A  \delta_{t,t(s)} C_{r(s)}^A\\
        &= \tau_{s} \kappa_s^A C_{r(s)}^A,
    \end{align} where we used that \begin{equation}
    \label{eq:identity_orbit}
        \mathbbm{1}_{B^n} = \sum_{t \in [m_B], E_t \ \text{diag}} C_t^B,
    \end{equation} and in the last line we used the definition of $\tau_s$ in \eqref{eq:tau_r_s}.
\end{itemize}
\end{remark}
By using this Proposition, we can efficiently compute the orbit basis coefficients of a concatenation of channels $\mathcal{N}_n: \CPTP(A^n \to B^n)$ and $\mathcal{D}_n \in \CPTP(B^n \to R)$ whose Choi matrices have the form:\begin{equation}
   \Gamma^{\mathcal{N}_n}_{A^n B^n} = \sum_{s= 1}^{m_{AB}}  c_{s}^{\mathcal{N}} C^{AB}_{s} \qquad \Gamma^{\mathcal{D}_{n}}_{B^nR} = \sum_{t= 1}^{m_{B}} \sum_{k,l= 1}^{d_{R}}   c_{t,k,l}^{\mathcal{D}} C^{B}_{t}\otimes \ket{k}\bra{l},
\end{equation}
by using \eqref{eq:concatenation_choi}, linearity and Proposition \ref{prop:serial_concatenation}, we obtain for the channel $\mathcal{M}'_n = \mathcal{D}_n \circ \mathcal{N}_n$:\begin{equation}
    \Gamma^{\mathcal{M}'_{n}}_{A^nR} = \sum_{r= 1}^{m_{A}} \sum_{k,l= 1}^{d_{R}}   c_{r,k,l}^{\mathcal{M}'} C^{A}_{r}\otimes \ket{k}\bra{l},
\end{equation}
where
\begin{equation}
\label{eq:concatenation_coefficients_AlicePOV}
    c_{r,k,l}^{\mathcal{M}'} = \sum_{s \in [m_{AB}]:\, r(s) = r} \kappa_s^A c_{t(s),k,l}^{\mathcal{D}} c_s^{\mathcal{N}}.
\end{equation}
The same argument can be used to efficiently compute the map\begin{equation}
\begin{aligned}
  \mathcal{T}_A\colon \mathrm{End}^{\mathfrak{S}_n}(A^{\otimes n}) \times \mathrm{End}^{\mathfrak{S}_n}(A^{\otimes n} \otimes B^{\otimes n}) &\to \mathrm{End}^{\mathfrak{S}_n}(B^{\otimes n})\\
   (C_r^{A}, C_s^{AB}) &\mapsto  \mathrm{Tr}_{A^n}\!\left[(C_{r^T}^A \otimes \mathbbm{1}_{B^n}) \cdot C_s^{AB}\right],
\end{aligned}
\end{equation}
as
\begin{equation}
 \mathcal{T}_A(C_r^A, C_s^{AB})  = \kappa_s^B \delta_{r, r(s)} \cdot C_{t(s)}^B,
\end{equation}
where\begin{equation}
\label{eq:kappa_B}
    \kappa_s^B \coloneqq \prod_{(a_B, b_B) \in [d_B]^2} \binom{(E_{t(s)})_{(a_B,b_B)}}{(E_{s})_{(1,a_B),(1, b_B)}, (E_{s})_{(1,a_B),(2, b_B)} , ..., (E_{s})_{(d_A,a_B),(d_A, b_B)}}.
\end{equation}
Using this, we can also efficiently compute the orbit basis coefficients of a concatenation of channels $\mathcal{N}_n: \CPTP(A^n \to B^n)$ and $\mathcal{E}_n \in \CPTP(R \to A^n)$ whose Choi matrices have the form:\begin{equation}
   \Gamma^{\mathcal{N}_n}_{A^n B^n} = \sum_{s= 1}^{m_{AB}}  c_{s}^{\mathcal{N}} C^{AB}_{s} \qquad \Gamma^{\mathcal{E}_{n}}_{RA^n} = \sum_{l,k= 1}^{d_{R}}   c_{k,l,r}^{\mathcal{E}} \ket{k}\bra{l}\otimes C_r^A,
\end{equation}
by using \eqref{eq:concatenation_choi}, linearity and Proposition \ref{prop:serial_concatenation}, we obtain for the channel $\mathcal{M}_n = \mathcal{N}_n \circ \mathcal{E}_n$:\begin{equation}
        \Gamma^{\mathcal{M}_n}_{R B^n} = \sum_{k,l= 1}^{d_{R}} \sum_{t= 1}^{m_{B}}  c_{k,l,t}^{\mathcal{M}} \ket{k}\bra{l} \otimes C^{B}_{t} 
\end{equation}
where
\begin{equation}
\label{eq:concatenation_coefficients_BobPOV}
     c_{k,l,t}^{\mathcal{M}}  = \sum_{s \in [m_{AB}]:\,t(s) = t}  \kappa_s^B c_{k,l, r(s)}^{\mathcal{E}} c_s^{\mathcal{N}}.
\end{equation}
These are used in the symmetric seesaw method of Section \ref{sec:symmetric_seesaw}, and we used the same notation here for consistency. \begin{remark}
    Note that to actually implement the map $\mathcal{T}_B$ in Proposition \ref{prop:serial_concatenation}, one needs to (pre-)compute $κ_s^A$, $t(s)$ and $r(s)$ for every orbit $C_s^{AB} \in \esn(A\n\otimes B\n)$. This can be done independently for every $s$, and in particular, if one works with objects only supported on a subset of the basis matrices $\{C_{s}^{AB}\}_s$ it is sufficient to only compute the partial-trace map on this subset. This is for example the case for the Choi matrix of a tensor product channel $\mathcal{N}\n$, $\Gamma^{\mathcal{N}\n}_{A^n B^n} \in \esn(A\n\otimes B\n)$, if the underlying Choi matrix of the single channel $\mathcal{N}$, $\Gamma^{\mathcal{N}}_{AB} \in \mathcal{L}(AB)$, is sparse. the contributing orbit basis matrices are then exactly the ones whose count matrices have the same sparsity structure as $\Gamma^{\mathcal{N}}_{AB}$ (by Lemma \ref{lem:support_and_counts}). 
\end{remark}
In the case where two channels are permutation covariant, i.e. $\Gamma^{\mathcal{N}_n}_{A^n B^n} \in \mathrm{End}^{\mathfrak{S}_n}(A^n \otimes B^n)$ and $\Gamma^{\mathcal{O}_n}_{B^n C^n} \in \mathrm{End}^{\mathfrak{S}_n}(B^n \otimes C^n)$, then a simple modification of the above result allows to efficiently compute orbit basis coefficients of the Choi matrix of their concatenation $\mathcal{P}_n \coloneqq \mathcal{O}_n \circ \mathcal{N}_n$, denoted $\{c^{\mathcal{O}}_w\}_{w \in [m_{AC}]}$. In this setting, since orbits $s \in [m_{AB}]$ and $u \in [m_{BC}]$ are associated with count matrices on systems $AB$ and $BC$, it is convenient to define tripartite count matrices $E_{ABC} \coloneqq E_z \in \mathbb{N}^{(d_A d_B d_C) \times (d_A d_B d_C)}$, with associated orbits $z \in [m_{ABC}]$. We will denote $s(z)$, $u(z)$ and $w(z)$ to denote the marginal orbits on $AB$, $BC$ and $AC$ obtained from $E_z$ by marginalizing with respect to the third system, analogously to \eqref{eq:A_marginal}:\begin{align}
\label{eq:V_marginal_s}
    (E_{AB})_{(a_A, a_B), (b_A, b_B)}  &\coloneqq (E_{s(z)})_{(a_A, a_B), (b_A, b_B)}  = \sum_{a_C, b_C = 1}^{d_C} (E_{z})_{(a_A, a_B, a_C), (b_A, b_B, b_C)},\\
     \label{eq:V_marginal_u}
       (E_{BC})_{(a_B, a_C), (b_B, b_C)} &\coloneqq (E_{u(z)})_{(a_B, a_C), (b_B, b_C)}  = \sum_{a_A, b_A = 1}^{d_A} (E_{z})_{(a_A, a_B, a_C), (b_A, b_B, b_C)},\\
        \label{eq:V_marginal_w}
       (E_{AC})_{(a_A, a_C), (b_A, b_C)}  &\coloneqq  (E_{w(z)})_{(a_A, a_C), (b_A, b_C)} =\sum_{a_B, b_B = 1}^{d_B} (E_{z})_{(a_A, a_B, a_C), (b_A, b_B, b_C)}.
\end{align}
\begin{proposition}
\label{prop:serial_concatenation_covariant}
Let $s \in [m_{AB}]$ be an orbit on $\mathrm{End}^{\mathfrak{S}_n}(A^{\otimes n} \otimes B^{\otimes n})$ with count matrix $E_s$, and let $u \in [m_{BC}]$ be an orbit on $\mathrm{End}^{\mathfrak{S}_n}(B^{\otimes n} \otimes C^{\otimes n})$ with count matrix $E_u$. Then, the action of the partial-trace map
\begin{equation}
\label{eq:partial_trace_map_covariant}
\begin{aligned}
  \mathcal{T}_B\colon \mathrm{End}^{\mathfrak{S}_n}(A^{\otimes n} \otimes B^{\otimes n}) \times \mathrm{End}^{\mathfrak{S}_n}(B^{\otimes n} \otimes C^{\otimes n}) &\to \mathrm{End}^{\mathfrak{S}_n}(A^{\otimes n}\otimes C^{\otimes n})\\
   (C_s^{AB}, C_u^{BC}) &\mapsto \mathrm{Tr}_{B^n}\!\left[\left(C_{s^{T_B}}^{AB}\otimes \mathbbm{1}_{C^n}\right) \cdot \left( \mathbbm{1}_{A^n} \otimes C_{u}^{BC}\right) \right]
\end{aligned}
\end{equation}
can be efficiently computed as
\begin{equation}
\label{eq:covariant_concatenation_identity}
    \mathcal{T}_B(C_s^{AB}, C_u^{BC}) = \sum_{w=1}^{m_{AC}} \mathcal{K}_{s,u}^w \, C_w^{AC},
\end{equation}
where $\mathcal{K}_{s,u}^w$ counts the number of compatible intermediate $B$-system configurations:
\begin{equation}
\label{eq:kappa_covariant}
    \mathcal{K}_{s,u}^w \;\coloneqq\; \sum_{E_{z}}
    \delta_{s,\, s(z)^{T_B}}\,
    \delta_{u,\, u(z)^{T_B}}\,
    \delta_{w,\, w(z)}
    \prod_{\substack{a_A, b_A \in [d_A] \\ a_C, b_C \in [d_C]}}
    \binom{(E_{w(z)})_{(a_A, a_C),(b_A, b_C)}}{\{(E_z)_{(a_A, a_B, a_C),(b_A, b_B, b_C)}\}_{a_B, b_B \in [d_B]}}.
\end{equation}
Here the sum extends over all tripartite count matrices $E_z \in \mathbb{N}^{(d_A d_B d_C) \times (d_A d_B d_C)}$ satisfying \eqref{eq:count_matrix}, the marginals $s(z), u(z), w(z)$ are defined in \eqref{eq:V_marginal_s}--\eqref{eq:V_marginal_w}, and the partial transpose on $B$ is understood as in \eqref{eq:count-orbit-partial-transpose-relation}.
\end{proposition}

\begin{proof}
The action of \eqref{eq:partial_trace_map_covariant} on a generic basis element can be expressed as $\mathcal{T}_B(C_s^{AB}, C_u^{BC}) = \sum_{w=1}^{m_{AC}} c_w^{s,u}\, C_w^{AC}$, for some coefficients $\{c_w^{s,u}\}_{w \in [m_{AC}]}$ to be determined. For fixed $s \in [m_{AB}]$ and $u \in [m_{BC}]$, we expand the partial trace and matrix multiplication elementwise:
\begin{align}
    &\left(\mathrm{Tr}_{B^n}\!\left[(C_{s}^{AB}\otimes \mathbbm{1}_{C^n})\cdot(\mathbbm{1}_{A^n}\otimes C_{u^{T_B}}^{BC})\right]\right)_{(\underline{i_A}\,\underline{i_C},\,\underline{j_A}\,\underline{j_C})} \notag\\
    &\qquad= \sum_{\underline{i_B}, \underline{j_B} \in [d_B]^n}
    (C_{s}^{AB})_{(\underline{i_A},\,\underline{i_B}),\,(\underline{j_A},\,\underline{j_B})}
    \cdot (C_{u^{T_B}}^{BC})_{(\underline{j_B},\,\underline{i_C}),\,(\underline{i_B},\,\underline{j_C})} \notag\\
    \label{eq:covariant_elementwise}
    &\qquad= \sum_{\underline{i_B}, \underline{j_B} \in [d_B]^n}
    (C_s^{AB})_{(\underline{i_A},\,\underline{i_B}),\,(\underline{j_A},\,\underline{j_B})}
    \cdot (C_u^{BC})_{(\underline{i_B},\,\underline{i_C}),\,(\underline{j_B},\,\underline{j_C})},
\end{align}
where in the last step we used \eqref{eq:count-orbit-partial-transpose-relation}. Since orbit matrices are $0$--$1$ valued, a nonzero contribution requires both factors in~\eqref{eq:covariant_elementwise} to equal $1$, i.e.
\begin{equation}
    \label{eq:condition_marginal_covariant_final}
    E^{(\underline{i_A}\,\underline{i_B},\,\underline{j_A}\,\underline{j_B})} = E_s,
    \qquad
    E^{(\underline{i_B}\,\underline{i_C},\,\underline{j_B}\,\underline{j_C})} = E_u.
\end{equation}
We are thus looking for the number of multi-index pairs $(\underline{i_B}, \underline{j_B}$) that satisfy these two constraints simultaneously. Note that the triple $(\underline{i_A}\, \underline{j_A}, \underline{i_B}\,\underline{j_B}, \underline{i_C}\, \underline{j_C})$ uniquely determines a tripartite count matrix $E_z \in \mathbb{N}^{(d_Ad_Bd_C)\times(d_Ad_Bd_C)}$, whose generic entry $(E_z)_{(a_A, a_B, a_C),(b_A, b_B, b_C)}$ counts the number of positions $k \in [n]$ where $(i_A)_k = a_A$, $(j_A)_k = b_A$, $(j_B)_k = a_B$, $(i_B)_k = b_B$, $(i_C)_k = a_C$, $(j_C)_k = b_C$. Substituting this into~\eqref{eq:condition_marginal_covariant_final} shows that the the conditions read precisely\begin{equation}
    s(z) = s\qquad u(z) = u\qquad w(z) = w.
\end{equation}
where w is the orbit of $(\underline{i_A}\,\underline{i_C},\underline{j_A}\,\underline{j_C})$.
It remains to count, for fixed outer indices $(\underline{i_A}, \underline{j_A}, \underline{i_C}, \underline{j_C})$ lying in orbit $w$ and fixed tripartite count matrix $E_z$ satisfying the three marginal constraints, the number of pairs $(\underline{i_B}, \underline{j_B})$ realising $E_z$. For each outer type $((a_A, b_A), (a_C, b_C))$, there are $(E_{w(z)})_{(a_A, a_C),(b_A, b_C)}$ positions whose outer labels are fixed, and the intermediate pairs $((j_B)_m, (i_B)_m)$ must be distributed among the $d_B \times d_B$ possible values $(a_B, b_B)$ with multiplicities $(E_z)_{(a_A, a_B, a_C),(b_A, b_B, b_C)}$. The number of such assignments is the multinomial coefficient
\begin{equation}
    \binom{(E_{w(z)})_{(a_A, a_C),(b_A, b_C)}}{\{(E_z)_{(a_A, a_B, a_C),(b_A, b_B, b_C)}\}_{a_B, b_B \in [d_B]}}.
\end{equation}
Taking the product over all $(a_A, b_A, a_C, b_C) \in [d_A]^2 \times [d_C]^2$ and summing over all tripartite $E_z$ satisfying the marginal constraints yields $\mathcal{K}_{s,u}^w$ as in~\eqref{eq:kappa_covariant}, independently of the chosen representative of orbit $w$. This proves~\eqref{eq:covariant_concatenation_identity}.
\end{proof}

The number of tripartite count matrices $E_z$ is bounded by $m_{ABC} = \binom{n + d_A^2 d_B^2 d_C^2 - 1}{d_A^2 d_B^2 d_C^2 - 1} = O(n^{d_A^2 d_B^2 d_C^2 - 1})$, and computing the three marginals and the multinomial product for each $E_z$ requires $\mathrm{poly}(n)$ operations per orbit. Hence $\mathcal{K}_{s,u}^w$ is computable in $\mathrm{poly}(n)$ time for fixed $d_A, d_B, d_C$, although with a scaling exponent larger than in the non-covariant case of Proposition~\ref{prop:serial_concatenation}.
\begin{remark}
When system $C$ is trivial ($d_C = 1$), Proposition~\ref{prop:serial_concatenation_covariant} recovers Proposition~\ref{prop:serial_concatenation} as a special case. Indeed, the tripartite count matrix $E_z$ collapses to a bipartite count matrix on $AB$, so that the marginal $E_{s(z)} \equiv E_z$ is itself a count matrix on $AB$, $E_{u(z)}$ reduces to a count matrix on $B$ alone, and $E_{w(z)}$ reduces to a count matrix on $A$ alone. In this regime, the three marginals of $z$ coincide with the joint orbit $s(z)$ and its $A$- and $B$-marginals as defined in~\eqref{eq:A_marginal}--\eqref{eq:B_marginal}, i.e.\ $w(z) = r(s(z))$ and $u(z) = t(s(z))$. The constraint $s(z) = s$ then fixes $E_z = E_{s}$, reducing the sum in~\eqref{eq:kappa_covariant} to a single term; the remaining two constraints become $u = t(s)$ and $w = r(s)$. The multinomial factor simplifies to
\begin{equation}
    \prod_{a_A, b_A \in [d_A]}
    \binom{(E_{r(s)})_{(a_A, b_A)}}{\{(E_s)_{(a_A, a_B),(b_A, b_B)}\}_{a_B, b_B \in [d_B]}} = \kappa_s^A,
\end{equation}
where we used~\eqref{eq:kappa_A}. We thus obtain
\begin{equation}
    \mathcal{K}_{s,u}^w = \kappa_s^A\, \delta_{u, t(s)}\, \delta_{w, r(s)},
\end{equation}
and~\eqref{eq:covariant_concatenation_identity} reads $\mathcal{T}_B(C_s^{AB}, C_u^{B}) = \kappa_s^A\, \delta_{u, t(s)}\, C_{r(s)}^A$, which is precisely~\eqref{eq:partial_trace_identity}.
\end{remark}By using this Proposition, we can efficiently compute the orbit basis coefficients of the concatenation of two permutation-covariant channels $\mathcal{N}_n \in \CPTP(A^n \to B^n)$ and $\mathcal{O}_n \in \CPTP(B^n \to C^n)$, whose Choi matrices have the form
\begin{equation}
    \Gamma^{\mathcal{N}_n}_{A^n B^n} = \sum_{s=1}^{m_{AB}} c_{s}^{\mathcal{N}}\, C^{AB}_{s}, \qquad
    \Gamma^{\mathcal{O}_n}_{B^n C^n} = \sum_{u=1}^{m_{BC}} c_{u}^{\mathcal{O}}\, C^{BC}_{u}.
\end{equation}
By using \eqref{eq:concatenation_choi}, linearity and Proposition~\ref{prop:serial_concatenation_covariant}, we obtain for $\mathcal{P}_n = \mathcal{O}_n \circ \mathcal{N}_n$
\begin{equation}
    \Gamma^{\mathcal{P}_n}_{A^n C^n} = \sum_{w=1}^{m_{AC}} c_w^{\mathcal{P}}\, C^{AC}_w,
\end{equation}
where
\begin{equation}
\label{eq:concatenation_coefficients_covariant}
    c_w^{\mathcal{P}} = \sum_{\substack{s \in [m_{AB}]\\ u \in [m_{BC}]}} \mathcal{K}_{s,u}^w \cdot c_{s}^{\mathcal{N}}\, c_{u}^{\mathcal{O}},
\end{equation}
with $\mathcal{K}_{s,u}^w$ given by~\eqref{eq:kappa_covariant}.

\begin{remark}
These results on concatenations of channels can be extended to compute the action of a generic permutation-covariant channel $\mathcal{N}_n \in \CPTP(A^n \to B^n)$ on permutation-invariant states, using the link-product formula
\begin{equation}
\label{eq:choi_jamiolkovski_isomorphism}
    \omega_{RB^n} \coloneqq \mathcal{N}_{n}(\rho_{RA^n}) = \Tr_{A^n}\bigl[(\rho_{RA^n}^{T_{A^n}} \otimes \mathbbm{1}_{B^n})\,\Gamma^{\mathcal{N}}_{A^nB^n}\bigr],
\end{equation}
or 
\begin{equation}
\label{eq:choi_jamiolkovski_isomorphism_covariant}
    \omega_{R^nB^n} \coloneqq \mathcal{N}_{n}(\rho_{R^nA^n}) = \Tr_{A^n}\bigl[(\rho_{R^nA^n}^{T_{A^n}} \otimes \mathbbm{1}_{B^n})\,\Gamma^{\mathcal{N}}_{A^nB^n}\bigr],
\end{equation}
which makes the action of $\mathcal{N}_n$ formally analogous to a channel concatenation where the first channel is a \textit{preparation channel}:
\begin{equation}
    \begin{aligned}
        \mathcal{P}_\rho: \mathbb{C} &\to \mathcal{D}(\mathcal{H})\\
        &1 \mapsto \rho,
    \end{aligned}
\end{equation}
such that $\Phi^{\mathcal{P}_\rho} = \rho$. Thus, the output state $\omega$ is then the Choi matrix of the concatenated channel $\mathcal{N}_n \circ \mathcal{P}_\rho$, in line with \eqref{eq:concatenation_choi}. Then, if $\rho_{RA^n} \in \mathrm{End}^{\mathfrak{S}_n}(R \otimes A^{\otimes n})$, expanded as
\begin{equation}
    \rho_{RA^n} = \sum_{r=1}^{m_A}\sum_{k,l=1}^{d_R} c_{k,l,r}^{\rho}\, \ket{k}\!\bra{l}_R \otimes C_r^A,
\end{equation}
applying Proposition~\ref{prop:serial_concatenation} on \eqref{eq:choi_jamiolkovski_isomorphism} yields the orbit-basis expansion
\begin{equation}
\label{eq:channel_action_R_A_n}
    \omega_{RB^n} = \sum_{t=1}^{m_B}\sum_{k,l=1}^{d_R} c_{k,l,t}^{\omega}\, \ket{k}\!\bra{l}_R \otimes C_t^B, \qquad c_{k,l,t}^{\omega} = \sum_{s \in [m_{AB}]:\, t(s) = t} \kappa_s^B\, c_{k,l,r(s)}^{\rho}\, c_s^{\mathcal{N}}.
\end{equation}
Similarly, if $\rho_{R^nA^n} \in \mathrm{End}^{\mathfrak{S}_n}(R^{\otimes n} \otimes A^{\otimes n})$, expanded as
\begin{equation}
    \rho_{R^nA^n} = \sum_{v=1}^{m_{RA}} c_v^{\rho}\, C_v^{RA},
\end{equation}
the preparation channel $\mathcal{P}_\rho: \mathbb{C} \to R^n A^n$ is itself permutation-covariant (the input is trivially one-dimensional), and Proposition~\ref{prop:serial_concatenation_covariant} applies on \eqref{eq:choi_jamiolkovski_isomorphism_covariant}, yielding the orbit-basis expansion
\begin{equation}
\label{eq:channel_action_R_n_A_n}
    \omega_{R^nB^n} = \sum_{w=1}^{m_{RB}} c_w^{\omega}\, C_w^{RB}, \qquad c_w^{\omega} = \sum_{\substack{v \in [m_{RA}]\\ s \in [m_{AB}]}} \mathcal{K}_{v,s}^w\, c_v^{\rho}\, c_s^{\mathcal{N}},
\end{equation}
where $\mathcal{K}_{v,s}^w$ is defined as in~\eqref{eq:kappa_covariant} with the role of system $B$ in the proposition played here by the contracted system $A$.
\end{remark}
These results allow to compute (partial) traces and (partial) transposes of permutation invariant operators. In SDPs, we also need to impose the positive semidefinite (PSD) constraints without constructing the orbit matrices, which have exponential size in $n$.
To do this, we make use of the representation theory of the symmetric group, and to decompose the PSD constraint into polynomially many blocks of polynomial size. Next section summarizes the construction (for more details, refer to \cite{JamesKerber1981, Sagan2001}). 

\subsection{Representation Theory of the Symmetric Group}\label{sec:endsn_block_diagonalization}
Since $\mathrm{End}^{\mathfrak{S}_n}(\mathcal{H}^{\otimes n})$ is a $*$-matrix algebra, there exists a $*$-algebra isomorphism which maps its element to block-diagonal form. By Schur--Weyl duality, the blocks (irreducible representations of $\mathfrak{S}_n$) are indexed by partitions of $n$:
\begin{equation}
\label{eq:schur_weyl_iso}
    \psi_{\mathcal{H}} : \mathrm{End}^{\mathfrak{S}_n}(\mathcal{H}^{\otimes n}) \rightarrow\bigoplus_{\lambda \in \mathrm{Par}(d,n)} \mathbb{C}^{m_\lambda \times m_\lambda}
\end{equation}
where $\mathrm{Par}(d,n)$ is the set of partitions of $n$ with at most $d = d_{\mathcal{H}}$ parts, and $m_\lambda \coloneqq |\mathcal{T}_{\lambda,d}|$ is the number of \textit{semistandard Young tableaux} of shape $\lambda$ with entries in $[d]$ \cite{Sagan2001}. Note that\begin{equation}
    \sum_{\lambda \in \mathrm{Par}(d,n)} m_\lambda^2 = m_\mathcal{H},
\end{equation}
as defined in \eqref{eq:dim_space_perm_inv}. Crucially, by \cite[Proposition 2.4.4]{Polak2020}, this is a positive map, i.e. it preserves positive semidefiniteness:\begin{equation}
    X \geq 0 \iff \psi(X) \geq 0 \iff [\psi(X)]_{\lambda} \geq 0 \quad \forall \lambda \in \mathrm{Par}(d,n),
\end{equation}
where $X \in \mathrm{End}^{\mathfrak{S}_n}(\mathcal{H}^{\otimes n})$.
 
Both the number of blocks $t \coloneqq |\mathrm{Par}(d,n)|$ and the size of each block $m_\lambda \coloneqq |\mathcal{T}_{\lambda,d}|$ are polynomial in $n$:
\begin{equation}
\label{eq:bound_SSYTs}
    t \leq (n+1)^d, \qquad m_\lambda \leq (n+1)^{d(d-1)/2}
\end{equation}

In general, computing the block-diagonal decomposition above and the map $\psi$ for a group $G$ is a non-trivial procedure, and we need some tools from representation theory of finite groups. We now summarize the construction of $\psi$ for the case $G = \mathfrak{S}_n$, following \cite{Fawzi2022}.

A \textit{partition} $\lambda$ of $n$ is a non-increasing sequence of natural numbers $(\lambda_1 \geq \cdots \geq \lambda_h > 0)$ with $\sum_i \lambda_i = n$; we write $\lambda \vdash_d n$ if the number of parts (\textit{height}) satisfies $h \leq d$. The \textit{Young diagram} of shape $\lambda$ is an array of $n$ boxes arranged in rows of lengths $\lambda_1, \ldots, \lambda_h$, and contains the same information as the partition $\lambda$.
The boxes are labeled $1, \ldots, n$ in lexicographic order of their positions (the \textit{canonical tableau}).

Given $\lambda \vdash_d n$, the \textit{row stabilizer} $R_\lambda \leq \mathfrak{S}_n$ consists of permutations that (by acting on the canonical tableau) permute boxes within each row, while the \textit{column stabilizer} $C_\lambda \leq \mathfrak{S}_n$ permutes boxes within each column. A $\lambda$-\textit{tableau} is a filling $\tau: \{1,\ldots,n\} \to [d]$ of the boxes with symbols from $[d]$. For example, the \textit{constant tableau} $\tau_\lambda$ is the tableau whose $i$-th row is filled with the number $i$.

For a permutation $\pi \in \mathfrak{S}_n$, we define $(\pi \tau)(k) = \tau(\pi^{-1}(k))$ for all $k \in \{1,..,n\}$, i.e. we permute the tableaux by permuting the positions of the boxes. Two tableaux $\tau, \tau'$ are called \textit{row-equivalent}, written $\tau \sim \tau'$, if one is obtained from the other by permuting entries within each row, i.e. $\pi \in R_\lambda$ s.t. $\tau' = \pi \tau$. A tableau is called \textit{standard} if the entries in each row and each column are increasing. A tableau is \textit{semistandard} (SSYT) if its entries are non-decreasing along each row and strictly increasing up each column. We denote by $\mathcal{T}_{\lambda,d}$ the set of SSYT of shape $\lambda$ with entries in $[d]$. Their size $|\mathcal{T}_{\lambda,d}|$ can be computed with a combinatorial expression called the \textit{hook-content formula} \cite{Stanley1971PlanePartitions2}, and satisfies the polynomial bound $  |\mathcal{T}_{\lambda, d}|  \leq (n+1)^{\frac{d(d-1)}{2}}$ for all $\lambda \in \text{Par}(d,n)$. Similarly, the number of standard Young tableaux is obtained via the \textit{hook-length formula} \cite{Frame1954} and is denoted by:\begin{equation}
    f_\lambda \equiv |\mathrm{SYT}(\lambda)|.
\end{equation}
Tableaux are in natural correspondence with multi-indices:\begin{equation}
\label{eq:multi-indices-tableaux}
    \underline{i} = (i_1, \cdots, i_n) \leftrightarrow \tau (k) = i_k.
\end{equation}
Given a count matrix $E \in \mathbb{N}^{d\times d}$, we can associate to it a partition $\lambda_E \in \mathrm{Par}(d,n)$ by summing the columns of $E$ and sorting the result in non-increasing order. We can also associate a semistandard tableau $\tau_E \in \mathcal{T}_{\lambda_E,d}$, which for every column $b$ of $E$, places symbol $a\in [d] :$ exactly $E_{a,b}$ times in that column, arranged in increasing order from bottom to top. Using the identification in \eqref{eq:multi-indices-tableaux},  we have $E^{(\tau_E, \tau_{\lambda_E})} = E$, where $\tau_{\lambda_E}$ is the constant tableau of shape $\lambda_E$.

We now define the representative matrix set for the action of $\mathfrak{S}_n$ on $\mathcal{H}^{\otimes n}$.
For each $\tau \in \mathcal{T}_{\lambda,d}$, the \textit{Young-symmetrized vector} $|u_\tau\rangle \in \mathcal{H}^{\otimes n}$ is defined as:
\begin{equation}
    |u_\tau\rangle = \sum_{\tau' \sim \tau} \sum_{c \in C_\lambda} \mathrm{sgn}(c) \bigotimes_{k=1}^{n} \ket{\tau'(c(k))}
\end{equation}
where $ \ket{\tau'(c(k))} \coloneqq |i_{\tau'(c(k))}\rangle$. For the constant tableau $\tau_λ$ of shape $λ$ we denote:\begin{equation}
\label{eq:young_vector_trivial}
    \ket{u_{\lambda}} \coloneqq \ket{u_{\tau_\lambda}} = \sum_{c \in C_\lambda} \text{sgn}(c) \bigotimes_{k=1}^{n} \ket{\tau_\lambda(c(k))}
\end{equation} Given a count matrix $E$, for all $\lambda \in \mathrm{Par}(d,n)$, one has \cite[Lemma 1]{Gijswijt2009}:\begin{equation}
        C_E\ket{u_{\lambda}} = \begin{cases}
            \ket{u_{\tau_E}} &\text{if}\ \ \lambda = \lambda_E,\\
            0 &\text{otherwise}.
        \end{cases}
    \end{equation}
The collection $\{|u_\tau\rangle\}_{\tau \in \mathcal{T}_{\lambda,d}}$ forms a basis, called the \textit{Young basis}, for a $\mathfrak{S}_n$-submodule isomorphic to the $\lambda$-isotypic component of $\mathcal{H}^{\otimes n}$ (the isotypic decomposition is a particular decomposition of $\HS\n$ into different isomorphism classes of representations), with multiplicity given by the number of standard Young tableaux of shape $λ$ with entries in $[d]$.  Collecting these vectors as columns of a matrix $U_\lambda = [|u_\tau\rangle : \tau \in \mathcal{T}_{\lambda,d}]$, the set $\{U_\lambda\}_{\lambda \in \mathrm{Par}(d,n)}$ is a \textit{representative set} for the action of $\mathfrak{S}_n$ on $\mathcal{H}^{\otimes n}$, and the blocks of the isomorphism are:
\begin{equation}
\begin{aligned}
    \psi:
    \mathrm{End}^{\mathfrak{S}_n}(S^{\otimes n})
    &\to \bigoplus_{\lambda \in \mathrm{Par}(d_S, n)}
    \mathbb{C}^{m_\lambda \times m_\lambda} \\
\label{eq:incomplete-isomorphism}
    X &\mapsto \bigoplus_\lambda
    U_\lambda^T X U_\lambda
    = \bigoplus_\lambda
    \left(\langle u_\tau | X | u_\gamma \rangle\right)_{\tau, \gamma \in \mathcal{T}_{\lambda,d}}.
\end{aligned}
\end{equation}
Given $\lambda \in \mathrm{Par}(d, n)$, the Young-symmetrized vectors $\{|u_\tau\rangle\}_{\tau \in \mathcal{T}_{\lambda, d}}$ are generally \textit{not} orthogonal. The map in \eqref{eq:incomplete-isomorphism} is a bijective linear map that preserves positive semi-definiteness, which suffices for reformulating PSD constraints. However, it is \textit{not} (yet) a $*$-algebra isomorphism; to obtain the full $*$-isomorphism~(\ref{eq:schur_weyl_iso}) one defines first the \textit{Gram matrix} $G_\lambda = U_\lambda^T U_\lambda$ with entries $(G_\lambda)_{(\tau,\gamma)} = \langle u_\tau | u_\gamma \rangle$, and then takes the Cholesky decomposition of its inverse $R_\lambda R_\lambda^T = G_\lambda^{-1}$. Then, one can construct the map $\widetilde{\psi}$ as:
\begin{equation}
\begin{aligned}
    \widetilde{\psi}:
    \mathrm{End}^{\mathfrak{S}_n}(S^{\otimes n})
    &\to \bigoplus_{\lambda \in \mathrm{Par}(d_S, n)}
    \mathbb{C}^{m_\lambda \times m_\lambda} \\
\label{eq:full_isomorphism}
    X &\mapsto \bigoplus_\lambda
    R_\lambda^T\,
    [\widetilde{\psi}(X)]_\lambda\,
    R_\lambda.
\end{aligned}
\end{equation}The map $\widetilde{\psi}$ is a $*$-isomorphism, i.e. it preserves both positive semi-definiteness and additionally it preserves products ($[\widetilde{\psi}(XY)]_\lambda = [\widetilde{\psi}(X)]_\lambda \, [\widetilde{\psi}(Y)]_\lambda$) and maps the identity to the identity
($[\widetilde{\psi}(\mathbbm{1}_{\mathcal{H}^{\otimes n}})]_\lambda = \mathbbm{1}_{m_\lambda}$). While the Gram correction is unnecessary for PSD-preservation, it becomes essential when we want to phrase a whole SDP in the block basis; this is typically advantageous in practice, and it becomes essential in the application of Section \ref{sec:symmetric_seesaw}, where we reformulate the power iteration method in the symmetric subspace. Throughout this work, a tilde is used to distinguish the full isomorphism~\eqref{eq:full_isomorphism} from the positive map~\eqref{eq:incomplete-isomorphism}. 
\

The construction of both $\psi$ and $\widetilde{\psi}$ trivially extends to $\mathrm{End}^{\mathfrak{S}_n}(R \otimes S^n)$, where there is a fixed reference system $R$ on which the group $\mathfrak{S}_n$ does not act. In this case, the representative set for $R \otimes S^{\otimes n}$ is $\{\mathbbm{1}_R \otimes U^S_\lambda\}_{\lambda \in \mathrm{Par}(d_S,n)}$,  and the following map is a bijective linear map that preserves positive semi-definiteness:
    \begin{equation}
\begin{aligned}
    \psi_{RS^n}:
    \mathrm{End}^{\mathfrak{S}_n}(R \otimes S^{\otimes n})
    &\to \bigoplus_{\lambda \in \mathrm{Par}(d_S, n)}
    \mathbb{C}^{m_\lambda^{RS^n} \times m_\lambda^{RS^n}} \\
\label{eq:isomorphism_with_reference}
    X_{RS^n} &\mapsto \bigoplus_{\lambda \in \mathrm{Par}(d_S, n)}
    (\mathbbm{1}_R \otimes (U^S_\lambda)^T)\,
    X_{RS^n}\,
    (\mathbbm{1}_R \otimes U^S_\lambda).
\end{aligned}
\end{equation}
with the polynomial bounds:
\begin{align}
     |\mathrm{Par}(d_S, n)| &\leq (n+1)^{d_S}, \\
    m_\lambda^{RS^n} = d_{R}\cdot |\mathcal{T}_{\lambda, d_{S}}| &\leq d_{R}\cdot (n+1)^{d_S(d_S-1)/2},\\
    \dim\, \mathrm{End}^{\mathfrak{S}_n}(R \otimes S^{\otimes n}) &\leq d_R^2 \cdot (n+1)^{d^2_S}.
\end{align}
Similarly, the full $*$-algebra isomorphism extends to:
\begin{equation}
\begin{aligned}
    \widetilde{\psi}_{RS^n}:
    \mathrm{End}^{\mathfrak{S}_n}(R \otimes S^{\otimes n})
    &\to \bigoplus_{\lambda \in \mathrm{Par}(d_S, n)}
    \mathbb{C}^{(d_R m_\lambda) \times (d_R m_\lambda)} \\
\label{eq:full_isomorphism_with_reference}
    X_{RS^n} &\mapsto \bigoplus_\lambda
    (\mathbbm{1}_R \otimes R_\lambda^T)\,
    [\psi_{RS^n}(X)]_\lambda\,
    (\mathbbm{1}_R \otimes R_\lambda).
\end{aligned}
\end{equation}
It remains to show that \eqref{eq:incomplete-isomorphism} can be computed efficiently. This automatically implies that also \eqref{eq:full_isomorphism} and their counterparts in \eqref{eq:isomorphism_with_reference} and \eqref{eq:full_isomorphism_with_reference} can be efficiently computed, because they are related to the original map $\psi$ by matrix multiplications by matrices of size polynomial in $n$ and tensor products with matrices on a finite-dimensional system $R$ of fixed dimension $d_R \in \mathbbm{N}$. This is a fairly standard result (see e.g.~\cite{Gijswijt2009, Litjens2016}) and is revised in the following Theorem.

\

Before stating the result, recall that the \textit{dual space} of an Hilbert space $\mathcal{H}^*$ is the vector space of all linear transformations $\varphi: \mathcal{H} \to \mathbb{C}$. The coordinate ring of $\mathcal{H}$, denoted as $O(\mathcal{H})$, is the algebra consisting of all complex-linear combinations of monomials, i.e. products of elements from $\mathcal{H}^*$, also called \textit{polynomials}. A polynomial $p \in O(\mathcal{H})$ is called \textit{homogeneous} if it is a complex-linear combination of a product of $n$ non-constant elements of $\mathcal{H}^*$ (for a fixed non-negative integer $n$), or phrased differently, if all non-zero terms have the same degree $n$. We denote by $O_n(\mathcal{H})$ the set of all homogeneous polynomials of degree $n$.

\begin{theorem}
\label{thm:Efficient_construction_Isomorphism}
    Let $S$ be a Hilbert space of dimension $d_S$, and let $n \in \mathbb{N}$. Then, for every $t \in [m_S]$, $\psi(C_t^S)$, which is the image of the orbit-basis matrix under the block-diagonalization $\psi$, can be computed in time $\mathrm{poly}(n)$.
\end{theorem}
\begin{proof}
We need to show that for all $\lambda \in \mathrm{Par}(d_S,n)$, the blocks $U_\lambda^T C^S_t U_\lambda$ can be computed in $\mathrm{poly}(n)$ time. This boils down to efficiently computing $\braket{u_{\tau} }{C^{S}_t u_\gamma}$, for every $\tau, \gamma \in \mathcal{T}_{\lambda, d_S}$, as:\begin{equation}
\label{eq:complete_matrix_after_isomorphism}
    U_\lambda^T C_t^{S} U_\lambda = \begin{pmatrix}
        \braket{u_1}{C_t u_1} & \braket{u_1}{C_t u_2} & \cdots & \cdots & \braket{u_1}{C_t u_{m_\lambda}} \\
                \braket{u_2} {C_t u_1} & \braket{u_2} {C_t u_2} & \cdots & \cdots & \braket{u_2}{C_t u_{m_\lambda}}\\
                \vdots & \ddots & \cdots & \cdots & \vdots \\
                  \braket{u_{m_\lambda}} {C_t u_1} &  \braket{u_{m_\lambda}} {C_t u_2} & \cdots & \cdots &  \braket{u_{m_\lambda}} { C_t u_{m_\lambda}}\\
    \end{pmatrix}.
\end{equation}
All $C_t, \ket{u_\tau}, \ket{u_\gamma}$ have exponential size in $n$, so we need to compute these inner products without constructing them explicitly.
The key idea~\cite{Gijswijt2009, Litjens2016} is to encode all matrix elements of~\eqref{eq:complete_matrix_after_isomorphism} simultaneously as coefficients of a homogeneous polynomial. For any $A \in \mathrm{End}(S^{\otimes n})$ and $L \in \mathcal{L}(S) \simeq \mathbb{C}^{d_S \times d_S}$, define the polynomial:
\begin{equation}
\label{eq:encoding_polynomial}
    f_A(L) \coloneqq \sum_E \mathrm{Tr}(A^\dagger C_E) \cdot x_E(L),
\end{equation}
where for a given count matrix $E$, we defined the associated monomial $x_E$ as the functional acting on an operator $X$ as follows: 
\begin{equation}
\label{eq:monomial_count_matrix}
\begin{aligned}
x_E: \mathbb{C}^{{d_S} \times d_S} &\to \mathbb{C}\\
    L &\mapsto \prod_{a,b=1}^{d_{S}} L_{(a,b)}^{E_{(a,b)}},
\end{aligned}
\end{equation}
and the sum in \eqref{eq:encoding_polynomial} is over all count matrices $E$ satisfying~\eqref{eq:count_matrix}.
Note that $f_A \in O_n(S)$, i.e. it is a homogeneous polynomial of degree $n$ in the $d_S^2$ entries of $L$, with at most $m_S$ monomials, each labeled by a count matrix. The inner products we need are recovered as coefficients:
\begin{equation}
    \mathrm{Tr}(A^\dagger C_E) = [x_E]\, f_A,
\end{equation}
where $[x_E]\, f_A$ denotes the coefficient of monomial $x_E$ in the polynomial $f_A$.

To obtain the elements of ~\eqref{eq:complete_matrix_after_isomorphism}, for all $\tau, \gamma \in \mathcal{T}_{\lambda, d_S}$, we set $A = |u_\tau\rangle\langle u_\gamma|$ and write $f_{\tau,\gamma} \coloneqq f_{|u_\tau\rangle\langle u_\gamma|}$, so that: \begin{equation}
\label{eq:encoding_polynomial_semistandard-tab}
\begin{aligned}
    f_{\tau, \gamma}: \mathcal{L}(S)& \to \mathbb{C}\\
       L &\mapsto  f_{\tau, \gamma}(L) = \sum_{E} \langle u_\tau | C_E  u_\gamma \rangle \cdot x_E(L),
\end{aligned}
    \end{equation}
that is:
\begin{equation}
    \langle u_\tau | C_E  u_\gamma \rangle = [x_E]\, f_{\tau,\gamma}.
\end{equation}
The polynomial $f_{\tau,\gamma}$ admits a closed-form expression that can be evaluated in $\mathrm{poly}(n)$ time for fixed $d_S$.
In Appendix \ref{app:two_solutions_representation_theory} we present two methods available from the literature to efficiently compute $f_{\tau,\gamma}$. Both  produce the same polynomial but via different approaches:
\begin{itemize}
    \item \textit{Method 1 \cite{Litjens2016, Fawzi2022}} derives a direct formula for $f_{\tau,\gamma}$ valid for any $\tau, \gamma$ using \textit{count functions}.
    \item \textit{Method 2 \cite{Gijswijt2009}} gives a simple product formula for the special case $\tau = \gamma = \tau_\lambda$ (constant tableau), then extends to general tableaux using transition operators.
\end{itemize}
For all partitions $\lambda \in \mathrm{Par}(d,n)$, either method is used a number of times equal to $m_\lambda \times m_\lambda$ (polynomial in $n$), so the computation of \eqref{eq:complete_matrix_after_isomorphism} is efficient.\end{proof}

\begin{corollary}
\label{cor:Efficient-PSD-Constraints}
    Let $S$ be a Hilbert space of dimension $d_S$, and let $n \in \mathbb{N}$. Then, for all variables $X = \sum_{r = 1}^{m_S} x_r C^S_r \in \mathrm{End}^{\mathfrak{S}_n}(S^n)$, the PSD constraint $X \geq 0$ can be imposed with a $\mathrm{poly}(n)$ overhead as a set of $|\mathrm{Par}(d_S,n)| = O(n^{d_S})$ PSD constraints on matrices of size $m_\lambda = O(n^{d^2_S})$ involving the coefficients $x_r$, where $|\mathrm{Par}(d_S,n)|$ and $m_\lambda$ are defined in \eqref{eq:bound_SSYTs}.
\end{corollary}
\begin{proof}
    Using the positive map $\psi_{S^n}$ of \eqref{eq:incomplete-isomorphism}, the operator $X = \sum_{r = 1}^{m_S} x_r C_r^S$ is mapped into block-diagonal form, and $X \geq 0$ is equivalent to imposing $[\psi(X)]_\lambda \geq 0$ for all $\lambda \in \operatorname{Par}(d_S, n)$, which gives $|\mathrm{Par}(d_S,n)| = O((n+1)^{d_S})$ independent constraints, each of polynomial size $m^S_\lambda$, where $m^S_\lambda = O\left((n+1)^{d_S^2}\right)$ by \eqref{eq:bound_SSYTs}. By linearity, using \eqref{eq:incomplete-isomorphism}, the action of $\psi_{S^n}$ can be written as
\begin{align}
\label{eq:incomplete-isomorphism-action-orbit}
\psi(X) = \sum_{r = 1}^{m_S} x_r\bigoplus_\lambda[\psi(C_r^S)]_\lambda, \quad \quad  [\psi(C_r^S)]_\lambda = U_\lambda^T C_r^S U_\lambda = \left(\langle u_\tau | C_r^S | u_\gamma \rangle\right)_{\tau, \gamma \in \mathcal{T}_{\lambda,d}}.
\end{align}
By Theorem \ref{thm:Efficient_construction_Isomorphism}, the action of $\psi$ on a generic orbit matrix $C_r$ can be computed in $\mathrm{poly}(n)$ time. Doing this for every $r \in [m_S]$ implements the change-of-basis map from orbits to irreducible representations (\textit{block basis}) and concludes the proof.
\end{proof}

\subsection{Efficient Operations in the Block Basis}
\label{sec:efficient-operations-in-block-basis}
Corollary \ref{cor:Efficient-PSD-Constraints} allows efficiently phrasing SDPs involving permutation invariant variables, by taking the orbit basis coefficients $\{x_t\}_{t \in m_S}$ as free variables and then imposing the PSD constraints as in Corollary \ref{cor:Efficient-PSD-Constraints}. While having coefficients in the orbit basis is very convenient when computing traces, transposes, HS norms and channel concatenations (see Section \ref{sec:orbit_basis_operations}), in practice it is often convenient to also phrase the SDPs directly in the block basis, so to avoid dealing with the complicated orbit-block relations at the solver level. For this purpose, we can use the full-isomorphism \eqref{eq:full_isomorphism} (or \eqref{eq:full_isomorphism_with_reference}) and directly use the $m_\lambda \times m_\lambda$ blocks $[\widetilde{\psi}_{S^n}(X_{S^n})]_{\lambda}$ as optimizing variables rather than the orbit basis coefficients $\{x_t\}_{t \in m_S}$. 

Recall that full Schur--Weyl decomposition of $S^{\otimes n}$ reads:\begin{equation}
\label{eq:Shur-Weyl-decomposition}
    S^{\otimes n} \cong \bigoplus_{\lambda \in \mathrm{Par}(d_S,n)} V_\lambda \otimes W_\lambda,
\end{equation} where $\dim V_\lambda = m_\lambda$ and $\dim W_\lambda = f_\lambda$. Permutations act non-trivially only on the irreps $V_\lambda$. Then, if $X_{S^n},Y_{S^n} \in \mathrm{End}^{\mathfrak{S}_n}(S^{\otimes n})$, one has\begin{align}
\label{eq:trace_block}
    \Tr[X_{S^n}] &= \sum_{\lambda \in \mathrm{Par}(d_S,n)}f_\lambda \Tr_{V_\lambda}[X_{S^n}] = f_\lambda \Tr[\widetilde{\psi}_{S^n}(X_{S^n})]_{\lambda}\\
     \mathrm{Tr}[X_{S^n}^\dag
   Y_{S^n}]
    &= \sum_{\lambda \in \mathrm{Par}(d_S,n)}
    f_\lambda\, \mathrm{Tr}_{V_\lambda}\!\left[\widetilde{\psi}_{S^n}(X_{S^n}^\dag
   Y_{S^n})\right]_{\lambda} = \sum_{\lambda \in \mathrm{Par}(d_S,n)}
    f_\lambda\,\Tr[
    [\widetilde{\psi}(X_{S^n}^\dag)]_\lambda \cdot
[ \widetilde{\psi}(Y_{S^n})]_\lambda],
\label{eq:inner_product_HS_block}
\end{align}
where we used that \eqref{eq:full_isomorphism} is a $*$-isomorphism, and $f_\lambda = |\mathrm{SYT}(\lambda)|$ is the dimension of the Specht module from the Schur--Weyl decomposition~\eqref{eq:Shur-Weyl-decomposition}.
Formula \eqref{eq:inner_product_HS_block} is useful for expressing objective functions in SDPs (recalling the the primal problem in \eqref{eq:SDP_definition}), and can be extended to easily implement partial trace constraints. For instance, one can conveniently express the constraints for $X_{RS^n} \in \mathcal{C}(R:S^n)$, the Choi set of valid Choi matrices of CPTP maps in \eqref{eq:choi_set}, or $X_{RS^n} \in \mathcal{C}^*(R:S^n)$, the Choi set of valid Choi matrices of CPU maps in \eqref{eq:choi_adj_set}.
Specifically\begin{align}
\label{eq:choi_constraint_block_CPU}
\mathrm{Tr}_R[X_{RS^n}] &= \mathbbm{1}_{S^n}   \iff \Tr_R  \left[[\widetilde{\psi}_{RS^n}(X_{RS^n})]_\lambda \right] =  \mathbbm{1}_{m_\lambda}
    \qquad \forall\, \lambda \in \mathrm{Par}(d_S, n),\\
\label{eq:choi_constraint_block_CPTP}
    \mathrm{Tr}_{S^n}[X_{RS^n}] &= \mathbbm{1}_{R}  \iff \sum_{\lambda \in \mathrm{Par}(d_S,n)} f_\lambda\,
    \mathrm{Tr}_{V_\lambda}\!\left[[
   \widetilde{\psi}_{RS^n}(X_{RS^n})]_\lambda\right]
    = \mathbbm{1}_R.
\end{align}
Note that \eqref{eq:choi_constraint_block_CPU} decouples block by block, while \eqref{eq:choi_constraint_block_CPTP} couples all blocks into a single $d_R \times d_R$ equation, where the contributions are weighted by the Specht dimensions $f_\lambda$. Hence these conditions can also easily be expressed when considering the $[
   \widetilde{\psi}_{RS^n}(X_{RS^n})]_\lambda$ as variables instead of the coefficients $x_{k,l,t}$ in $X_{RS^n} = \sum_{k,l,t} \ket{k}\bra{l}_R \otimes C_t^{S}$.

Note for instance that using the orbit basis coefficients and the map \eqref{eq:incomplete-isomorphism}, the constraint in \eqref{eq:choi_constraint_block_CPU} would become instead:\begin{equation}
\label{eq:choi_constraint_orbit_CPU}
\mathrm{Tr}_R[X_{RS^n}] = \mathbbm{1}_{S^n}   \iff \sum_{k=1}^{d_R} \sum_{t=1}^{m_S} x_{k,k,t} U_\lambda^T C_t^{B} U_\lambda= [|\psi_{B^n}(\mathbbm{1}_{S^n})|]_\lambda \quad  \forall \lambda \in \text{Par}(d_S,n),
\end{equation}
where the left-hand side can be computed efficiently using the construction of Theorem \ref{thm:Efficient_construction_Isomorphism}, while the right-hand side is given by \begin{equation}
\label{eq:gram_matrix_as_action_on_identity}
    [|\psi_{B^n}(\mathbbm{1}_{B^n})|]_\lambda  = U_\lambda^T \left( \sum_{t: E_t \text{ diagonal}} C^B_t \right) U_\lambda = U_\lambda^T U_\lambda =  G_\lambda,
\end{equation}
where we used \eqref{eq:identity_orbit}, and the Gram matrix $G_\lambda$ can be computed efficiently as $(G_\lambda)_{\tau,\gamma} = \langle u_\tau | u_\gamma \rangle = f_{\tau,\gamma}(\mathbbm{1}_{d_S})$, where $f_{\tau,\gamma}$ was defined in \eqref{eq:encoding_polynomial}. We also report an explicit expression of $G_\lambda$ in Appendix \ref{sec:gram_matrix}. This example shows how working in the block basis in SDPs is computationally advantageous, because it allows to impose the PSD constraint 'for free' on the complex block matrices, and makes constraints decouple into independent blocks, like in the case of \eqref{eq:choi_constraint_block_CPU}, instead of implementing the complicated orbit-blocks relations as in \eqref{eq:choi_constraint_orbit_CPU}.
\

In many cases, one also needs to implement concatenations (Proposition~\ref{prop:serial_concatenation}) and (partial) transposes (Proposition~\ref{lem:transpose_orbit_basis}), which are easily formulated using the orbit basis. In these cases, one needs an interconversion of the corresponding constraints between orbit and block basis, which can be implemented with a fixed overhead by precomputing the change-of-basis matrices. For example, in Section \ref{sec:PPT_REE}, we will need to implement the partial transpose $T_{B^n}$ in the block basis, which does not respect the Schur--Weyl block decomposition and in general \emph{mixes blocks} corresponding to different partitions $\lambda$ (while in the orbit basis, it corresponds to a simple permutation of orbit coefficients, see \eqref{eq:coefficient_partial_transpose_relation}). For this reason, in practice one defines the linear map\begin{equation}
\label{eq:partial_transpose_block}
\begin{aligned}
    \widetilde{T_B}: \bigoplus_{\lambda \in \mathrm{Par}(d_{A}d_B, n)} \mathbb{C}^{m_\lambda \times m_\lambda} &\to \bigoplus_{\lambda \in \mathrm{Par}(d_{A}d_B, n)} \mathbb{C}^{m_\lambda \times m_\lambda}\\
    \oplus_{\lambda} [\widetilde{\psi}(X)]_{\lambda} &\mapsto \left(\widetilde{\psi} \circ T_B \circ \widetilde{\psi}^{-1}\right) (\oplus_{\lambda} [\widetilde{\psi}(X)]_{\lambda}).
\end{aligned}
\end{equation}
where $T_B: \esn((AB)^n) \to \esn((AB)^n) $ is the partial transpose, by considering the elements as matrices in $\mathcal{L}((\HS_A\otimes \HS_B)\n)$.
In practice, the map in \eqref{eq:partial_transpose_block} can be precomputed once (in $\mathrm{poly}(n)$ time) and stored as a sparse matrix.

\subsection{General Matrix Algebras -- Exploiting Additional Quantum-Classical Structure}
\label{sec:cq-algebras}
In previous sections we have shown how to efficiently analyze and optimize over the space $\esn(\HS\n) \equiv \esn((\mathbb{C}^{d_{\HS} \times d_{\HS}})\n)$. While $\mathbb{C}^{d_{\HS} \times d_{\HS}}$ is a matrix $*$-algebra, not every $*$-algebra is necessarily of this form. In finite dimensions though, every $*$-algebra is composed of direct sums of blocks of this form, i.e.\ every $*$-algebra $\A$ can be written as (this can be seen for example as a consequence of the Wedderburn–Artin theorem)
\begin{equation}
\A = \bigoplus_{i=1}^\ell \mathbb{C}^{d_i \times d_i}.    
\end{equation}
This section will deal with studying $\esn(\A\n)$ for such general $*$-algebras, in particular we will show that all of our previous results can be extended to this setting.

\begin{remark}
Physically, a block-diagonal algebra (like $\A$ above) imposes a classical-quantum structure. Specifically, a state $\rho = \bigoplus_i \rho_i \in \A$ (that also satisfies $\rho \geq 0$ and $\Tr(\rho) = 1$) can be thought of as a classical probability distribution $p_i = \Tr(\rho_i)$ together with a quantum state $\rho_i/p_i \in \mathcal{D}(\mathbb{C}^{d_i \times d_i})$ for each outcome:
\begin{equation}
\label{eq:CQ_state}
    \rho = \sum_{i = 1}^\ell p_i \ket{i}\!\bra{i}\otimes \rho_i.
\end{equation}
\end{remark}

Since the dimension of $\A$ is smaller than the dimension of the non-block-diagonal matrices $\mathbb{C}^{d \times d}$, where $d = \sum_i d_i$, this will be the case even more so for $\esn(\A\n)$ compared to $\esn((\mathbb{C}^{d \times d})\n)$ (where the single-copy dimension determines the exponent with which the dimension grows with $n$), and hence exploiting such a classical-quantum structure --- if it is present --- significantly brings down computational complexity.

To understand $\esn(\A\n)$, we follow the approach in \cite[Section 4]{Gijswijt2009}. Note first that
\begin{equation}\label{eq:direct_sum_multiply_out}
\A\n = \left(\bigoplus_{i=1}^\ell \mathbb{C}^{d_i \times d_i}\right)\n = \bigoplus_{i_1=1}^\ell \cdots \bigoplus_{i_n=1}^\ell \bigotimes_{k=1}^n \mathbb{C}^{d_{i_k} \times d_{i_k}} = \bigoplus_{\underline{i} \in [\ell]^n} \bigotimes_{k=1}^n \mathbb{C}^{d_{i_k} \times d_{i_k}}.
\end{equation}
Now, for any element of $\esn(\A\n)$, also the right-hand side of \eqref{eq:direct_sum_multiply_out} has to be invariant under permutations of the $n$ systems (i.e., reorderings of the tensor product $\bigotimes_{k=1}^n$). Let us define the multi-index 
\begin{equation}
\label{eq:def_underline_mu}
    \underline{\mu} \coloneqq (\mu_1, \ldots, \mu_\ell) \in \mathbb{N}^\ell, \qquad |\underline{\mu}| \coloneqq \sum_{j=1}^\ell \mu_j,
\end{equation}
where $\mu_j$ counts how many times the symbol $j \in [\ell]$ appears in the string $\underline{i} \in [\ell]^n$. If a block $j$ appears $\mu_j$ times in $\underline{i}$, then the state on these $\mu_j$ copies of the block has to be an element of $\operatorname{End}^{\mathfrak{S}_{\mu_j}}((\mathbb{C}^{d_{j} \times d_j})^{\otimes \mu_j})$. Moreover, applying permutations to a string $\underline{i}$ that mix systems coming from different (original) blocks $j$ gives a result on a different string $\underline{i}'$. For the total direct sum to stay invariant, this means that for two strings $\underline{i}, \underline{i}' \in [\ell]^n$ that contain the same number of each symbol $j \in [\ell]$ (i.e.\ the same multi-index $\underline{\mu}$), their corresponding blocks must be equal up to a permutation of systems.

Removing this \enquote{duplication modulo a permutation} immediately gives the following $*$-isomorphism:
\begin{equation}\label{eq:general_algebra_end_sn_decomposition}
    \esn(\A\n) = \esn\!\left(\!\Bigl(\bigoplus_{i=1}^\ell \mathbb{C}^{d_i \times d_i}\Bigr)^{\otimes n}\!\right) \cong \bigoplus_{\substack{\underline{\mu} \in \mathbb{N}^\ell \\ |\underline{\mu}| = n }} \bigotimes_{j = 1}^\ell \operatorname{End}^{\mathfrak{S}_{\mu_j}}\!\bigl((\mathbb{C}^{d_{j} \times d_j})^{\otimes \mu_j}\bigr).
\end{equation}
Note that this block-diagonal decomposition can in general not be implemented by an isometry; however, based on the explanation above, it is immediate that one can implement by an isometry a mapping where every block on the right-hand side appears as a direct sum with $\binom{n}{\underline{\mu}}$ (i.e.\ the multinomial coefficient of $\underline{\mu}$) identical copies.

Let us now construct the orbit basis of $\esn(\A\n)$. For this, we can consider $\esn(\A\n)$ in \eqref{eq:general_algebra_end_sn_decomposition} as a subset of $\esn((\mathbb{C}^{d \times d})\n)$ where $d = \sum_i d_i$, and then take the orbit basis of the latter as explained above. Now, for a basis element $C_E$ of $\esn((\mathbb{C}^{d \times d})\n)$ corresponding to count matrix $E$, since each entry in the count matrix counts how often the corresponding element appears, we have $C_E \in \esn(\A\n)$ if and only if the count matrix is itself of the block-diagonal form, i.e.\ $E \in \A \cap \mathbb{N}^{d \times d}$. Thus, the orbit basis of $\esn(\A\n)$ contains all orbits corresponding to count matrices with the given block-decomposition. Additionally, on the right-hand side of \eqref{eq:general_algebra_end_sn_decomposition}, we have a natural basis given by tensor products of basis elements $C_{E_1} \otimes \cdots \otimes C_{E_\ell}$, where $C_{E_j} \in \operatorname{End}^{\mathfrak{S}_{\mu_j}}((\mathbb{C}^{d_{j} \times d_j})^{\otimes \mu_j})$ is an orbit basis element corresponding to the count matrix $E_j$ as previously.

\begin{lemma}\label{lem:general_algebra_end_sn_trivial_isomorphism}
The two spaces $\esn\bigl((\bigoplus_{i=1}^\ell \mathbb{C}^{d_i \times d_i})\n\bigr)$ and $\bigoplus_{\underline{\mu} \in \mathbb{N}^\ell:\, |\underline{\mu}| = n} \bigotimes_{j = 1}^\ell \operatorname{End}^{\mathfrak{S}_{\mu_j}}((\mathbb{C}^{d_{j} \times d_j})^{\otimes \mu_j})$ are $*$-isomorphic, and there is a one-to-one mapping between the basis elements of the two bases described above.
\end{lemma}
\begin{proof}
Let $E = \bigoplus_{i=1}^\ell E_i$, with $E_i \in \mathbb{N}^{d_i \times d_i}$, be a count matrix corresponding to a basis element $C_E \in \esn(\A\n)$. Define the multi-index $\underline{\mu}$ via
\begin{equation}
    \mu_i \coloneqq \sum_{a,b=1}^{d_i} (E_i)_{(a,b)}, \qquad i \in [\ell].
\end{equation}
This satisfies $|\underline{\mu}| = \sum_{a,b=1}^{d} E_{(a,b)} = n$, so $\underline{\mu}$ is a valid multi-index in the direct-sum decomposition of the right-hand side. The map
\begin{equation}
\label{eq:basis_element_correspondence}
    \phi: C_E \longmapsto C_{E_1} \otimes C_{E_2} \otimes \cdots \otimes C_{E_\ell} \in \bigotimes_{j=1}^\ell \mathrm{End}^{\mathfrak{S}_{\mu_j}}\bigl((\mathbb{C}^{d_j\times d_j})^{\otimes \mu_j}\bigr)
\end{equation}
sends the basis element $C_E$ to the corresponding tensor product of orbit basis elements in the $\underline{\mu}$-th summand. The inverse is obtained by gluing the per-block count matrices into the block-diagonal $E = \bigoplus_i E_i$, so \eqref{eq:basis_element_correspondence} is a bijection, which we extend by linearity over $\mathbb{C}$. Moreover, $\phi(C_E^\dagger) = \phi(C_{E^T})$ and 
\begin{equation}
    (C_{E_1} \otimes \cdots \otimes C_{E_\ell})^\dagger = C_{E_1^T} \otimes \cdots \otimes C_{E_\ell^T}.
\end{equation}
Since $E^T = \bigoplus_i E_i^T$ for block-diagonal $E$, extending by linearity, $\phi(X^\dagger) = \phi(X)^\dagger$ for all $X \in \esn\bigl((\bigoplus_{i=1}^\ell \mathbb{C}^{d_i \times d_i})\n\bigr)$. Finally, $\phi(XY) = \phi(X)\phi(Y)$ follows from the block-diagonal structure of $\A$: writing $E = \bigoplus_i E_i$ and $E' = \bigoplus_i E'_i$, the product $C_E \cdot C_{E'}$ vanishes unless $\mu_i(E) = \mu_i(E')$ for all $i$ (i.e.\ both elements live in the same $\underline{\mu}$-summand of the right-hand side). This proves that $\phi$ is a $*$-isomorphism.
\end{proof}

By applying the previous block-diagonalization (Theorem \ref{thm:Efficient_construction_Isomorphism}) to each block on the right-hand side of \eqref{eq:general_algebra_end_sn_decomposition}, we obtain the following block-diagonalization of $\esn(\A\n)$. Throughout, given a multi-index $\underline{\mu}$ as in \eqref{eq:def_underline_mu}, we will write
\begin{equation}
\label{eq:def_underline_lambda}
    \underline{\lambda} \coloneqq (\lambda_1, \ldots, \lambda_\ell), \qquad \lambda_j \in \mathrm{Par}(d_j, \mu_j),
\end{equation}
for a tuple of partitions, one for each block of $\A$. We also define the per-block Specht and SSYT dimensions
\begin{equation}
\label{eq:dimensions_underline_lambda}
    f_{\underline{\lambda}} \coloneqq \prod_{j=1}^\ell f_{\lambda_j}, \qquad m_{\underline{\lambda}} \coloneqq \prod_{j=1}^\ell m_{\lambda_j},
\end{equation}
as the products of the corresponding single-block quantities.

\newcommand{\pt}{\widetilde{\psi}}
\newcommand{\Par}{\operatorname{Par}}
\begin{theorem}[{Block-diagonalization for $\esn(\A\n)$, see also \cite[Section 4]{Gijswijt2009}}]
\label{thm:block_diag_general_algebra}
There exists a $*$-isomorphism
\begin{equation}
\label{eq:full_isomorphism_general}
    \pt: \esn(\A\n) \to \bigoplus_{\substack{\underline{\mu} \in \mathbb{N}^\ell \\ |\underline{\mu}| = n}} \bigoplus_{\underline{\lambda}:\, \lambda_j \in \Par(d_j, \mu_j)} \mathbb{C}^{m_{\underline{\lambda}} \times m_{\underline{\lambda}}},
\end{equation}
where $m_{\underline{\lambda}}$ is defined in \eqref{eq:dimensions_underline_lambda}. Moreover, for every orbit basis element $C_E \in \esn(\A\n)$ --- which corresponds to a count matrix $E = \bigoplus_{i=1}^\ell E_i$ as above --- one can compute $\pt(C_E)$ in time polynomial in $n$.
\end{theorem}
\begin{proof}
This follows directly by combining Lemma \ref{lem:general_algebra_end_sn_trivial_isomorphism} with the efficient block-diagonalization of $\esn((\mathbb{C}^{d_j \times d_j})\n)$ established in Theorem \ref{thm:Efficient_construction_Isomorphism}.
\end{proof}
The block-diagonalization of Theorem \ref{thm:block_diag_general_algebra} allows phrasing SDPs over $\esn(\A\n)$ directly in the block basis, by taking the blocks $\widetilde{[X]}_{\underline{\mu}, \underline{\lambda}} \coloneqq [\pt(X)]_{\underline{\mu}, \underline{\lambda}} \in \mathbb{C}^{m_{\underline{\lambda}} \times m_{\underline{\lambda}}}$ as the optimization variables, in direct generalization of Section \ref{sec:efficient-operations-in-block-basis}. The trace and Hilbert--Schmidt inner product extend straightforwardly: for $X, Y \in \esn(\A\n)$,
\begin{align}
\label{eq:trace_block_general}
    \Tr[X] &= \sum_{\underline{\mu}: |\underline{\mu}| = n}\sum_{\underline{\lambda}:\,\lambda_j\in \Par(d_j,\mu_j)} f_{\underline{\lambda}}\, \Tr\!\bigl[\widetilde{[X]}_{\underline{\mu}, \underline{\lambda}}\bigr], \\
\label{eq:HS_inner_block_general}
    \Tr[X^\dagger Y] &= \sum_{\underline{\mu}: |\underline{\mu}| = n}\sum_{\underline{\lambda}:\,\lambda_j\in \Par(d_j,\mu_j)} f_{\underline{\lambda}}\, \Tr\!\bigl[\widetilde{[X]}_{\underline{\mu}, \underline{\lambda}}^\dagger\, \widetilde{[Y]}_{\underline{\mu}, \underline{\lambda}}\bigr],
\end{align}
where $f_{\underline{\lambda}} = \prod_j f_{\lambda_j}$ as in \eqref{eq:dimensions_underline_lambda}. These are direct consequences of the fact that $\pt$ is a $*$-isomorphism and the trace decomposes blockwise with multiplicity $f_{\underline{\lambda}}$.

In the bipartite setting, where $X_{R\,\A^n} \in \mathrm{End}^{\mathfrak{S}_n}(R \otimes \A^{\otimes n})$ for some auxiliary system $R$ of dimension $d_R$, the partial trace over $\A^n$ also admits a block-diagonal expression. Writing $\widetilde{[X]}_{\underline{\mu}, \underline{\lambda}} = \sum_{k,l=1}^{d_R} \ket{k}\!\bra{l}_R \otimes \widetilde{[X]}^{(k,l)}_{\underline{\mu}, \underline{\lambda}}$, one has
\begin{equation}
\label{eq:partial_trace_general_block}
    \Tr_{\A^n}[X_{R\,\A^n}] = \mathbbm{1}_R \quad\Longleftrightarrow\quad \sum_{\underline{\mu}: |\underline{\mu}| = n}\sum_{\underline{\lambda}:\,\lambda_j \in \Par(d_j,\mu_j)} f_{\underline{\lambda}}\, \Tr_{V_{\underline{\lambda}}}\!\bigl[\widetilde{[X]}_{\underline{\mu}, \underline{\lambda}}\bigr] = \mathbbm{1}_R,
\end{equation}
where $\Tr_{V_{\underline{\lambda}}}$ denotes the partial trace over the irrep space $V_{\underline{\lambda}} = \bigotimes_j V_{\lambda_j}$ within each block (acting trivially on the $R$ register). This is the natural generalization of the CPTP/trace-preserving constraint \eqref{eq:choi_constraint_block_CPTP} to a system of general algebra type.

\begin{remark}
In the case where $\ell = 2$ and $d_1 = d_2 = d_S$, one has $\A = \mathbb{C}^{d_1 \times d_1} \oplus \mathbb{C}^{d_2 \times d_2}$, and the multi-index reduces to a pair $\underline{\mu} = (k, n-k)$ with $k \in \{0, 1, \ldots, n\}$, while the partition tuple $\underline{\lambda} = (\lambda_k, \lambda_{n-k})$ ranges over $\Par(d_S, k) \times \Par(d_S, n-k)$. The decomposition \eqref{eq:full_isomorphism_general} becomes
\begin{equation}
\label{eq:binary_block_decomposition}
    \esn(\A\n) \cong \bigoplus_{k=0}^n \bigoplus_{\substack{\lambda_k \in \Par(d_S, k) \\ \lambda_{n-k} \in \Par(d_S, n-k)}} \mathbb{C}^{m_{\lambda_k} m_{\lambda_{n-k}} \times m_{\lambda_k} m_{\lambda_{n-k}}},
\end{equation}
with multiplicity $f_{\underline{\lambda}} = f_{\lambda_k}\, f_{\lambda_{n-k}}$ for each block. This decomposition allows to significantly reduce the involved dimension when working in $\mathrm{End}^{\mathfrak{S}_n}(\A^{\otimes n})$ compared to the naive embedding in $\mathrm{End}^{\mathfrak{S}_n}((\mathbb{C}^{2d_S})^{\otimes n})$. Indeed, while the latter has orbit-basis dimension
\begin{equation}
    \dim \mathrm{End}^{\mathfrak{S}_n}\bigl((\mathbb{C}^{2d_S})^{\otimes n}\bigr) = \binom{n + 4d_S^2 - 1}{n},
\end{equation}
the former has dimension
\begin{equation}
\label{eq:dimension_ell_2}
    \dim \mathrm{End}^{\mathfrak{S}_n}(\A^{\otimes n}) = \sum_{k=0}^n \binom{k + d_S^2 - 1}{k}\binom{n-k + d_S^2 - 1}{n-k} = \binom{n + 2d_S^2 - 1}{n},
\end{equation}
where in the last identity we used the Vandermonde convolution formula~\cite{Vandermonde1772}. The relative dimensional saving thus scales as $O(n^{-2d_S^2})$. As a concrete example, for $d_S = 2$ (a single qubit controlled by a binary classical flag, so that $\A \subset \mathbb{C}^{4 \times 4}$), Table~\ref{tab:CQ_dimension_comparison} reports the comparison of the two dimensions for several values of $n$.
\begin{table}[ht]
\centering
\renewcommand{\arraystretch}{1.25}
\begin{tabular}{r r r r}
\toprule
$n$ & $\dim \mathrm{End}^{\mathfrak{S}_n}((\mathbb{C}^{4})^{\otimes n}) = \binom{n+15}{n}$ & $\dim \mathrm{End}^{\mathfrak{S}_n}(\A^{\otimes n}) = \binom{n+7}{n}$ & ratio (\%) \\
\midrule
 2 &         136 &      36 & $26.5\%$ \\
 3 &         816 &     120 & $14.7\%$ \\
 4 &       3{,}876 &     330 & $8.5\%$ \\
 5 &      15{,}504 &     792 & $5.1\%$ \\
 6 &      54{,}264 &   1{,}716 & $3.2\%$ \\
 7 &     170{,}544 &   3{,}432 & $2.0\%$ \\
 8 &     490{,}314 &   6{,}435 & $1.3\%$ \\
\bottomrule
\end{tabular}
\caption{Comparison of orbit-basis dimensions for $d_S = 2$ (qubit-flagged algebra $\A = \mathbb{C}^{2\times 2}\oplus \mathbb{C}^{2\times 2}$, embedded in $\mathbb{C}^{4\times 4}$). The ratio decreases as $O(n^{-8})$, reflecting the asymptotic improvement from $O(n^{15})$ to $O(n^{7})$.}
\label{tab:CQ_dimension_comparison}
\end{table}
\end{remark}

Since the orbit basis we chose for $\esn(\A\n)$ is a subset of the orbit basis for $\esn((\mathbb{C}^{d \times d})\n)$, all the properties of such basis matrices $C_E$ we showed earlier (e.g.\ Lemma \ref{lemma:trace_HS_orbit}, Lemma \ref{lem:transpose_orbit_basis}) directly apply to the general setting.
We can also use this to lift Proposition \ref{prop:serial_concatenation} to the general setting. For this we need to consider a bipartite setup as follows:
\begin{align}
    \A_A &= \bigoplus_{i = 1}^{\ell_A}\mathbb{C}^{d_i^{(A)} \times d_i^{(A)}} \subset \mathcal{L}\bigl(\HS_A = \mathbb{C}^{\sum_i d_i^{(A)}}\bigr),\\
    \A_B &= \bigoplus_{j = 1}^{\ell_B}\mathbb{C}^{d_j^{(B)} \times d_j^{(B)}} \subset \mathcal{L}\bigl(\HS_B = \mathbb{C}^{\sum_j d_j^{(B)}}\bigr),\\
    \A_{AB} &= \A_A \otimes \A_B = \bigoplus_{i = 1}^{\ell_A} \bigoplus_{j = 1}^{\ell_B} \mathbb{C}^{d_{i}^{(A)}d_{j}^{(B)} \times d_{i}^{(A)}d_{j}^{(B)}} \subset \mathcal{L}(\HS_A \otimes \HS_B).
\end{align}
We will also use the notation $\HS_{A^{(i)}} \coloneqq \mathbb{C}^{d_i^{(A)}}$ and similarly for $\HS_{B^{(j)}}$. As described above, the orbit basis matrices of these spaces are associated to block-diagonal count matrices. Thus, for any orbit basis matrix $C^{AB}_{s} \in \esn(\A_{AB}\n)$ associated with count matrix $E_s$, we have the decomposition:\begin{equation}
    E_s = \bigoplus_{(i,j) \in [\ell_A] \times [\ell_B]} E_{s_{ij}}.
\end{equation} Let $\mu_{ij}$ be the sum of all entries of $E_{s_{ij}}$ (hence each $E_{s_{ij}}$ corresponds to an orbit basis matrix $C^{A^{(i)}B^{(j)}}_{s_{ij}} \in \operatorname{End}^{\mathfrak{S}_{\mu_{ij}}}((\HS_{A^{(i)}} \otimes \HS_{B^{(j)}})^{\otimes \mu_{ij}})$), and define $\mu^{A}_i \coloneqq \sum_{j=1}^{\ell_B}\mu_{ij}$ and $\mu^{B}_j \coloneqq \sum_{i=1}^{\ell_A}\mu_{ij}$.

Similarly to Proposition \ref{prop:serial_concatenation} we can then again consider the \enquote{reduced count matrices}:
\begin{alignat}{3}
    (E_{r_{ij}(s_{ij})})_{a_A, b_A} &\coloneqq \sum_{a_B, b_B = 1}^{d_{j}^{(B)}} (E_{s_{ij}})_{(a_A, a_B), (b_A, b_B)} \qquad &&\text{ corresponding to } C^{A^{(i)}}_{r_{ij}} \in \operatorname{End}^{\mathfrak{S}_{\mu_{ij}}}(\HS_{A^{(i)}}^{\otimes \mu_{ij}}) \\
    (E_{r_{i}(s)})_{a_A, b_A} &\coloneqq \sum_{j = 1}^{\ell_B} (E_{r_{ij}(s_{ij})})_{a_A, b_A} \qquad &&\text{ corresponding to } C^{A^{(i)}}_{r_i} \in \operatorname{End}^{\mathfrak{S}_{\mu^{A}_i}}(\HS_{A^{(i)}}^{\otimes \mu^{A}_{i}}) \\
    (E_{t_{ij}(s_{ij})})_{a_B, b_B} &\coloneqq \sum_{a_A, b_A = 1}^{d_{i}^{(A)}} (E_{s_{ij}})_{(a_A, a_B), (b_A, b_B)} \qquad &&\text{ corresponding to } C^{B^{(j)}}_{t_{ij}} \in \operatorname{End}^{\mathfrak{S}_{\mu_{ij}}}(\HS_{B^{(j)}}^{\otimes \mu_{ij}}) \\
    (E_{t_{j}(s)})_{a_B, b_B} &\coloneqq \sum_{i = 1}^{\ell_A} (E_{t_{ij}(s_{ij})})_{a_B, b_B} \qquad &&\text{ corresponding to } C^{B^{(j)}}_{t_j} \in \operatorname{End}^{\mathfrak{S}_{\mu^{B}_j}}(\HS_{B^{(j)}}^{\otimes \mu^{B}_{j}})
\end{alignat}
Then $E_r = \bigoplus_{i = 1}^{\ell_A} E_{r_i}$ corresponds to the orbit basis matrix $C^{A}_r \in \esn(\A_A\n)$, and $E_t = \bigoplus_{j = 1}^{\ell_B} E_{t_j}$ corresponds to the orbit basis matrix $C^{B}_t \in \esn(\A_B\n)$. We also write $r = r(s)$ and $t = t(s)$ for the corresponding marginal orbits, as in \eqref{eq:A_marginal} and \eqref{eq:B_marginal}.

\begin{proposition}\label{prop:general_serial_contenation}
Let $\A_A, \A_B, \A_{AB}$ be as in the text above this proposition. By considering the natural inclusion $\esn(\A_{AB}\n) \subset \esn((\HS_A \otimes \HS_B)\n) \subset \mathcal{L}((\HS_A \otimes \HS_B)\n)$ (and similarly for $\esn(\A_{A}\n)$ and $\esn(\A_{B}\n)$) we can define $\mathrm{Tr}_{B^n}\!\bigl[C_s^{AB} \cdot (\mathbbm{1}_{A^n} \otimes (C_{t}^B)^T)\bigr]$. Then, this yields an element in $\esn(\A_{A}\n)$, and in fact the partial-trace map
\begin{equation}
\begin{aligned}
  \mathcal{T}_B\colon \esn(\A_{AB}\n) \times \esn(\A_B\n) &\to \esn(\A_{A}\n)\\
   (C_s^{AB}, C_t^B) &\mapsto \mathrm{Tr}_{B^n}\!\bigl[C_s^{AB} \cdot (\mathbbm{1}_{A^n} \otimes (C_{t}^B)^T)\bigr]
\end{aligned}
\end{equation} 
can again be efficiently computed as
\begin{equation}
\label{eq:partial_trace_identity_general}
 \mathcal{T}_B(C_s^{AB}, C_t^B) = \kappa_s^A\, \delta_{t, t(s)}\, C_{r(s)}^A,
\end{equation}
where 
\begin{equation}
    \kappa_s^A = \prod_{a_A, b_A = 1}^{\sum_i d_i^{(A)}} \binom{(E_{r(s)})_{(a_A,b_A)}}{\{(E_s)_{(a_A,a_B),(b_A,b_B)}\}_{a_B, b_B \in [\sum_j d_j^{(B)}]}}.
\end{equation}
Additionally, $\kappa_s^A$ satisfies the following relation: let $\mathcal{T}_B^{(ij)}$ be the partial trace map from Proposition \ref{prop:serial_concatenation} for each individual block, such that $\mathcal{T}_B^{(ij)}(C^{A^{(i)}B^{(j)}}_{s_{ij}}, C^{B^{(j)}}_{t_{ij}}) = \kappa_{s_{ij}}^{A^{(i)}}\, \delta_{t_{ij}(s_{ij}),t_{ij}}\, C^{A^{(i)}}_{r_{ij}(s_{ij})}$. Then
\begin{equation}
\label{eq:kappa_general_decomposition}
    \kappa_s^{A} = \prod_{i = 1}^{\ell_A} \frac{\mu^A_i!}{\prod_{j = 1}^{\ell_B} \mu_{ij}!}\, \bigl(\|C_{r_i}^{A^{(i)}}\|_{\mathrm{HS}}\bigr)^{-1} \prod_{j=1}^{\ell_B} \|C_{r_{ij}}^{A^{(i)}}\|_{\mathrm{HS}}\, \kappa_{s_{ij}}^{A^{(i)}}.
\end{equation}
\end{proposition}
\begin{proof}
The first statement follows directly from Proposition \ref{prop:serial_concatenation} by considering everything as elements of $\mathcal{L}(\HS_A\n)$, $\mathcal{L}(\HS_B\n)$ and $\mathcal{L}(\HS_{AB}\n)$. For the last statement, note that
\begin{align}
    \kappa_s^A &= \prod_{a_A, b_A = 1}^{\sum_i d_i^{(A)}} \frac{(E_{r(s)})_{a_A, b_A}!}{\prod_{a_B, b_B = 1}^{\sum_j d_j^{(B)}} (E_{s})_{(a_A, a_B),(b_A, b_B)}!} \notag\\
    &= \prod_{i = 1}^{\ell_A} \prod_{a_A, b_A = 1}^{d_i^{(A)}} \frac{(E_{r_i(s)})_{a_A, b_A}!}{\prod_{j = 1}^{\ell_B}\prod_{a_B, b_B = 1}^{d_j^{(B)}} (E_{s_{ij}})_{(a_A, a_B),(b_A, b_B)}!} \\
    &= \prod_{i = 1}^{\ell_A} \prod_{a_A, b_A = 1}^{d_i^{(A)}} \frac{(E_{r_i(s)})_{a_A, b_A}!}{\prod_{j = 1}^{\ell_B} (E_{r_{ij}(s_{ij})})_{a_A, b_A}!} \prod_{j = 1}^{\ell_B} \binom{(E_{r_{ij}(s_{ij})})_{a_A, b_A}}{\{(E_{s_{ij}})_{(a_A, a_B),(b_A, b_B)}\}_{a_B, b_B \in [d_j^{(B)}]}} \\
    &= \prod_{i = 1}^{\ell_A} \frac{\mu^A_i!}{\prod_{j = 1}^{\ell_B} \mu_{ij}!}\, \frac{\prod_{a_A, b_A = 1}^{d_i^{(A)}}(E_{r_{i}(s)})_{a_A, b_A}!}{\mu_i^A!} \prod_{j = 1}^{\ell_B} \frac{\mu_{ij}!}{\prod_{a_A, b_A = 1}^{d_i^{(A)}}(E_{r_{ij}(s_{ij})})_{a_A, b_A}!}\, \kappa_{s_{ij}}^{A^{(i)}}
\end{align}
which gives the desired expression after using \eqref{eq:inner_product_HS}.
\end{proof}
\section{Symmetric Seesaw Iteration Method}
\label{sec:seesaw}
\subsection{Introduction: Entanglement Transmission and Channel Fidelity}
A central problem in quantum information theory is to determine the minimum probability of error incurred when transmitting quantum information over $n$ parallel copies of a quantum channel $\mathcal{N}_{A\to B}$. This problem can be cast as an optimization task, in which one maximizes the overlap with the desired result on the receiving end over all admissible encoding and decoding strategies.

Specifically, the goal of ($n$-shot) \textit{entanglement transmission}~\cite{Schumacher1996, Buscemi2010} is to transmit part of a maximally entangled state of Schmidt rank $d$ through $n$ parallel uses of $\mathcal{N}$, denoted $\mathcal{N}^{\otimes n}$, with the help of encoding and decoding channels $(\mathcal{E}, \mathcal{D})$  (see Figure \ref{fig:quantumcommunicationquantumchannel}). The figure of merit is the \textit{channel fidelity} $F_c(\mathcal{N}^{\otimes n}, d)$, defined below.
\begin{definition}
    Given $n,d \in \mathbb{N}$ and a quantum channel $\mathcal{N}\in \CPTP(A \to B)$, with $A,B$ finite dimensional Hilbert spaces, with $d \leq \min\{d_A, d_B\}$ the \textit{channel fidelity} is defined as the joint optimization problem:\begin{equation}
\label{eq:channel_fidelity}
\begin{aligned}
    F_c(\mathcal{N}^{\otimes n}, d) 
    &\coloneqq \max_{\substack{\mathcal{E} \in \CPTP(A'\to A^n) \\ \mathcal{D} \in \CPTP(B^n\to B')}} 
    F\!\Bigl(\Phi_{RB'}^d, \bigl(\mathrm{id}_R \otimes (\mathcal{D}_{B^n \to B'} \circ \mathcal{N}_{A^n \to B^n} \circ \mathcal{E}_{A' \to A^n})\bigr)(\Phi^d_{RA'})\Bigr).
\end{aligned}
\end{equation}
where $F(\rho, \sigma) \coloneqq \left\|\sqrt{\rho}\sqrt{\sigma}\right\|^2_1$ is the (squared) fidelity as also introduced above.
\end{definition}

Note that a high channel fidelity is a strong requirement, because reliable entanglement transmission implies reliable transmission, on average, of all non-entangled input states \cite{Barnum1998}. The channel fidelity can also be interpreted in terms of the Choi representation:\begin{equation}
\label{eq:reformulation_channel_fidelity}
     F_c(\mathcal{N}^{\otimes n}, d) = \max_{\mathcal{E}_n, \mathcal{D}_n} 
    F\!\Bigl(\Phi^d_{RB'}, \Phi^{\mathcal{D}_n \circ \mathcal{N}^{\otimes n} \circ \mathcal{E}_n}_{RB'}\Bigr).
\end{equation}
The objective function in (\ref{eq:reformulation_channel_fidelity}) is also known in the literature as the entanglement fidelity \cite{Schumacher1996,barnum2000} of the overall channel, defined as:\begin{equation}
\label{eq:entanglement_fidelity}
    F_e(\mathcal{M}) = F\!\Bigl(\Phi^d, \Phi^{\mathcal{M}}\Bigr) = \bra{\Phi^d}\Phi^{\mathcal{M}}\ket{\Phi^d},
\end{equation}
and quantifies the ability of the channel to preserve entanglement as the overlap of its Choi state with a MES (note that $F_e(\mathcal{M})= 1 \iff \Phi^{\mathcal{M}} = \Phi^d \iff \mathcal{M} = \mathrm{id}_d$, the ideal channel).

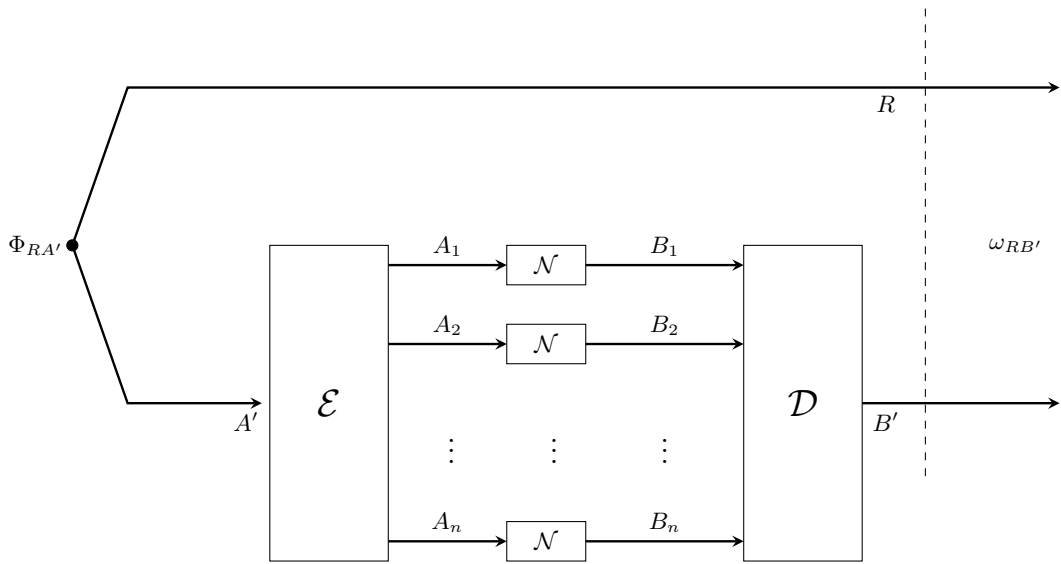
\begin{figure}[!ht]
\centering
\begin{tikzpicture}[>=stealth, scale=1.05]

\tikzset{
    thickarrow/.style={->, line width=1.5pt},
    qarrow/.style={->, line width=0.9pt},
    wire/.style={line width=0.9pt}
}

\draw (1,1) rectangle (2.5,5);
\node at (1.75,3) {\Large $\mathcal{E}$};

\draw (7,1) rectangle (8.5,5);
\node at (7.75,3) {\Large $\mathcal{D}$};

\foreach \y in {4.75,3.75,1.25} {
    \draw (4,\y-0.25) rectangle (5,\y+0.25);
    \node at (4.5,\y) {$\mathcal{N}$};
}
\node at (3.3,2.5) {\large \vdots};
\node at (4.6,2.5) {\large\vdots};
\node at (6,2.5) {\large\vdots};

\foreach \y/\lbl in {4.75/$A_1$,3.75/$A_2$,1.25/$A_n$} {
    \draw[qarrow] (2.5,\y) -- (4,\y) node[midway, above] {\lbl};
}

\foreach \y/\lbl in {4.75/$B_1$,3.75/$B_2$,1.25/$B_n$} {
    \draw[qarrow] (5,\y) -- (7,\y) node[midway, above] {\lbl};
}

\filldraw[black] (-1.5,5) circle (2pt);
\node[left] at (-1.5,5) {$\Phi_{RA'}$};

\draw[qarrow] (-1.5,5) -- (-0.8,7) -- (11,7);
\draw[qarrow] (-1.5,5) -- (-0.8,3) -- (0.9,3);
\node[below] at (8.8,7) {$R$};
\node[below] at (0.7,3) {$A'$};

\draw[qarrow] (8.5,3) -- (11,3);
\node[below] at (8.8,3) {$B'$};

\node[right] at (10,5) {$\omega_{RB'}$};
\draw[dashed] (9.3,8) -- (9.3,2);
\end{tikzpicture}
\caption{Entanglement transmission protocol over $n$ parallel uses of $\mathcal{N}$. The systems $R, A'$ and $B'$ are isomorphic of dimension $d$. The encoder $\mathcal{E}_{A' \to A^n}$ encodes the part $A'$ of the maximally entangled state $\Phi^d_{RA'}$ into the channel input systems. Later, the decoder $\mathcal{D}_{B^n \to B'}$ recovers the state from the channel output systems. The performance of the code is measured using the fidelity \eqref{eq:channel_fidelity}.}
\label{fig:quantumcommunicationquantumchannel}
\end{figure}

In this work, we focus on lower bounds on $F_c(\mathcal{N}^{\otimes n}, d)$ for finite $n\in \mathbb{N}$. Two main approaches have appeared in the literature. \begin{itemize}
    \item  Finite-blocklength \textit{error bounds} allow to find upper bounds on the infidelity of the final state with the MES using some \textit{decoupling theorem}, which features an exponential dependence of the error with $n$, where the optimal exponent can be expressed using one-shot entropies (e.g., smooth quantum conditional entropies ~\cite{dupuis2010decoupling, dupuis2014one} or quantum Renyi entropies \cite{dupuis2023privacy,Colomer_2024,Cheng2024, berta2026tight}). In practice, these achievability bounds are not tight for small $n$, as they are designed to recover asymptotic results in the limit $n \to \infty$. 
    \item Numerical lower bounds ~\cite{Reimpell2005, Reimpell2008, Fletcher2007, Taghavi2010, Johnson2017} directly tackle the optimization problem in \eqref{eq:channel_fidelity} by finding explicit encoder/decoder pairs. These \textit{seesaw} (or alternating convex search) methods transform the joint optimization into a pair of SDPs that are solved alternately, increasing the fidelity at every iteration. Since they are tailored to the specific channel, they are often good and converge in practice to a value that is very close to the actual channel fidelity. Moreover, differently from entropic error bounds, they provide an explicit encoding/decoding protocol that could be implemented in practice. 
\end{itemize}

In this section, we focus on the latter. Specifically, we address their intrinsic limitation that the corresponding SDPs grow exponentially with $n$, making them practically useless for computing lower bounds to $F_c(\mathcal{N}^{\otimes n}, d)$ beyond a few channel uses (e.g., 8/9 uses for qubit channels). In this work, we show how to extend these methods to reach values of $n$ in the tens of channel uses, by providing an algorithm based on permutation-invariant (PI) codes. Recent demonstrations of neutral-atom systems controlling over 50 qubits~\cite{Bernien2017, Manetsch2025} and superconducting and trapped-ion platforms exceeding 20 qubits~\cite{Kam2024, Moses2023} show that high-dimensional quantum systems are practically relevant now and that methods that don't have exponential computational complexity (even if they give only sub-optimal bounds) are thus also highly relevant and applicable for practical quantum communication scenarios.

\subsubsection{Channel Fidelity as Bilinear Optimization}
\label{subsec:channel_fid_as_bilinear}
The channel fidelity \eqref{eq:channel_fidelity} is particularly amenable to numerical optimization because the underlying objective function, i.e. the entanglement fidelity \eqref{eq:entanglement_fidelity}, is \textit{linear} in the input channel (differently e.g. from the diamond norm \cite{Kitaev1997, Watrous2009}). This, together with the fact that CPTP constraints on $\mathcal{E}$ and $\mathcal{D}$ can be phrased as convex constraints using \eqref{eq:choi_cptp}, makes it possible to express \eqref{eq:channel_fidelity} as a \textit{bilinear optimization problem} \cite{Berta2022}:
 \begin{equation}
\label{eq:channel_fidelity_as_bilinear_optimization}
\boxed{
\begin{aligned}
F_c(\mathcal{N}^{\otimes n},d) = \quad 
& \max_{\Phi^{\mathcal{E}_n}_{A'A^n},\Phi^{\mathcal{D}_n}_{B^nB'}} 
d_A^n d \ \Tr\!\Big[
(\Phi^d_{A'B'} \otimes \Phi^{\mathcal{N}^{\otimes n}}_{A^nB^n})\cdot 
\big((\Phi^{\mathcal{E}_n}_{A'A^n})^T \otimes \Phi^{\mathcal{D}_n^*}_{B'B^n}\big)
\Big] \\
\text{s.t.} \quad 
& \Phi^{\mathcal{E}_n}_{A'A^n} \geq 0, 
\quad \Tr_{A^n}[\Phi^{\mathcal{E}_n}_{A'A^n}] = \pi_d, \\
& \Phi^{\mathcal{D}_n^*}_{B'B^n} \geq 0, 
\quad \Tr_{B'}[\Phi^{\mathcal{D}_n^*}_{B'B^n}] = \frac{\mathbbm{1}_{d_B^n}}{d},
\end{aligned}
}
\end{equation}

where we used the conventions of Figure \eqref{fig:quantumcommunicationquantumchannel}, namely $R \simeq A' \simeq B'$ are $d-$dimensional Hilbert spaces and $d_A$ and $d_B$ denote the input and output dimensions of $\mathcal{N}_{A \to B}$.
Bilinear optimization problems as \eqref{eq:channel_fidelity_as_bilinear_optimization} have been proven \cite{Barman2017} to be NP-Hard to approximate up to a sufficiently small constant factor. Indeed, \eqref{eq:channel_fidelity_as_bilinear_optimization} can be phrased as \textit{constrained separability problem}, i.e. a constrained maximization over the set of separable states \eqref{eq:separable_state}, a known NP-hard problem in quantum information theory \cite{Gharibian2010}. Converging (although not efficient) approximations by an SDP outer hierarchy based on $k-$extendibility and quantum De Finetti theorems \cite{Caves2002, Christandl2007} have also been shown in \cite{Doherty2002, Berta2022}, and symmetry was used in \cite{Chee2023} to make the SDP scaling polynomial with $k$.

In the recent work~\cite{Kossmann2025} the i.i.d.\ symmetry of $\mathcal{N}^{\otimes n}$ was exploited to obtain efficiently compute upper bounds on the channel fidelity starting from the expression in~\eqref{eq:channel_fidelity_as_bilinear_optimization}. Specifically, in \cite[Proposition~5.1]{Kossmann2025}, using a similar proof technique as in Lemma \ref{lem:covariant_singlet_fraction}, the authors showed that the product state $\rho_{A' A^n B' B^n} \coloneqq (\Phi^{\mathcal{E}}_{A' A^n})^T \otimes \Phi^{\mathcal{D}^*}_{B' B^n}$ can be restricted to be permutation-invariant:
    \begin{equation}
        \rho_{A' A^n B' B^n} \in \mathrm{End}^{\mathfrak{S}_n}(A' \otimes
        A^n \otimes B' \otimes B^n).
    \end{equation}
 Note that this \textit{joint} symmetry does not by itself guarantee that optimal encoder and decoder can be chosen to be \textit{separately} symmetric. 

Seesaw methods provide an an iterative algorithm to find lower bounds on bilinear optimization problems like \eqref{eq:channel_fidelity_as_bilinear_optimization}, by iteratively optimizing over one variable while keeping the other fixed, transforming it into an SDP in the remaining variable. These methods give no guarantee about the optimality of the solution nor estimates of the error (differently from the aforementioned upper bounds based on de Finetti theorems \cite{Berta2022, Chee2023, Kossmann2025}) but they are often good in practice and typically converge to a value close to the actual optimum.

\subsection{SDPs in Quantum Channel Coding}
The intuition behind the seesaw method can be explained directly manipulating equation \eqref{eq:channel_fidelity} and using the Hilbert Schmidt adjoint:
\begin{align}
        F_c(\mathcal{N}^{\otimes n}, d) &= \displaystyle \max_{\mathcal{E}_n,\mathcal{D}_n} \Tr[\Phi^d_{RB'} \Phi_{RB'}^{\mathcal{D}_n\circ\mathcal{N}^{\otimes n}\circ\mathcal{E}_n}] \\
        \label{eq:seesaw_first}
        &= \displaystyle \max_{\mathcal{E}_n,\mathcal{D}_n} \Tr[\Phi_{RB^n}^{\mathcal{D}^*_n}\Phi_{RB^n}^{\mathcal{N}^{\otimes n}\circ\mathcal{E}_n}] \\
         \label{eq:seesaw_second}
        &= \displaystyle \max_{\mathcal{E}_n,\mathcal{D}_n} \Tr[\Phi_{RA^n}^{(\mathcal{D}_n\circ \mathcal{N}^{\otimes n})^*} \Phi_{RA^n}^{\mathcal{E}_n}].
    \end{align}
From equations \eqref{eq:seesaw_first} and \eqref{eq:seesaw_second}, we notice that for a fixed encoder $\mathcal{E}_n$ (respectively decoder $\mathcal{D}_n$) the optimization can be cast as an SDP in the remaining variable. These SDPs are essentially the SDP for the maximal singlet fraction \eqref{eq:maximal_singlet_fraction}.

\subsubsection{Maximal Fidelity of Recovery}
\label{subsec:maximal_fidelity_recovery}
Let us consider Bob's (i.e. the decoder's) perspective in Figure \ref{fig:quantumcommunicationquantumchannel}. Suppose that in \eqref{eq:seesaw_first} we are provided with an encoder $\mathcal{E}_n$ (e.g., a random seed  or an ansatz coming from a previous iteration): then, calling $\mathcal{M}_{A' \to B^n}  = \mathcal{N}^{\otimes n}_{A^n \to B^n}\circ\mathcal{E}_{A' \to A^n}$, its Choi state $\Phi^{\mathcal{M}}_{RB^n}$ is known, and the optimization problem becomes an SDP involving only $\Phi^{\mathcal{M}}_{RB^n}$:

\begin{definition}
    Given a quantum channel $\mathcal{M}_{R'\to B}$ with input Hilbert space of dimension $d_{R} = d$, we define the \textit{maximal fidelity of recovery} as \cite{Fletcher2007}:  
    \begin{equation}
        \label{eq:max_fidelity_recovery}
    F_D(\mathcal{M}) \coloneqq \displaystyle \max_{\mathcal{D} \in  \CPTP(B \to R)} \Tr[\Phi_{RB}^{\mathcal{D}^*}\Phi_{RB}^{\mathcal{M}}].
\end{equation}
One immediately sees that this is equal to the maximal singlet fraction of the Choi state: \begin{equation}
\label{eq:recovery_as_singlet_fraction}
     F_D(\mathcal{M}) = F^{(B)}_{\Phi}(\Phi_{RB}^{\mathcal{M}}),
\end{equation}
so this can be formulated as the SDP:
\begin{equation}
\label{eq:SDP_Maximal_Recovery_Coefficient}
\boxed{
\begin{aligned}
\textbf{Primal:}\quad & \\
\max_{\Phi^{\mathcal{D}^*}_{RB}} \quad & \Tr\!\big(\Phi_{RB}^{\mathcal{M}}\, \Phi^{\mathcal{D}^*}_{RB}\big) \\
\textrm{s.t.} \quad & \Tr_R[\Phi^{\mathcal{D}^*}_{RB}] = \frac{\mathbbm{1}_{B}}{d} \\
& \Phi^{\mathcal{D}^*}_{RB}\geq 0
\end{aligned}
\qquad\qquad
\begin{aligned}
\textbf{Dual:}\quad & \\
\min_{Y_{B}} \quad & \frac{1}{d}\,\Tr(Y_{B}) \\
\textrm{s.t.} \quad & \Phi_{RB}^{\mathcal{M}} \leq \mathbbm{1}_R \otimes Y_{B} \\
&
\end{aligned}
}
\end{equation}
\end{definition}

This quantity, as the name suggests, corresponds to the optimal entanglement fidelity that can be achieved via a local operation at Bob's side for the given encoder. It finds operational interpretation in the quantum recovery task (see Figure \ref{fig:quantumrecoverytask_mes}), considered first in \cite{Barnum2002}, which is the fully quantum analogue of the guessing task that gives operational interpretation to the conditional min-entropy \cite{Koenig2009}.

\begin{figure}[h]
\centering
\begin{tikzpicture}[
    every node/.style={font=\Large},
    block/.style={draw, line width=0.7pt, minimum width=2.5cm, minimum height=1.2cm, align=center},
    arrow/.style={line width=0.7pt, -{Stealth[scale=1.0]}}
]
\filldraw[black] (-1.5,1.2) circle (1.3pt);
\node[left] at (-1.5,1.2) {$\Phi_{RR'}$};
\draw[arrow] (-1.5,1.2) -- (-0.8,2.2) -- (10,2.2);
\node[above] at (9,2.2) {$R$};
\draw[arrow] (-1.5,1.2) -- (-0.8,0.2) -- (1.2,0.2)  node[midway, above] {$R'$};
\node[block, right=0.7cm of {(0.5,0.2)}, anchor=west] (channel) {$\mathcal{M}_{R' \to B}$};
\draw[arrow] (channel.east) -- ++(1.5,0) node[midway, above] {$B$};
\node[block, right=4.2cm of {(1,0.2)}, anchor=west] (decoder) {$\mathcal{D}_{B \to R'}$};
\draw[arrow] (decoder.east) -- ++(2.2,0) node[midway, above] {$R'$};
\node[right] at (9.5,1.2) {$\omega_{RR'} \overset{?}{\approx} \Phi_{RR'}$};
\draw[dashed, gray, line width=0.5pt] (9,2.5) -- (9,0);
\end{tikzpicture}
\caption{Quantum recovery task: given a fixed channel $\mathcal{M}_{R' \to B}$, the goal is to find the optimal decoder $\mathcal{D}_{B \to R'}$ maximizing the fidelity of the output state with the maximally entangled state.}
\label{fig:quantumrecoverytask_mes}
\end{figure}
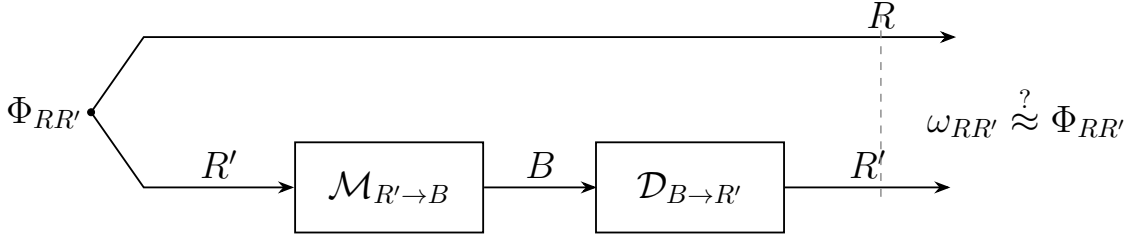

By \eqref{eq:recovery_as_singlet_fraction}, the maximal fidelity of recovery of a channel is a 'dynamical' version of the maximal singlet fraction; if $F_{\Phi}(\rho_{AB})$ is a measure of how much perfect entanglement (i.e., a MES \eqref{eq:MES}) can be retrieved from a state $\rho_{AB}$ by acting locally on it, $F_D(\mathcal{M}_{A \to B})$ can be interpreted as a measure of how well noiseless communication can be recovered from $\mathcal{M}_{A \to B}$ by acting on its output system. This correspondence becomes particularly transparent when comparing Figures~\ref{fig:maximal_singlet_fraction_fixed} and~\ref{fig:quantumrecoverytask_mes}. In particular, using \eqref{eq:maximal_singlet_fraction_as_generalized_roboustness}, we immediately obtain:
\begin{equation}
\label{eq:recovery_as_generalized_roboustness}
    F_{\text{D}}(\mathcal{M}) = \min_{Y_B \geq 0}\left\{\frac{1}{d}\Tr[Y_B]: \mathcal{R}^{Y_B} \geq \mathcal{M}\right\},
\end{equation}
where the sign $\geq$ means that $\mathcal{R}^{Y_B} - \mathcal{M}$ is a CP map, and $\mathcal{R}_{R'\to B}^{\sigma_B}(X) = \Tr [X] \sigma_B$ for some $\sigma_B \in \mathcal{D}(B)$. Thus, $F_D$ can be interpreted as a generalized robustness to replacement channels. The correspondence in \eqref{eq:recovery_as_singlet_fraction} allows to prove several properties of the maximal fidelity of recovery. Moreover, using a channel variant of Equation \eqref{eq:conditional_min_entropy_from_singlet_fraction}, the maximal fidelity of recovery can be related with a notion of conditional min entropy of a channel \cite{Gour2019, GourWilde2021}. These properties and connections are explored in detail in Appendix \ref{sec:entropy_channel}.

Note in particular that by Lemma \ref{lem:covariant_singlet_fraction}, if the output system is $B^n$ and $\mathcal{M}_{R' \to B^n}$ is symmetric in the sense of \eqref{eq:condition_symmetry_encoder}, then the optimal $\mathcal{D}_{B^n \to R'}$ can be restricted to be symmetric in the sense of \eqref{eq:condition_symmetry_decoder}, i.e.\begin{equation}
\label{eq:symmetric_restriction_F_D}
    \Phi^{\mathcal{M}}_{RB^n}\in \esn(R \otimes B\n) \Longrightarrow  \Phi^{\mathcal{D}^*}_{RB^n}\in \esn(R \otimes B\n).
\end{equation}

\subsubsection{Maximal Fidelity of Preparation}
\label{subsec:maximal_fidelity_preparation}
Let's consider Alice's (i.e.\ the encoder's) perspective in Figure \ref{fig:quantumcommunicationquantumchannel}. Suppose that in~\eqref{eq:seesaw_second} we are provided with a decoder $\mathcal{D}_n$ (e.g., a random seed or a decoder coming from a previous iteration): then, calling $\mathcal{M}'_{A^n \to B'} = \mathcal{D}_{B^n \to B'} \circ \mathcal{N}^{\otimes n}_{A^n \to B^n}$, the Choi state of its adjoint $\Phi^{\mathcal{M}'^*}_{B'A^n} = \Phi^{(\mathcal{D}_n \circ \mathcal{N}^{\otimes n})^*}_{B'A^n}$ is known, and the optimization problem becomes an SDP with input $\Phi^{\mathcal{M}'^*}_{B'A^n}$:
\begin{definition}
\label{def:maximal_fidelity_preparation}
Given a quantum channel $\mathcal{M}'_{A \to R}$ with output Hilbert space of dimension $d_{R} = d$, we call the \textit{maximal fidelity of preparation} as
\begin{align}
\label{eq:max_fidelity_preparation}
    F_E(\mathcal{M}') &\coloneqq \displaystyle \max_{\mathcal{E}\in \CPTP(R \to A)} \Tr[\Phi_{RA}^{\mathcal{M}'^*} \Phi_{RA}^{\mathcal{E}}].
\end{align}
This can again be formulated as an SDP:
\begin{equation}
\label{eq:SDP_Maximal_Preparation_Coefficient}
\boxed{
\begin{aligned}
\textbf{Primal:}\quad & \\
\max_{\Phi^{\mathcal{E}}_{RA}} \quad & \Tr\!\big(\Phi_{RA}^{\mathcal{M}'^*}\, \Phi^{\mathcal{E}}_{RA}\big) \\
\textrm{s.t.} \quad & \Tr_{A}[\Phi^{\mathcal{E}}_{RA}] = \frac{\mathbbm{1}_R}{d} \\
& \Phi^{\mathcal{E}}_{RA} \geq 0
\end{aligned}
\qquad\qquad
\begin{aligned}
\textbf{Dual:}\quad & \\
\min_{Y_R} \quad & \frac{1}{d}\,\Tr(Y_R) \\
\textrm{s.t.} \quad & \Phi_{RA}^{\mathcal{M}'^*} \leq Y_R \otimes \mathbbm{1}_{A} \\
&
\end{aligned}
}
\end{equation}
\end{definition}

This quantity, as the name suggests, corresponds to the optimal entanglement fidelity achievable via a local operation at Alice's side for the given decoder. It finds operational interpretation in the quantum preparation task (Figure~\ref{fig:quantumpreparationtask_mes}): given a fixed channel $\mathcal{M}'_{A \to R}$, find the optimal encoder $\mathcal{E}_{R \to A}$ maximizing the fidelity of the output state with the maximally entangled state.

\begin{remark}
\label{rem:isometric_encoder}
It is intuitive that the optimal encoding map should not add noise, i.e., it should be an isometry (in fact, standard quantum error-correcting codes, such as the nine-qubit code~\cite{Shor1995} or the five-qubit code~\cite{Laflamme1996} have isometric encoders). The objective function $f(\mathcal{E}) \coloneqq F_e(\mathcal{M}' \circ \mathcal{E}) = \Tr(\Phi_{RA}^{\mathcal{M}'^*} \Phi_{RA}^{\mathcal{E}})$ is linear in $\mathcal{E}$ and since $\CPTP(R' \to A)$ is convex and compact, the maximum is attained at an extreme point. In line with our intuition, isometries are indeed extreme points of the set of CPTP maps, but there exist extremal CPTP maps that are not isometric~\cite[Theorem 2.31]{Watrous2018}. Thus, the optimization set in $F_{\mathrm{E}}$ cannot in general be reduced to isometries.
\end{remark}

The maximal fidelity of preparation is essentially equivalent to its maximal fidelity of recovery, in the sense specified below.

\begin{proposition}
    Let $\mathcal{W}_{A \to B}$ be a quantum channel of input dimension $d_A$ and output dimension $d_B$. Then \begin{equation}
    \label{eq:link_F_E_F_D}
         F_E(\mathcal{W}) =\frac{d^2_A}{d^2_B} F_D(\mathcal{W}).
    \end{equation}
\end{proposition}
\begin{proof}
On the one hand, since $\mathcal{D} \in \CPTP(A \to B)$ implies $\mathcal{D}^* \in \operatorname{CPU}(A \to B)$:\begin{equation}
\label{eq:proof_F_D}
    F_D(\mathcal{W}) =\frac{1}{d_A^2}\max_{\Gamma^{\mathcal{D}^*}_{AB}\in \mathcal{C}^*(A;B)}\Tr[ \Gamma^{\mathcal{D}^*}_{AB}\Gamma^{\mathcal{W}}_{AB}].
\end{equation}
On the other:
\begin{align}
    F_E(\mathcal{W}) &= \frac{1}{d_B^2} \max_{\Gamma^{\mathcal{E}}_{BA}\in \mathcal{C}(B;A)}\Tr[ \Gamma^{\mathcal{W}^*}_{BA}\Gamma^{\mathcal{E}}_{BA}]\\
    &=\frac{1}{d_B^2}\max_{\Gamma^{\mathcal{E}^*}_{AB}\in \mathcal{C}^*(A;B)}\Tr[ \Gamma^{\mathcal{W}}_{AB}\Gamma^{\mathcal{E}^*}_{AB}]\\
    &=\frac{1}{d_B^2}\max_{\Gamma^{\mathcal{E}^*}_{AB}\in \mathcal{C}^*(A;B)}\Tr[ \Gamma^{\mathcal{E}^*}_{AB}\Gamma^{\mathcal{W}}_{AB}]
    \label{eq:equality_objective}
\end{align}
where the second equality follows by \eqref{eq:unitality-trace-preservation-are-dual} and \eqref{eq:adjoint_choi}. Comparing  \eqref{eq:equality_objective} and \eqref{eq:proof_F_D} and relabeling $\mathcal{E}\leftrightarrow \mathcal{D}$, we obtain the claim.
\end{proof} By \eqref{eq:link_F_E_F_D}, all properties of Lemma \ref{lem:max_recovery_properties} transfer to the maximal fidelity of preparation. Note that if the channel $\mathcal{W}$ is square (i.e. $d_A = d_B$), then
\begin{equation}
\label{eq:condition_coincidence_M}
    F_{\mathrm{D}}(\mathcal{W}) = F_{\mathrm{E}}(\mathcal{W}),
\end{equation}
which operationally means that the maximum fidelity achievable by controlling only the channel input (preparation task, Figure \ref{fig:quantumpreparationtask_mes}) equals that achievable by controlling only the channel output (recovery task, Figure \ref{fig:quantumrecoverytask_mes}). 
\begin{figure}[htb]
\centering
\begin{tikzpicture}[
    every node/.style={font=\Large},
    block/.style={draw, line width=0.7pt, minimum width=2.5cm, minimum height=1.2cm, align=center},
    arrow/.style={line width=0.7pt, -{Stealth[scale=1.0]}}
]
\filldraw[black] (-1.5,1.2) circle (1.3pt);
\node[left] at (-1.5,1.2) {$\Phi_{RR'}$};
\draw[arrow] (-1.5,1.2) -- (-0.8,2.2) -- (10,2.2);
\node[above] at (9,2.2) {$R$};
\draw[arrow] (-1.5,1.2) -- (-0.8,0.2) -- (1.2,0.2)  node[midway, above] {$R'$};
\node[block, right=0.7cm of {(0.5,0.2)}, anchor=west] (encoder) {$\mathcal{E}_{R' \to A}$};
\draw[arrow] (encoder.east) -- ++(1.5,0) node[midway, above] {$A$};
\node[block, right=4.2cm of {(1,0.2)}, anchor=west] (channel) {$\mathcal{M}'_{A \to R'}$};
\draw[arrow] (channel.east) -- ++(2.2,0) node[midway, above] {$R'$};
\node[right] at (9.5,1.2) {$\omega_{RR'} \overset{?}{\approx} \Phi_{RR'}$};
\draw[dashed, gray, line width=0.5pt] (9,2.5) -- (9,0);
\end{tikzpicture}
\caption{Quantum preparation task: given a fixed channel $\mathcal{M}'_{A \to R'}$, the goal is to find the optimal encoder $\mathcal{E}_{R' \to A}$ maximizing the fidelity of the output state with the maximally entangled state.}
\label{fig:quantumpreparationtask_mes}
\end{figure}
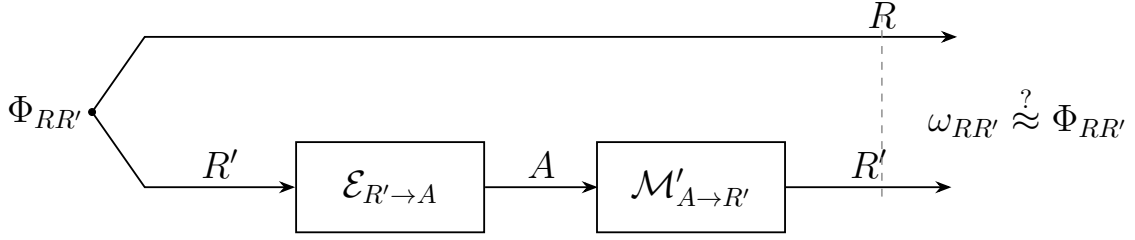

However, in the seesaw method, the channels $\mathcal{M}$ and $\mathcal{M}'$ for which we compute $F_D$ and $F_E$ do not have in general equal input and output dimensions. In the recovery task \eqref{eq:seesaw_first}, we consider $F_D(\mathcal{M}_{A' \to B^n})$ with $\mathcal{M}_{A' \to B^n} = \mathcal{N}_{A^n \to B^n}^{\otimes n}\circ \mathcal{E}_{A' \to A^n}$; in the preparation task \eqref{eq:seesaw_second}, we consider $F_E(\mathcal{M}'_{A^n \to B'})$ with $\mathcal{M}'_{A^n \to B'} =  \mathcal{D}_n \circ\mathcal{N}_{A^n \to B^n}^{\otimes n}$. Beyond the dimensional mismatch, the two SDPs have a genuinely different structure once symmetry is exploited: the permutation-invariant system $S^n$ --- $B^n$ for $F_D$ in the recovery task, $A^n$ for $F_E$ in the preparation task --- sits on opposite sides of the partial trace normalization constraint in \eqref{eq:SDP_Maximal_Recovery_Coefficient} and \eqref{eq:SDP_Maximal_Preparation_Coefficient}. As we saw in Section~\ref{sec:efficient-operations-in-block-basis}, this difference matters at the implementation level, because the block-diagonal reformulation treats $\Tr_{S^n}$ and $\Tr_R$ asymmetrically (see Section \ref{sec:efficient-operations-in-block-basis}). This is why it is natural to treat $F_D$ and $F_E$ as separate quantities.
\subsubsection{Seesaw Algorithm}
We can now formally introduce the seesaw algorithm~\cite{Reimpell2005,Fletcher2007,Taghavi2010,Johnson2017}. 

\begin{definition}
\label{def:seesaw}
Given a quantum channel $\mathcal{N}_{A \to B}$, a number of channel uses $n \in \mathbb{N}$, an input dimension $d \leq \min\{d_A^n, d_B^n\}$, and a convergence threshold $\delta > 0$, in the setting of Figure \ref{fig:quantumcommunicationquantumchannel}, the seesaw iteration method produces a lower bound $\tilde{F} \leq F_c(\mathcal{N}^{\otimes n}, d)$ together with an encoding--decoding pair $(\tilde{\mathcal{E}}, \tilde{\mathcal{D}})$ that achieves it. Let $A'$ and $B'$ both be systems of the given input dimension $d$. The algorithm then proceeds as follows.

\begin{enumerate}
    \item \textbf{Initialization.} Choose a random encoder $\mathcal{E}_0 \in \CPTP(A' \to A^n)$ and a random decoder $\mathcal{D}_0 \in \CPTP(B^n \to B')$ as initial seeds. Set $i = 0$.

    \item \textbf{Iteration.} Repeat:
    \begin{enumerate}
        \item \textit{Optimize decoder.} Solve the SDP~(\ref{eq:SDP_Maximal_Recovery_Coefficient}) for the channel $\mathcal{M}_i = \mathcal{N}^{\otimes n} \circ \mathcal{E}_{i}$, obtaining the optimal decoder $\mathcal{D}_{i+1} \in \CPTP(B^n \to B)$ and the value\begin{equation}
    \label{eq:max_recovery_seesaw}
\max_{\mathcal{D}_{i+1}} F(\Phi^d_{RB'}, \Phi^{\mathcal{D}_{i+1} \circ \mathcal{M}_i}_{RB'}) \eqqcolon F_{\mathrm{D}}(\mathcal{M}_i)
    \end{equation}
        \item \textit{Optimize encoder.} Solve the SDP~(\ref{eq:SDP_Maximal_Preparation_Coefficient}) for the channel $\mathcal{M}_i' =\mathcal{D}_{i+1} \circ \mathcal{N}^{\otimes n}$, obtaining the optimal encoder $\mathcal{E}_{i+1}\in \CPTP(A' \to A^n)$ and the value \begin{equation}
    \label{eq:max_preparation_seesaw}
 \max_{\mathcal{E}_{i+1}} F(\Phi^d_{RB'}, \Phi^{\mathcal{M}'_i \circ \mathcal{E}_{i+1}}_{RB'}) \eqqcolon F_{\mathrm{E}}(\mathcal{M}'_i)
    \end{equation}
        \item \textit{Convergence check.} If
        \begin{equation}
        \label{eq:seesaw_convergence_check}
            F_{\mathrm{E}}(\mathcal{M}_i') - F_{\mathrm{D}}(\mathcal{M}_i) < \delta,
        \end{equation}
        terminate. Otherwise, set $i \leftarrow i+1$ and repeat.
    \end{enumerate}
    \item \textbf{Output.} Return $\tilde{F} \coloneqq F_e(\mathcal{D}_{i+1} \circ \mathcal{N}^{\otimes n} \circ \mathcal{E}_{i+1})$ and the maps $(\tilde{\mathcal{E}}, \tilde{\mathcal{D}}) = (\mathcal{E}_{i+1}, \mathcal{D}_{i+1})$ (as Choi matrices).
\end{enumerate}
\end{definition}

The seesaw method was originally introduced in the context of quantum error correction to evaluate the performance of quantum codes, and the reason why it works is given by the following simple lemma. 

\begin{lemma}
\label{lemma:fidelity_increase}
The seesaw iteration method ensures a monotonic increase in the entanglement fidelity at every half-step. In other words, for every $i \in \mathbb{N}$:
\begin{equation}
F_{e}(\mathcal{D}_{i+1} \circ \mathcal{N}^{\otimes n}\circ \mathcal{E}_{i+1}) \geq F_{e}(\mathcal{D}_{i+1} \circ \mathcal{N}^{\otimes n}\circ \mathcal{E}_{i}) \geq F_{e}(\mathcal{D}_{i} \circ \mathcal{N}^{\otimes n}\circ \mathcal{E}_{i})
\end{equation}
\end{lemma}
\begin{proof}
At a generic step $i$ of the iteration, one has:
\begin{align}
   F_{e}(\mathcal{D}_{i+1} \circ \mathcal{N}^{\otimes n}\circ \mathcal{E}_{i+1}) &= F_{\mathrm{E}}(\mathcal{D}_{i+1} \circ \mathcal{N}^{\otimes n}) \geq F_{e}(\mathcal{D}_{i+1} \circ \mathcal{N}^{\otimes n} \circ \mathcal{E}_{i} ) \\
   &= F_{\mathrm{D}}(\mathcal{N}^{\otimes n} \circ \mathcal{E}_{i}) \geq F_{e}(\mathcal{D}_{i} \circ \mathcal{N}^{\otimes n} \circ \mathcal{E}_{i})
\end{align}
where each equality holds because the preceding SDP optimizes over the corresponding map, and each inequality holds because the previous iterate is a feasible (but generally suboptimal) point in that optimization. 
\end{proof}
By Lemma \ref{lemma:fidelity_increase}, the seesaw method guarantees convergence to a \textit{local} optimum $\tilde{F}$, whose value depends on the initial condition given by the random seed $(\mathcal{E}_0, \mathcal{D}_0)$. To obtain better results, ideally close to the global optimum, one should repeat the method with multiple random initializations and take the maximum over the resulting fidelities. However, it should be noted that local minima may cause the algorithm to converge to a flat region and there is no certificate for optimality. The performance of the method is strongly dependent on the choice of the initial seed: some pairs $(\mathcal{E}_0, \mathcal{D}_0)$ lead to significantly faster convergence than others. For instance, choosing $\mathcal{E}_0$ to be an isometric (noiseless) encoder is often advantageous, as it provides a warm start for the subsequent iterations and accelerates convergence.
\

The entire optimization procedure can be formulated only using the Choi matrices of the involved channels, exploiting the concatenation \eqref{eq:concatenation_choi} and adjoint \eqref{eq:adjoint_choi} rules for Choi matrices, and essentially consists in solving repeatedly the SDPs of $F_E$ and $F_D$ for varying inputs $\mathcal{M}_i$ and $\mathcal{M}_i '$, updated at every step. Besides standard SDP solvers, such as MOSEK \cite{Mosek2024} and SCS \cite{Odonoghue2016}), the authors of \cite{Reimpell2005} also introduced an alternative method called the channel power iteration method, inspired by the classical power iteration, and is reviewed in the next subsection. 

\subsubsection{Channel Power Iteration Method}
\label{sec:powermethod}
Each half-step \eqref{eq:SDP_Maximal_Recovery_Coefficient} and \eqref{eq:SDP_Maximal_Preparation_Coefficient} of the seesaw iteration can be solved via the \textit{channel power iteration method}~\cite[Section 3.2.1]{Reimpell2008}. The method is inspired by the classical power method~\cite{vonMises1929} for computing the dominant eigenvector of a positive semidefinite matrix $A$, which iteratively applies $v \mapsto Av/\|Av\|$. 
The channel analogue replaces the linear map $v \mapsto Av$ with a sandwiching operation on the Choi matrix of the encoder (resp. decoder adjoint), followed by a normalization that enforces the unitality (resp. trace-preservation) constraint in \eqref{eq:SDP_Maximal_Recovery_Coefficient} (resp. in \eqref{eq:SDP_Maximal_Preparation_Coefficient}). The next definition formalizes the algorithm, focusing on $F_D$ for concreteness (completely analogous results hold for $F_E$, by \eqref{eq:link_F_E_F_D}).
\begin{definition}
\label{def:power_method}
Let $\mathcal{M}_{R' \to B}$ be a quantum channel, $\mathcal{D}_0 \in \CPTP(B \to R')$ be any decoder map (see figure \ref{fig:quantumrecoverytask_mes}) and $\delta_p > 0$ be a convergence threshold. The channel power iteration produces an estimate of $F_{\mathrm{D}}(\mathcal{M})$ and a decoder $\tilde{\mathcal{D}}_{i+1}$ that achieves it via the following iterative algorithm. Set $j = 0$ and repeat:
\begin{enumerate}
    \item \textbf{Sandwich.} Compute the unnormalized Choi matrix:
    \begin{equation}
    \label{eq:power_step}
       X_{RB} \coloneqq \Gamma^{\mathcal{M}}_{RB}  \Gamma^{\mathcal{D}^*_j}_{RB}\Gamma^{\mathcal{M}}_{RB}
    \end{equation}
    \item \textbf{Normalize.} Compute $X_{B} \equiv \Tr_{R}[X_{RB}]$ and obtain the Choi matrix of the normalized decoder $\mathcal{D}^*_{j+1}$ by applying $X_{B}^{-1/2}$ (pseudoinverse), i.e.:
    \begin{equation}
    \label{eq:normalization_choi}
        \Gamma^{\mathcal{D}^*_{j+1}}_{RB} \coloneqq (\mathbbm{1}_R \otimes X_{B}^{-1/2}) X_{RB} (\mathbbm{1}_R \otimes X_{B}^{-1/2})
    \end{equation} 
    This enforces the unitality constraint $\Tr_{R}[\Gamma^{\mathcal{D}^*_{j+1}}_{RB}] = \mathbbm{1}_{B}$.
    \item \textbf{Convergence check.} If $F_e(\mathcal{D}_{j+1} \circ \mathcal{M}) - F_e(\mathcal{D}_j \circ \mathcal{M}) < \delta_p$, terminate. Otherwise, set $j \leftarrow j+1$ and repeat.
\end{enumerate}
Return $\mathcal{D} \coloneqq \mathcal{D}_{j+1}$ and $\tilde{F} \coloneqq F_e(\mathcal{D} \circ \mathcal{M})$.
\end{definition}
The same algorithm can be applied to estimate $F_E$, with the only difference in the normalization step, where the TP constraint is enforced by taking the partial trace with respect to system $B$. It also trivially extends to the $n$-shot setting of Figure \ref{fig:quantumcommunicationquantumchannel}, substituting $B\to B^n$. Under suitable regularity conditions (which are often satisfied in practice), the method satisfies monotonicity at every step:\begin{equation}
    F_e(\mathcal{D}_{j+1} \circ \mathcal{M}) \geq F_e(\mathcal{D}_j \circ \mathcal{M}) \quad \forall j \in \mathbb{N},
\end{equation} and any globally optimal decoder is a fixed point of the iteration; we refer to Reimpell's dissertation \cite[Chapter 3.2]{Reimpell2008} for details of the convergence analysis. 

This method provides a valid alternative at each half-step in the seesaw method \ref{def:seesaw}: in step 2(a) one sets $\mathcal{M}_i = \mathcal{N}^{\otimes n} \circ \mathcal{E}_i$ and runs the iteration seeded with $\mathcal{D}_i$; in step 2(b) one works with $\mathcal{M}_i'^* = (\mathcal{D}_{i+1} \circ \mathcal{N}^{\otimes n})^*$ seeded with $\mathcal{E}_i$. The only difference with standard SDP solvers is that we now require the decoder coming from the previous iteration as initial seed to start the iteration. Moreover, it compares favorably with standard SDP solvers in terms of computational complexity. As shown in \cite[Sections 3.2.1 and 3.2.2]{Reimpell2008}, computing $F_D$ (primal SDP) requires $\mathcal{O}((d_B d)^6)$ floating-point operations, whereas the power iteration reduces this cost to $\mathcal{O}((d_B d)^3)$, since each step is dominated by a single matrix--matrix multiplication. Moreover, SDP solvers typically rely on interior-point or first-order methods, which are inherently CPU-bound. In contrast, the power iteration consists entirely of matrix operations (see \eqref{eq:power_step} and \eqref{eq:normalization_choi}) which can be efficiently implemented on GPUs using Python libraries such as \texttt{CuPy}, or can benefit from additional structure (e.g., sparsity or block-diagonal structure), leading to further speedups. In practice, in all examples (both in \cite{Parentin2026} and in Section~\ref{sec:numerical_examples}), the channel power iteration method proved to be significantly faster than SDP solvers.



\

\subsection{Symmetric Seesaw Method}
\label{sec:symmetric_seesaw}

Even without a convergence guarantee, the seesaw iteration algorithm proved to be very successful in the context of approximate error correction, reproducing or even outperforming known bounds for several channels (see \cite[Section 3.3]{Reimpell2008}). However, it suffers the fundamental bottleneck that the Choi matrices involved in the optimization grow exponentially with $n$. Since increasing the number of uses can help improving the fidelity, i.e. for all $d \in \mathbb{N}$ and for all $\mathcal{N} \in \CPTP(A \to B)$\begin{equation}
\label{eq:monotonicity-channel-fidelity}
    F_c(\mathcal{N}^{\otimes n}, d) \leq F_c(\mathcal{N}^{\otimes m}, d) \quad  \forall n \leq  m,
\end{equation}
it is very important in practice to run the seesaw method for a large number of channel uses, i.e. beyond the practical limiting number of $n =8/9$ (for transmission over qubit channels). This becomes crucial if if the goal is to surpass a threshold fidelity value, as in the case of \cite{Parentin2026}, but can be equally important whenever we want to find or benchmark good error-correcting codes of larger blocklength, as it provides with an optimized coding strategy that works for any quantum channel. 
In this section, we address this problem, introducing a new algorithm that we call the \textit{symmetric seesaw method}, which exploits the permutation symmetry of i.i.d.\ channels to reduce the complexity from exponential to polynomial in~$n$.

\subsubsection{Permutation-Invariant Codes}
\label{sec:perm-inv-codes}

\begin{definition}
    Let $n,d \in \mathbbm{N}$ and consider $n$-shot entanglement transmission over $\mathcal{N}^{\otimes n}$ (Figure ~\ref{fig:quantumcommunicationquantumchannel}). We say that the encoder $\mathcal{E}_{A' \to A^n}$ and the decoder $\mathcal{D}_{B^n \to B'}$ are \textit{symmetric} if\begin{equation}
\label{eq:symmetric_enc-dec}
    \Gamma^{\mathcal{E}_n}_{RA^n} \in \mathrm{End}^{\mathfrak{S}_n}(R \otimes A^{\otimes n}), \qquad \Gamma^{\mathcal{D}_n^*}_{RB^n} \in \mathrm{End}^{\mathfrak{S}_n}(R \otimes B^{\otimes n})
\end{equation}
where (as previously) $\mathfrak{S_n}$ acts trivially on the $R$ system. The pair $(\mathcal{E}_{A' \to A^n}, \mathcal{D}_{B^n \to B'})$ constitutes a \textit{permutation-invariant} (PI) code.
\end{definition}
PI codes encode and decode quantum information using permutation-invariant states with respect to permutations of the $n$ subsystems. These codes were first studied by Ruskai~\cite{Ruskai2000} and Pollatsek and Ruskai~\cite{PollatsekRuskai2004}, and later extended to systematic families in ~\cite{Ouyang2014, Aydin2024}. They have natural physical motivation: systems of indistinguishable bosons are necessarily permutation-invariant, and PI codes are inherently immune to permutation errors. Beyond quantum error correction, PI codes have found applications in quantum storage~\cite{Ouyang2021}, metrology~\cite{Ouyang2022}, and quantum communication~\cite{Bhalerao2025}.

The key idea underlying PI codes is that i.i.d.\ channels $\mathcal{N}^{\otimes n}$ preserve any permutation symmetry present at the input\footnote{Note the following results apply to any channel whose Choi matrix satisfies $\Gamma^{\mathcal{N}_n}_{A^n B^n} \in \mathrm{End}^{\mathfrak{S}_n}(A^n \otimes B^n)$. For example, this includes all channels whose Choi states have \textit{de Finetti} form:
\begin{equation}
    \Gamma^{\mathcal{N}_n}_{A^n B^n} = \sum_{x \in \mathcal{X}} p_X(x)
    \left(\Gamma^{\mathcal{N}^x}_{AB}\right)^{\otimes n}
\end{equation}
where $p_X$ is a probability distribution over some alphabet $\mathcal{X}$. Such channels model scenarios with correlated (non-i.i.d.) noise.}. This enables to use the representation theoretic tools of Section \ref{sec:perm-inv} to develop efficient algorithms that reduce the complexity of the underlying SDPs from exponential to polynomial in~$n$. 

While PI codes need not be optimal (see the discussion in subsection \ref{subsec:channel_fid_as_bilinear}), this restriction is a computationally motivated ansatz that has proven very effective in practice: Bhalerao and Leditzky~\cite{Bhalerao2025} recently showed that restricting to permutation invariant states in the optimization of the coherent information for multiple channel uses significantly improves the best known quantum capacity thresholds for several channel families (e.g., the 2-Pauli and BB84 channels). In this work, we show that PI codes are similarly effective for the finite-blocklength approximate error correction problem, yielding improved lower bounds on the channel fidelity for several important quantum channels (Section~\ref{sec:numerical_examples}).

Having motivated our choice, we impose encoder and decoder symmetry as a computational ansatz, and see how this restriction reduces the optimization space in \eqref{eq:SDP_Maximal_Recovery_Coefficient} and \eqref{eq:SDP_Maximal_Preparation_Coefficient} from exponential to polynomial dimension in $n$ using representation theory. The key structural result is that the seesaw iteration \textit{preserves} the symmetric subspace:
\begin{lemma}
\label{lemma:seesaw_symmetric}
\label{prop:seesaw_symmetric}
Let $n \in \mathbb{N}$ be a number of channel uses, and $d \in \mathbb{N}$ be an input dimension. Given a permutation-covariant channel $\mathcal{N}_n$, i.e. such that $\Gamma^{\mathcal{N}_n}_{A^nB^n} \in \mathrm{End}^{\mathfrak{S}_n}(A^{\otimes n} \otimes B^{\otimes n})$, if the encoder/decoder seed $(\mathcal{E}_0,\mathcal{D}_0)$ is symmetric in the sense of~\eqref{eq:symmetric_enc-dec}, then the entire seesaw iteration for $F_c(\mathcal{N}_n, d)$ can be restricted to the symmetric subspace. In particular, the SDPs~\eqref{eq:SDP_Maximal_Recovery_Coefficient} and~\eqref{eq:SDP_Maximal_Preparation_Coefficient} can be restricted to the space $\mathrm{End}^{\mathfrak{S}_n}(R \otimes S^{\otimes n})$, where $S = A$ for $F_E$ and $S = B$ for $F_D$.
\end{lemma}
\begin{proof}
We proceed by induction on the iteration number $i\in \mathbb{N}$. The base case $i = 0$ holds by assumption. Given a generic iteration $i>0$, assume that $\mathcal{E}_i$ is symmetric in the sense of~\eqref{eq:symmetric_enc-dec}. By permutation covariance of $\mathcal{N}_n$, using~\eqref{eq:concatenation_choi}, the Choi matrix of the concatenated channel $\mathcal{M}_i = \mathcal{N}_n \circ \mathcal{E}_i$ satisfies:\begin{equation}
    \Gamma^{\mathcal{M}_i}_{RB^n} \in \mathrm{End}^{\mathfrak{S}_n}(R \otimes
B^{\otimes n}).
\end{equation}
Then, by applying \eqref{eq:symmetric_restriction_F_D}, the optimal solution of~\eqref{eq:SDP_Maximal_Recovery_Coefficient} can be restricted to $\mathrm{End}^{\mathfrak{S}_n}(R \otimes
B^{\otimes n})$, which by~\eqref{eq:adjoint_choi} implies that $\mathcal{D}_{i+1}$ is symmetric. The same reasoning can be applied to the encoder step: $\mathcal{D}_{i+1}$ symmetric and $\mathcal{N}_n$ permutation-covariant imply
$\Gamma^{\mathcal{M}_n'^*}_{RA^n} = \Gamma^{(\mathcal{D}_{i+1} \circ \mathcal{N}_n)^*}_{RA^n} \in
\mathrm{End}^{\mathfrak{S}_n}(R \otimes A^{\otimes n})$. By \eqref{eq:link_F_E_F_D}, we can again apply \eqref{eq:symmetric_restriction_F_D}, to conclude that $\mathcal{E}_{i+1}$ can be chosen symmetric. We have shown that at every iteration $i$, if the encoding/decoding pair $(\mathcal{E}_i, \mathcal{D}_i)$ is symmetric in the sense of~\eqref{eq:symmetric_enc-dec}, then also the optimal $(\mathcal{E}_{i+1}, \mathcal{D}_{i+1})$ can be restricted to lie in the symmetric subspace. 
\end{proof}

Motivated by the lemma, we call the seesaw method for which the original seeds are restricted to lie in the permutation-invariant subspace~\eqref{eq:symmetric_enc-dec} as the \textit{symmetric seesaw method}.

\begin{definition}
\label{def:symmetric-seesaw}
Given a quantum channel $\mathcal{N}_{A \to B}$, a number of channel uses $n \in \mathbb{N}$, an input dimension $d \leq \min\{d_A^n, d_B^n\}$, and a convergence threshold $\delta > 0$, the \textit{symmetric seesaw method} proceeds exactly as in Definition~\ref{def:seesaw}, with the additional requirement that the initial seeds $(\mathcal{E}_0, \mathcal{D}_0)$ satisfy~\eqref{eq:symmetric_enc-dec}. By Lemma~\ref{prop:seesaw_symmetric}, the entire iteration then remains in the permutation-invariant subspace, and the two (primal) SDPs \eqref{eq:SDP_Maximal_Recovery_Coefficient} and \eqref{eq:SDP_Maximal_Preparation_Coefficient} become equivalent to:
\begin{equation}
\begin{aligned}
\label{eq:Symmetric_SDP_Maximal_Recovery_Coefficient}
F_D^S (\mathcal{M}_{A' \to B^n}) \coloneqq \max_{\Phi^{\mathcal{D}_n^*}_{RB^n}} \quad & \Tr\!\big(\Phi_{RB^n}^{\mathcal{M}_n}\, \Phi^{\mathcal{D}_n^*}_{RB^n}\big) \\
\textrm{s.t.} \quad & \Tr_R[\Phi^{\mathcal{D}_n^*}_{RB^n}] = \frac{\mathbbm{1}_{B^n}}{d} \quad \Phi^{\mathcal{D}_n^*}_{RB^n}\geq 0 \quad \Phi^{\mathcal{D}_n^*}_{RB^n} = \mathcal{S}_{B^n}(\Phi^{\mathcal{D}_n^*}_{RB^n})
\end{aligned}
\end{equation}
for $\mathcal{M}_{A' \to B^n} = \mathcal{N}^{\otimes n}_{A^n \to B^n} \circ \mathcal{E}_{A' \to A^n}$, and:
\begin{equation}
\begin{aligned}
\label{eq:Symmetric_SDP_Maximal_Preparation_Coefficient}
F_E^S (\mathcal{M}'_{A^n \to B'}) \coloneqq \max_{\Phi^{\mathcal{E}_n}_{RA^n}} \quad & \Tr\!\big(\Phi_{RA^n}^{\mathcal{M}'^*_n}\, \Phi^{\mathcal{E}_n}_{RA^n}\big) \\
\textrm{s.t.} \quad & \Tr_{A^n}[\Phi^{\mathcal{E}_n}_{RA^n}] = \frac{\mathbbm{1}_R}{d}, \quad \Phi^{\mathcal{E}_n}_{RA^n} \geq 0, \quad \Phi^{\mathcal{E}_n}_{RA^n} = \mathcal{S}_{A^n}(\Phi^{\mathcal{E}_n}_{RA^n} )
\end{aligned}
\end{equation}
for $\mathcal{M}'_{A^n \to B'} = \mathcal{D}_{A^n \to B'} \circ \mathcal{N}^{\otimes n}_{A^n \to B^n} $. Here, $\mathcal{S}_{S^n} (X_{S^n}) \coloneqq \frac{1}{n!} \sum_{\pi \in \mathfrak{S}_n} P_{S^n}(\pi) X_{S^n}P^\dag_{S^n}(\pi)$ denotes the symmetrizer channel. This procedure produces a lower bound $\tilde{F}^S \leq F_c(\mathcal{N}^{\otimes n}, d)$ together with a symmetric encoding--decoding pair $(\tilde{\mathcal{E}}, \tilde{\mathcal{D}})$ that achieves it.
\end{definition}
\begin{remark}
\label{rem:general_group_invariance}
Lemma \ref{prop:seesaw_symmetric} extends to any finite group $G$: if $\mathcal{N}_n$ is $G$-covariant and the seeds are $G$-invariant, then the entire iteration can be restricted to the $G$-invariant subspace, giving rise to a \textit{$G$-invariant seesaw method}\footnote{For example, if $\mathcal{N}_n = \mathcal{N}_1^{\otimes k}\otimes \mathcal{N}_2^{\otimes n-k}$ for some $k \in [n]$, then one can restrict to encoders and decoders that are invariant under the action of the subgroup $G = \mathfrak{S}_k \times \mathfrak{S}_{n-k} \leq \mathfrak{S}_n$. A similar case arises for the channels involved in superactivation \cite[Section 5.2]{Parentin2026}.} with estimate $\tilde{F}^G$. The two restrictions can also be combined when the channel enjoys both symmetries, yielding an estimate $\tilde{F}^{GS}$. 
Since each restriction shrinks the feasible set of encoders and decoders, the optimal values over these restricted sets satisfy:
\begin{equation}
    F_c^{GS}(\mathcal{N}^{\otimes n}, d) \leq F_c^{S}(\mathcal{N}^{\otimes n}, d) \leq F_c(\mathcal{N}^{\otimes n}, d)
\end{equation}
where $F_c^S$ and $F_c^{GS}$ denote the channel fidelity optimized over $\mathfrak{S}_n$- and $(G \times \mathfrak{S}_n)$-invariant codes, respectively, and the inequalities can in general be strict. Up to stochastic deviations, the seesaw estimates $\tilde{F}^{GS} \leq \tilde{F}^S \leq \tilde{F}$ are expected to inherit the same ordering, since each is a lower bound on the corresponding restricted optimum. None of these restrictions is guaranteed to preserve the global optimum $F_c$ (see the discussion in Section~\ref{sec:perm-inv-codes}), and the seesaw method itself is inherently dependent on the random seed. The advantage of these restrictions is computational: \textit{by accepting possible sub-optimality, we can push the number of channel uses $n$ much further} than the unrestricted seesaw. The numerical results (Section~\ref{sec:numerical_examples} and~\cite{Parentin2026}) show that this trade-off is highly favorable in practice.
\end{remark}
The symmetric seesaw method defined in~\ref{def:symmetric-seesaw} provides a lower bound $\tilde{F}^S$ to the optimal fidelity~\eqref{eq:channel_fidelity} over $n$ channel uses of $\mathcal{N}$, where the maximization is restricted to PI codes. Differently from the channel fidelity, which is monotonic in the number of channel uses (cf.~\eqref{eq:monotonicity-channel-fidelity}), the symmetry constraint restricts the structure of admissible codes, and forbids strategies like one that disregards some channel uses (which may be favorable e.g. in the presence of strong noise). Therefore, the sequence of estimated fidelities $\{\tilde{F}^S_n\}_{n  \leq N}$, where $N\in \mathbb{N}$ is the maximum number of channel uses, need not be monotonic. To mitigate this issue in numerical simulations, one may instead define the reported performance as
\begin{equation}
    \tilde{F}^S \coloneqq \max_{n \in [N]} F_e\!\big(\tilde{\mathcal{D}}_n \circ \mathcal{N}^{\otimes n}, \tilde{\mathcal{E}}_n\big),
\end{equation}
where \(N\) is the maximum number of channel uses considered. Since the seesaw method is stochastic, for each \(n \in [N]\) the procedure should be repeated multiple times, and \(\tilde{F}_n\) taken as the maximum fidelity obtained.

\begin{remark}
\label{rem:extendibility_link}
The symmetry conditions~\eqref{eq:symmetric_enc-dec} are equivalent to requiring that the one-copy reduced maps
\[
\mathcal{E}^{(1)} = \mathrm{Tr}_{A^{n-1}} \circ \mathcal{E}_n,
\qquad
\mathcal{D}^{*(1)} = \mathrm{Tr}_{B^{n-1}} \circ \mathcal{D}^*_n
\]
are $n$-extendible \cite{Pankowski2013,Kaur2019}. It is important to distinguish this from the case in which the channel $\mathcal{M}$ is $k$-extendible. If $\mathcal{M}$ is $k$-extendible, then (see~\eqref{eq:Fidelity_recovery_k_extendible}) its maximal fidelity of recovery is upper bounded by\begin{equation}
\label{eq:bound_k_extendible}
    F_D(\mathcal{M}) \leq \frac{k + d - 1}{dk},
\end{equation} reflecting the fact that $k$-extendible channels are of limited use for quantum communication. By contrast, in~\eqref{eq:Symmetric_SDP_Maximal_Recovery_Coefficient} we maximize the fidelity of recovery over $n$-extendible \emph{recovery maps}, i.e., over the same feasible set, but with objective function\begin{equation}
    \Tr\!\big(\Phi^{\mathcal{D}_n^*}_{RB^n}\,\Phi^{\mathcal{M}_n}_{RB^n}\big),
\end{equation}
which is very different from the single-copy quantity\begin{equation}
    \Tr\!\big(\Phi^{\mathcal{D}^{*(1)}}_{RB}\,\Phi^{\mathcal{M}^{(1)}}_{RB}\big),
\end{equation}
which would instead satisfy the upper bound in~\eqref{eq:bound_k_extendible} for $n = k$. In the limit $n \to \infty$, de Finetti theorems for quantum channels~\cite{Berta2019} imply that the one-copy reduced maps $\mathcal{E}^{(1)}$ and $\mathcal{D}^{*(1)}$ become asymptotically entanglement-breaking, consistently with the bound~\eqref{eq:bound_k_extendible} approaching $1/d$ (see Lemma \ref{lem:max_recovery_properties}). However, this asymptotic behavior does \emph{not} constrain the $n$-copy objective, which is the relevant quantity in the symmetric seesaw. For this reason, the quantities studied here are not directly related to $n$-unextendibility measures for quantum channels, such as those considered in~\cite{Kaur2021,Singh2025}, even though they involve optimization over the same feasible sets.
\end{remark}

\subsubsection{Efficient Solution of the Symmetric Seesaw Method} Using the tools of Section \ref{sec:perm-inv}, we now show that the symmetric seesaw method has an efficient (poly($n$)) solution. In the following, we'll always work with the (unnormalized) Choi matrices of the involved channels, better suited for expressing the intermediate concatenation steps \eqref{eq:concatenation_choi}, and we scale the SDPs \eqref{eq:Symmetric_SDP_Maximal_Recovery_Coefficient} and \eqref{eq:Symmetric_SDP_Maximal_Preparation_Coefficient} accordingly. 
The Choi matrices of the channel, encoder, and decoder involved in the symmetric seesaw method are parametrized in the orbit basis as:
\begin{align}
 \label{eq:choi_channel}
    \Gamma^{\mathcal{N}^{\otimes n}}_{A^n B^n} &= \sum_{s= 1}^{m_{AB}}  c_{s}^{\mathcal{N}} C^{AB}_{s}, \\
    \label{eq:choi_encoder}
        \Gamma^{\mathcal{E}_{n}}_{R A^n}& =  \sum_{l,k= 1}^{d_{R}} \sum_{r= 1}^{m_{A}}  c_{l,k,r}^{\mathcal{E}} \ket{l}\bra{k} \otimes C^{A}_{r},\\
            \label{eq:choi_decoder}
        \Gamma^{\mathcal{D}_{n}}_{B^nR} &= \sum_{t= 1}^{m_{B}} \sum_{k,l= 1}^{d_{R}}   c_{t,k,l}^{\mathcal{D}} C^{B}_{t}\otimes \ket{k}\bra{l},
\end{align}
where:
\begin{equation}
\label{eq:dimensions_perm-inv}
    m_A = \binom{n + d_A^2 - 1}{n} ,\quad m_B =\binom{n + d_B^2 - 1}{n}, \quad m_{AB} =\binom{n + d_A^2 d_B^2 - 1}{n}.
\end{equation}
During the iteration, we'll also need to build the 'intermediate' channels $\mathcal{M}_n = \mathcal{N}^{\otimes n} \circ \mathcal{E}_n$, s.t. $\Gamma^{\mathcal{M}_n}_{R B^n} \in \mathrm{End}^{\mathfrak{S}_n}(R \otimes B^{\otimes n})$ and  $\mathcal{M}'_n =\mathcal{D}_n  \circ  \mathcal{N}^{\otimes n} $ s.t. $ \Gamma^{\mathcal{M}'_{n}}_{A^nR}\in \mathrm{End}^{\mathfrak{S}_n}(A^{\otimes n} \otimes R)$, respectively. To fix the notation, we write their parametrization in the orbit basis as:
\begin{align}
 \label{eq:choi_M_Bob_pov}
    \Gamma^{\mathcal{M}_n}_{R B^n} &= \sum_{l,k= 1}^{d_{R}} \sum_{t= 1}^{m_{B}}  c_{l,k,t}^{\mathcal{M}} \ket{l}\bra{k} \otimes C^{B}_{t} ,\\
    \label{eq:choi_M_Alice_pov}
        \Gamma^{\mathcal{M}'_{n}}_{A^nR}& = \sum_{r= 1}^{m_{A}} \sum_{k,l= 1}^{d_{R}}   c_{r,k,l}^{\mathcal{M}'} C^{A}_{r}\otimes \ket{k}\bra{l}.
\end{align}
We can now state the following theorem, which proves that the symmetric seesaw method can be solved efficiently, i.e. using a number of variables and constrains of polynomial size in $n$. The proof will lay out all the required steps to efficiently implement the algorithm using the tools of Section \ref{sec:perm-inv}, and in Section \ref{sec:Python_package} we briefly comment on the Python implementation we provide of this method.
\begin{theorem}
\label{thm:seesaw_polytime}
    Let $d \in \mathbb{N}$ and let $\mathcal{N}_{A \to B}$ be a quantum channel with input dimension $d_A$ and output dimension $d_B$. Given $n \in \mathbb{N}$, consider the $n$-shot entanglement transmission setting of Figure \ref{fig:quantumcommunicationquantumchannel}. Then, for a finite number of iterations $i\in \mathbbm{N}$ the symmetric seesaw method gives a $\mathrm{poly}(n)$-computable lower bound on the channel fidelity. 
\end{theorem}
\begin{proof}
The proof is constructive and shows how to efficiently implement an entire iteration cycle of the symmetric seesaw (see Definition \ref{def:seesaw}) phrasing all steps in the orbit basis (see Section \ref{sec:orbit_basis_operations}), where the orbit basis coefficients of encoder \eqref{eq:choi_encoder} and decoder \eqref{eq:choi_decoder} play the role of optimizing variables in the SDPs.

\emph{Initialization $(i = 0)$}. 
The coefficients $\{c_{s}^{\mathcal{N}}\}_{s\in [m_{AB}]}$ of the i.i.d.\ channel $\mathcal{N}^{\otimes n}$ in \eqref{eq:choi_channel} can be directly obtained by applying~\eqref{eq:tensor_product_orbit} to the channel's (single-copy) Choi matrix $\Gamma^{\mathcal{N}}_{AB}$. The coefficients of the symmetric seeds $\{c_{k,l,r}^{\mathcal{E}_0} \}_{r\in [m_{A}], k,l \in [d_R]}$ and $\{c_{t,k,l}^{\mathcal{D}_0} \}_{t\in [m_{B}], k,l \in [d_R]}$ are generated by sampling random CPTP maps directly in the block basis. Concretely, using the full isomorphism~\eqref{eq:full_isomorphism_with_reference},  for each $\lambda \in \mathrm{Par}(d_A, n)$ we sample the block $[\widetilde{\psi}_{RA^n}(\Gamma^{\mathcal{E}_0}_{RA^n})]_\lambda \in \mathbb{C}^{(d_R m_\lambda) \times (d_R m_\lambda)}$ separately from an ensemble of Choi matrices of CPTP maps (see e.g. \cite{Bruzda2009}), weighted across blocks by a probability vector drawn from the uniform Dirichlet distribution on $|\mathrm{Par}(d_A, n)|$ entries. The trace-preserving constraint~\eqref{eq:choi_constraint_block_CPTP} is then enforced by rescaling each block by $f_\lambda^{-1}$, with $f_\lambda$ the Specht dimension. The orbit-basis coefficients are then recovered by applying the inverse isomorphism $\widetilde{\psi}_{RA^n}^{-1}$. This gives orbit basis coefficients of a valid Choi matrix in $\esn(R \otimes A^n) \subset \mathcal{L}(R \otimes A^n)$. The decoder seed $\mathcal{D}^*_0$ is built analogously in the block basis of $\mathrm{End}^{\mathfrak{S}_n}(R \otimes B^{\otimes n})$ (enforcing the unitality constraint \eqref{eq:choi_constraint_block_CPU}) and converted to Choi-matrix coefficients via the adjoint relation of Lemma~\ref{lem:transpose_orbit_basis}.

\emph{Iteration $(i > 0)$}. 
We now show how to efficiently implement an entire round of the symmetric seesaw iteration. By Lemma \ref{lemma:seesaw_symmetric}, the entire iteration can be carried out in the permutation-invariant subspace. Given the coefficients $\{c_{s}^{\mathcal{N}}\}_{s\in [m_{AB}]}$ and  $\{c_{k,l,r}^{\mathcal{E}_i} \}_{r\in [m_{A}], k,l \in [d_R]}$, we can efficiently compute the coefficients $\{c_{k,l,t}^{\mathcal{M}_i} \}_{t\in [m_{B}], k,l \in [d_R]}\}$ in \eqref{eq:choi_M_Bob_pov} using Proposition \ref{prop:serial_concatenation} (specifically, Equation \eqref{eq:concatenation_coefficients_BobPOV}). Then, in order to efficiently solve the SDP \eqref{eq:Symmetric_SDP_Maximal_Recovery_Coefficient}, we can use item 1 of Lemma \ref{lemma:trace_HS_orbit} to express the objective function and the map $\psi_{RB^n}$ introduced in \eqref{eq:isomorphism_with_reference} to rewrite the PSD constraint and the partial trace constraint. This takes the $F_D^S(\mathcal{M}_i)$ to the form:
\begin{equation}
\label{eq:Symmetric_SDP_Maximal_Recovery_Coefficient_newbasis_isomorphism}
\begin{aligned}
    F^S_{\mathrm{D}}(\mathcal{M}_n) = \max_{c_{k,l,t}^{\mathcal{D}^*} \in \mathbb{C}} \quad & \frac{1}{d^2} \sum_{k,l=1}^{d_R} \sum_{t=1}^{m_B} \overline{c_{k,l,t}^{\mathcal{D}^*}} c_{k,l,t}^{\mathcal{M}} \|C_t^B\|_{\mathrm{HS}}^2 \\
    \text{subject to:} \quad 
    & \sum_{k,l=1}^{d_R} \sum_{t=1}^{m_B} c_{k,l,t}^{\mathcal{D}^*} \ket{k}\bra{l}_R \otimes U_\lambda^T C_t^{B} U_\lambda \geq 0  \ \ \ \ &\forall \lambda \in \text{Par}(d_B,n),\\
    & \sum_{k=1}^{d_R} \sum_{t=1}^{m_B} c_{k,k,t}^{\mathcal{D}^*} U_\lambda^T C_t^{B} U_\lambda= G_\lambda \ &\forall \lambda \in \text{Par}(d_B,n),
\end{aligned}
\end{equation}
where we used $[\psi_{B^n}(\mathbbm{1}_{B^n})]_\lambda = G_\lambda$ (see~\eqref{eq:gram_matrix_as_action_on_identity}) and the $U_λ$ are the ones implementing $\psi_{RB^n}$ from \eqref{eq:isomorphism_with_reference}. Since this $\psi_{RB^n}$ is a bijective linear map that preserves positive semidefiniteness, the two SDPs~\eqref{eq:Symmetric_SDP_Maximal_Recovery_Coefficient} and~\eqref{eq:Symmetric_SDP_Maximal_Recovery_Coefficient_newbasis_isomorphism} have the same feasible set and the same optimal value. The products $U_\lambda^T C_t^{B} U_\lambda $ can be implemented efficiently by Theorem \ref{thm:Efficient_construction_Isomorphism}. The SDP \eqref{eq:Symmetric_SDP_Maximal_Recovery_Coefficient_newbasis_isomorphism} has a number of variables equal to $O(d_R^2(n+1)^{d_B^2})$, a number of constraints equal to $O((n+1)^{d_B})$ of size equal to $O(d_R(n+1)^{d_B(d_B-1)/2})$, and is strictly feasible, and so it can be solved in time $\mathrm{poly}(n)$ (see e.g. \cite[Section 1.9]{deKlerk2002}, and references therein). This gives $F_D(\mathcal{M}_i)$ and the coefficients of the adjoint of the symmetric decoder map $\mathcal{D}_{i+1}^*$. To obtain the coefficients \eqref{eq:choi_decoder}, we can simply use Lemma \ref{eq:adjoint_choi} to get\begin{equation}
c_{t,k,l}^{\mathcal{D}_{i+1}} = c_{l,k,t^T}^{\mathcal{D}_{i+1}^*} \quad \forall t \in [m_B], l,k \in [d_R],
    \end{equation}
where $t^T$ is defined as in \eqref{eq:count-orbit-transpose-relation}. Then, we can efficiently compute the coefficients $\{c_{r,k,l}^{\mathcal{M}'_i} \}_{r\in [m_{A}], k,l \in [d_R]}$ in \eqref{eq:choi_M_Alice_pov} using Proposition \ref{prop:serial_concatenation} (specifically, Equation \eqref{eq:concatenation_coefficients_AlicePOV}). Before solving the SDP \eqref{eq:Symmetric_SDP_Maximal_Preparation_Coefficient}, we need to take the adjoint $\mathcal{M'}_i^*$, and we use again Lemma \ref{eq:adjoint_choi} to get the corresponding coefficients \begin{equation}
        c_{l,k,r}^{\mathcal{M}'^*} = c_{r^T,k,l}^{\mathcal{M}'}  \quad \forall r \in [m_A], l,k \in [d_R].
    \end{equation}
Using the positive map $\psi_{RA^n}$ introduced in \eqref{eq:isomorphism_with_reference}, we rewrite the SDP in \eqref{eq:Symmetric_SDP_Maximal_Preparation_Coefficient} as
\begin{equation}
\label{eq:Symmetric_SDP_Maximal_Preparation_Coefficient_newbasis_isomorphism}
\begin{aligned}
    F^S_{\mathrm{E}}(\mathcal{M}'_n) = \max_{c_{k,l,r}^{\mathcal{E}} \in \mathbb{C}} \quad & \frac{1}{d^2} \sum_{k,l=1}^{d_R} \sum_{r=1}^{m_A} \overline{c_{k,l,r}^{\mathcal{M}'^*}} c_{k,l,r}^{\mathcal{E}} \|C_r^A\|_{\mathrm{HS}}^2 \\
    \text{subject to:} \quad 
    &  \sum_{k,l=1}^{d_R} \sum_{r=1}^{m_A} c_{k,l,r}^{\mathcal{E}} \ket{k}\bra{l}_R \otimes U_\lambda^T C_{r}^{A}U_\lambda \geq 0 \ \ \ \ &\forall \lambda \in \text{Par}(d_A,n), \\
    &  \sum_{r=1}^{m_A} c_{k,l,r}^{\mathcal{E}} \Tr[C_r^A] \ket{k}\bra{l} =\mathbbm{1}_A.
\end{aligned}
\end{equation}
Note that to efficiently express the partial trace constraint in \eqref{eq:Symmetric_SDP_Maximal_Preparation_Coefficient_newbasis_isomorphism} we use item 2 of Lemma \ref{lemma:trace_HS_orbit}. Similarly to \eqref{eq:Symmetric_SDP_Maximal_Recovery_Coefficient_newbasis_isomorphism}, this SDP has variables and constraints in number and in size polynomial in $n$, so it admits a $\mathrm{poly}(n)$ solution. This gives $F_E(\mathcal{M}'_i)$ and the coefficients of the symmetric encoder map $\mathcal{E}_{i+1}$.
This completes one cycle of the symmetric seesaw iteration. If \begin{equation}
    F_{\mathrm{E}}(\mathcal{M}_i') - F_{\mathrm{D}}(\mathcal{M}_i) < \delta,
\end{equation}
the coefficients of the encoder and decoder's Choi matrix in the orbit basis are returned together with $\tilde{F}^S \equiv F_{\mathrm{E}}(\mathcal{M}_i')$, otherwise, the iteration continues. Since each round consists of a sequence of $\mathrm{poly}(n)$ operations, for all finite $i \in \mathbb{N}$ the symmetric seesaw method gives a $\mathrm{poly}(n)$-computable lower bound on the channel fidelity.
\end{proof}

In practice, since the symmetric seesaw method solves $F_D$ and $F_E$ repeatedly for varying input channels, one can precompute the change-of-basis map from orbits to blocks (determined by the \eqref{eq:complete_matrix_after_isomorphism}) and the Gram matrices $G_\lambda$, and then solve efficiently \eqref{eq:Symmetric_SDP_Maximal_Recovery_Coefficient_newbasis_isomorphism} and \eqref{eq:Symmetric_SDP_Maximal_Preparation_Coefficient_newbasis_isomorphism} with a fixed overhead (polynomial in $n$), that consists of a matrix multiplication.

\subsubsection{Block Basis and Symmetric Power Method}
\label{sec:symmetric_power_method}
While Theorem (\ref{thm:seesaw_polytime}) serves as theoretical validation of the symmetric seesaw method, it is generally much more convenient to solve the optimization problems \eqref{eq:Symmetric_SDP_Maximal_Recovery_Coefficient} and \eqref{eq:Symmetric_SDP_Maximal_Preparation_Coefficient} directly in the block basis, for the reasons explained in Section \ref{sec:efficient-operations-in-block-basis}. In this application, this becomes mandatory if we want to use the symmetric version of the power iteration method (Definition \ref{def:power_method}). 

The reduced SDPs~\eqref{eq:Symmetric_SDP_Maximal_Recovery_Coefficient_newbasis_isomorphism} and~\eqref{eq:Symmetric_SDP_Maximal_Preparation_Coefficient_newbasis_isomorphism} use the orbit-basis coefficients $\{c_{k,l,t}^{\mathcal{D}^*}\}$ and $\{c_{k,l,r}^{\mathcal{E}}\}$ as variables, with the PSD constraint expressed via the blocks $[C_t^B]_\lambda = U_\lambda^T C_t^B U_\lambda$ and $[C_r^A]_\lambda = U_\lambda^T C_r^A U_\lambda$. To phrase the SDPs entirely in the block basis we need the full $*$-isomorphism~\eqref{eq:full_isomorphism_with_reference} rather than the PSD-preserving map \eqref{eq:isomorphism_with_reference}.

Given an operator $X_{RS^n} \in \mathrm{End}^{\mathcal{S}_n}(R\otimes S^{\otimes n})$, with
orbit coefficients $\{x_{k,l,r}\}$ (as the ones involved in the symmetric seesaw method), we will denote the $\lambda$-block, for $\lambda \in \mathrm{Par}(d_S, n)$, as\begin{equation}
\widetilde{[X]}_\lambda \equiv [\widetilde{\psi}_{RS^n}(X_{RS^n})]_\lambda= (\mathbbm{1}_R \otimes R_\lambda^T)\,
 [\psi_{RS^n}(X)]_\lambda (\mathbbm{1}_R \otimes R_\lambda).
\end{equation} 
Every $\lambda-$block $\widetilde{[X]}_\lambda$ has dimension $(d_Rm_\lambda , d_Rm_\lambda)$, and can be expressed as:\begin{equation}
    \widetilde{[X]}_\lambda = \sum_{k,l=1}^{d_R} \sum_{r=1}^{m_S}x_{k,l,r} \ket{k}\!\bra{l}\otimes  \widetilde{[C_r^S]}_\lambda \eqqcolon \sum_{k,l=1}^{d_R} \ket{k}\!\bra{l}\otimes \widetilde{[X]}_\lambda^{(k,l)},
\end{equation}
where we denoted $\widetilde{[X]}_\lambda^{(k,l)}\equiv \sum_{r=1}^{m_S} x_{k,l,r}\, \widetilde{[C_r^S]}_\lambda$.
We can then use all results of Section \ref{sec:efficient-operations-in-block-basis} to rewrite the SDPs \eqref{eq:Symmetric_SDP_Maximal_Recovery_Coefficient_newbasis_isomorphism} and~\eqref{eq:Symmetric_SDP_Maximal_Preparation_Coefficient_newbasis_isomorphism} in the block basis. Specifically, we use \eqref{eq:inner_product_HS_block} to rewrite the objective functions, \eqref{eq:choi_constraint_block_CPU} to express the unitality constraint for $F_D$, and \eqref{eq:choi_constraint_block_CPTP} to express the trace-preservation constraint for $F_E$. The resulting SDPs are:
\begin{equation}
\label{eq:block_SDP_FD}
\begin{aligned}
     F^S_{\mathrm{D}}(\mathcal{M}_n) = \max_{\{
    \widetilde{[\Gamma^{\mathcal{D}_n^*}]}_\lambda\}: \   \widetilde{[\Gamma^{\mathcal{D}_n^*}]}_\lambda \in \mathbb{C}^{d_R m_\lambda \times d_R m_\lambda } } \quad
    & \frac{1}{d^2} \sum_\lambda f_\lambda\,
    \mathrm{Tr}\!\left[
    \widetilde{[\Gamma^{\mathcal{D}_n^*}]}_\lambda \cdot
    \widetilde{[\Gamma^{\mathcal{M}_n}]}_\lambda\right] \\
    \text{s.t.} \quad
    & \widetilde{[\Gamma^{\mathcal{D}_n^*}]}_\lambda \geq 0
    \quad \forall\, \lambda \in \mathrm{Par}(d_B, n) \\
    & \sum_{k=1}^{d_R}
    \widetilde{[\Gamma^{\mathcal{D}_n^*}]}_\lambda^{(k,k)}
    = \mathbbm{1}_{m_\lambda}
    \quad \forall\, \lambda \in \mathrm{Par}(d_B, n),
\end{aligned}
\end{equation}
\begin{equation}
\label{eq:block_SDP_FE}
\begin{aligned}
 F^S_{\mathrm{E}}(\mathcal{M}'_n) = \max_{\{
    \widetilde{[\Gamma^{\mathcal{E}_n}]}_\lambda\}: \   \widetilde{[\Gamma^{\mathcal{E}_n}]}_\lambda \in \mathbb{C}^{d_R m_\lambda \times d_R m_\lambda } }\quad
    & \frac{1}{d^2} \sum_\lambda f_\lambda\,
    \mathrm{Tr}\!\left[
    \widetilde{[\Gamma^{\mathcal{E}_n}]}_\lambda \cdot
    \widetilde{[\Gamma^{\mathcal{M}'^*_n}]}_\lambda\right] \\
    \text{s.t.} \quad
    & \widetilde{[\Gamma^{\mathcal{E}_n}]}_\lambda \geq 0
    \quad \forall\, \lambda \in \mathrm{Par}(d_A, n) \\
    & \sum_\lambda f_\lambda\,
    \mathrm{Tr}_{V_\lambda}\!\left(
    \widetilde{[\Gamma^{\mathcal{E}_n}]}_\lambda\right) =\mathbbm{1}_R.
\end{aligned}
\end{equation}

The variables in these SDPs are the block matrices $\widetilde{[\Gamma^{\mathcal{D}^*}]}_\lambda$ (resp.\ $\widetilde{[\Gamma^{\mathcal{E}}]}_\lambda$), each of size $(d_R m_\lambda, d_R m_\lambda)$, rather than the complex orbit-basis coefficients $\{c_{k,l,t}^{\mathcal{D}^*}\}$ (resp.\ $\{c_{k,l,r}^{\mathcal{E}}\}$). The total number of variables is the same, but all constraints and objective functions decouple into independent blocks, giving a significant computational advantage over the orbit-basis SDPs~\eqref{eq:Symmetric_SDP_Maximal_Recovery_Coefficient_newbasis_isomorphism} and~\eqref{eq:Symmetric_SDP_Maximal_Preparation_Coefficient_newbasis_isomorphism}. Therefore, in applications the block-basis implementation of the SDPs in \eqref{eq:block_SDP_FD} and \eqref{eq:block_SDP_FE} is preferable. Since the intermediate steps (i.e., concatenation and adjoint) are performed in the orbit basis, we need to move back and forth from the block to the orbit-basis via~\eqref{eq:full_isomorphism_with_reference} and its inverse. These interconversions can be implemented with a fixed overhead by precomputing the change-of-basis matrices prior to the iteration.

Since $\widetilde{\psi}$ is a $*$-algebra isomorphism, the channel power iteration of Definition~\ref{def:power_method} translates directly to block-wise operations, thus providing an alternative to SDP solvers for computing \eqref{eq:block_SDP_FD} and \eqref{eq:block_SDP_FE}. We describe this for the decoder case~\eqref{eq:block_SDP_FD}; the encoder case~\eqref{eq:block_SDP_FE} is analogous, up to the appropriate change in the normalization constraint.

\begin{definition}
\label{def:symmetric_power_method_D}
Let $\mathcal{M}_{A' \to B^n}$ be a symmetric channel with $\Gamma^{\mathcal{M}}_{RB^n} \in \mathrm{End}^{\mathfrak{S}_n}(R \otimes B^{\otimes n})$, $\mathcal{D}_0 \in \CPTP(B^n \to B')$ a symmetric initial decoder, and $\delta_p > 0$ a convergence threshold. The symmetric channel power iteration produces an estimate of $F^S_{\mathrm{D}}(\mathcal{M})$ and a symmetric decoder $\tilde{\mathcal{D}}_{j+1}$ that achieves it via the following iterative algorithm. Set $j = 0$ and repeat:
\begin{enumerate}
    \item \textbf{Sandwich.} For each $\lambda \in \mathrm{Par}(d_B, n)$, compute the unnormalized $\lambda$-block:
    \begin{equation}
    \label{eq:power_step_block}
       \widetilde{[X^{(j)}]}_\lambda \coloneqq \widetilde{[\Gamma^{\mathcal{M}}]}_\lambda \cdot \widetilde{[\Gamma^{\mathcal{D}^*_j}]}_\lambda \cdot \widetilde{[\Gamma^{\mathcal{M}}]}_\lambda.
    \end{equation}
    \item \textbf{Normalize.} For each $\lambda$, compute $\widetilde{[X_R^{(j)}]}_\lambda \equiv \sum_{k=1}^{d_R} \widetilde{[X^{(j)}]}_\lambda^{(k,k)} \in \mathbb{C}^{m_\lambda \times m_\lambda}$ and obtain the $\lambda$-block of the normalized decoder $\mathcal{D}^*_{j+1}$ by applying $\bigl(\widetilde{[X_R^{(j)}]}_\lambda\bigr)^{-1/2}$ (pseudoinverse), i.e.:
    \begin{equation}
    \label{eq:normalization_choi_block}
        \widetilde{[\Gamma^{\mathcal{D}^*_{j+1}}]}_\lambda \coloneqq \bigl(\mathbbm{1}_{d_R} \otimes \widetilde{[X_R^{(j)}]}_\lambda^{-1/2}\bigr)\, \widetilde{[X^{(j)}]}_\lambda \, \bigl(\mathbbm{1}_{d_R} \otimes \widetilde{[X_R^{(j)}]}_\lambda^{-1/2}\bigr).
    \end{equation}
    This enforces the unitality constraint~\eqref{eq:choi_constraint_block_CPU} blockwise: $\sum_{k=1}^{d_R}\widetilde{[\Gamma^{\mathcal{D}^*_{j+1}}]}_\lambda^{(k,k)} = \mathbbm{1}_{m_\lambda}$ for each $\lambda \in \mathrm{Par}(d_B, n)$.
    \item \textbf{Convergence check.} Compute the entanglement fidelity in block basis via~\eqref{eq:inner_product_HS_block}:
    \begin{equation}
        F_e^{(j+1)} \coloneqq F_e(\mathcal{D}_{j+1} \circ \mathcal{M}) = \frac{1}{d^2} \sum_{\lambda \in \mathrm{Par}(d_B, n)} f_\lambda\, \Tr\!\left[\widetilde{[\Gamma^{\mathcal{M}}]}_\lambda \cdot \widetilde{[\Gamma^{\mathcal{D}^*_{j+1}}]}_\lambda\right].
    \end{equation}
    If $F_e^{(j+1)} - F_e^{(j)} < \delta_p$, terminate. Otherwise, set $j \leftarrow j+1$ and repeat.
\end{enumerate}
Return $\tilde{\mathcal{D}} \coloneqq \mathcal{D}_{j+1}$ and $\tilde{F}^S \coloneqq F_e^{(j+1)}$.
\end{definition}

The same algorithm can be applied to estimate $F^S_E$, with the only difference in the normalization step. The trace-preservation constraint~\eqref{eq:choi_constraint_block_CPTP} couples all blocks through the Specht dimensions $f_\lambda$, so the local normalization~\eqref{eq:normalization_choi_block} must be replaced by sandwiching each block with $\bigl(T^{(j)}\bigr)^{-1/2}$ on the $R$ system, where:
\begin{equation}
\label{eq:encoder_global_trace}
    T^{(j)} \coloneqq \sum_{\lambda \in \mathrm{Par}(d_A,n)} f_\lambda\, \mathrm{Tr}_{V_\lambda}\!\bigl[\widetilde{[X^{(j)}]}_\lambda\bigr] \in \mathbb{C}^{d_R \times d_R}
\end{equation}
is a global $d_R \times d_R$ matrix obtained by aggregating partial traces across all blocks. The normalization step for the encoder thus reads:
\begin{equation}
   \widetilde{[\Gamma^{\mathcal{E}_{j+1}}]}_\lambda \coloneqq \bigl((T^{(j)})^{-1/2} \otimes \mathbbm{1}_{m_\lambda}\bigr)\, \widetilde{[X^{(j)}]}_\lambda \, \bigl((T^{(j)})^{-1/2} \otimes \mathbbm{1}_{m_\lambda}\bigr),
\end{equation}
which enforces $\sum_{\lambda} f_\lambda\, \mathrm{Tr}_{V_\lambda}\!\bigl[\widetilde{[\Gamma^{\mathcal{E}_{j+1}}]}_\lambda\bigr] = \mathbbm{1}_R$, i.e.~\eqref{eq:choi_constraint_block_CPTP}.

Each block has size at most $(d_R m_\lambda) \times (d_R m_\lambda)$, with $m_\lambda = O((n+1)^{d_S(d_S-1)/2})$ by~\eqref{eq:bound_SSYTs}, so the matrix multiplications in~\eqref{eq:power_step_block} and~\eqref{eq:normalization_choi_block} all involve polynomial-size dense matrices. 
\begin{remark}
\label{rem:parallelization}
The block decomposition makes the symmetric power iteration particularly well suited for GPU acceleration: since the blocks are independent (i.e., the computation is embarrassingly parallel for $F_{\mathrm{D}}$, and nearly so for $F_{\mathrm{E}}$ up to the global aggregation~\eqref{eq:encoder_global_trace}), all loops over $\lambda$ can be parallelized (e.g., using the \texttt{CuPy} package \cite{CuPy2023}). In practice, this makes the power iteration significantly faster than any SDP solver for large $n$.
\end{remark}
For these reasons, the symmetric seesaw method combined with power method described above is the best solution in practice for our examples; in Section \ref{sec:numerical_examples} we apply it to the qubit depolarizing and amplitude-damping channels for up to $n = 20$ channel uses, giving the tightest lower bounds on their channel fidelity in a wide parameter range.  Moreover, in \cite{Parentin2026}, a variant of the method was used to establish for the first time non-asymptotic superactivation of quantum capacity. In the next section we this variant in full generality, extending the symmetric seesaw to the case of flagged channels.

\subsection{Seesaw Method for Flagged Channels}
\label{sec:modified_seesaw_flagged}

We now extend the symmetric seesaw method of Section~\ref{sec:symmetric_seesaw} to channels of \emph{flagged} form. Given $\ell$ flag values, and a corresponding output algebra
\begin{equation}
    \A_B = \bigoplus_{j = 1}^{\ell} \mathcal{L}(\HS_{B^{(j)}}),
\end{equation}
a flagged channel $\mathcal{N}_{A\to \A_B}$ acts as
\begin{equation}
\label{eq:flagged_channel_seesaw}
    {\mathcal{N}}_{A\to \A_B}(\rho) = \sum_{i=1}^\ell p_i\, \ket{i}\!\bra{i} \otimes \mathcal{N}^i_{A\to B^{(i)}}(\rho),
\end{equation}
with $\{\mathcal{N}^i\}_{i\in [\ell]} \subset \CPTP(A\to B^{(i)})$ and probability distribution $\{p_i\}_{i\in[\ell]}$. Hence, the Choi matrix $\Gamma^{\mathcal{N}}_{A B}$ of $\mathcal{N}$ is an element of $\mathcal{L}(\HS_A) \otimes \A_B$, and the Choi matrix $\Gamma^{\mathcal{N}\n}_{A^n B^n}$ of the $n$-fold tensor product channel $\mathcal{N}^{\otimes n}$ is an element of $
\esn((\mathcal{L}(\HS_A) \otimes \A_B)\n)$, falling exactly within the framework of Section~\ref{sec:cq-algebras}.

For a matrix $\Gamma^{\mathcal{M}_n}_{R B^n} \in \esn(\mathcal{L}(\HS_R) \otimes (\A_B)\n)$ it is easy to see that in the optimization problem for 
\begin{equation}
    F_D(\mathcal{M}_n) = \displaystyle {1 \over d^2} \max_{\Gamma_{RB}^{\mathcal{D}^*_n} \in  \mathcal{C^*}(R: B^n)} \Tr[\Gamma_{RB^n}^{\mathcal{D}^*_n}\Gamma_{RB^n}^{\mathcal{M}_n}]
\end{equation}
one can restrict $\Gamma^{\mathcal{D}^*_n}_{RB^n}$ to be an element of $\esn(\mathcal{L}(\HS_R) \otimes (\A_B)\n)$ (the restriction to $\mathcal{L}(\HS_R) \otimes (\A_B)\n$ follows since elements outside of the block-diagonal structure of $\A_B\n$ just don't contribute in the trace, and the restriction to $\esn(\mathcal{L}(\HS_R) \otimes (\A_B)\n)$ follows by Lemma \ref{lem:max_recovery_properties}).
As seen in Table \ref{tab:CQ_dimension_comparison}, taking into account this direct sum structure of the decoder significantly reduces the dimension of the spaces to optimize over.

\subsubsection{Efficient Seesaw Iteration}
\label{sec:modified_seesaw_efficient}

The seesaw iteration then proceeds very much analogously to the previous section, just applying the results of Section~\ref{sec:cq-algebras} for the algebra $\A_B$: 

\smallskip
\noindent\emph{Initialization.}
The initialization proceeds in the exact same way as in Theorem \ref{thm:seesaw_polytime}.

\smallskip
\noindent\emph{Iteration.}
We apply the block-diagonalization of $\A_B$ from Theorem \ref{thm:block_diag_general_algebra}, where each block is indexed by a partition tuple $\underline{\lambda} = (\lambda_1,\ldots,\lambda_\ell)$ with $\lambda_j\in \Par(d_B,\mu_j)$, and we have $f_{\underline{\lambda}} = \prod_j f_{\lambda_j}$ and $m_{\underline{\lambda}} = \prod_j m_{\lambda_j}$, together with the trace relations of~\eqref{eq:HS_inner_block_general}--\eqref{eq:partial_trace_general_block}. This allows us to then write the decoder SDP as
\begin{equation}
\label{eq:block_SDP_FD_flagged_generalized}
\begin{aligned}
F^S_{\mathrm{D}}(\mathcal{M}_n) = \max_{ \{ \widetilde{[\Gamma^{\mathcal{D}_n^*}]}_{\underline{\mu}, \underline{\lambda}} \} } \quad
& \frac{1}{d^2} \sum_{\underline{\mu}: |\underline{\mu}|=n} \sum_{\underline{\lambda} \in \mathrm{Par}(d_B, \underline{\mu})} f_{\underline{\lambda}} \,
\mathrm{Tr}\!\left[
\widetilde{[\Gamma^{\mathcal{D}_n^*}]}_{\underline{\mu}, \underline{\lambda}} \cdot
\widetilde{[\Gamma^{\mathcal{M}_n}]}_{\underline{\mu}, \underline{\lambda}}\right] \\
\text{s.t.} \quad
& \widetilde{[\Gamma^{\mathcal{D}_n^*}]}_{\underline{\mu}, \underline{\lambda}} \geq 0
\quad \forall\, \underline{\mu}, \underline{\lambda} \\
& \sum_{k=1}^{d_R}
\widetilde{[\Gamma^{\mathcal{D}_n^*}]}_{\underline{\mu}, \underline{\lambda}}^{(k,k)}
= \mathbbm{1}_{m_{\underline{\lambda}}}
\quad \forall\, \underline{\mu}, \underline{\lambda}
\end{aligned}
\end{equation}
Given the channel $\mathcal{N}\n$ and the solution of the decoder SDP $\mathcal{D}_n^*$, to obtain the effective channel $\mathcal{M}_n'$ from the encoder's perspective we apply \ref{prop:general_serial_contenation}. Since the input side of the channel is not blocked, the encoder optimization proceeds exactly as in the previous section.

The standard symmetric seesaw method of Section~\ref{sec:symmetric_seesaw} is the special case $\ell = 1$ (no flag), while the superactivation analysis of \cite{Parentin2026} corresponds to $\ell = 2$ with a qubit-flagged channel ($d_S = 2$, a single qubit controlled by a binary flag). Finally, notice that the symmetric power iteration method (Definition \ref{def:symmetric_power_method_D}) extends straightforwardly to solve SDPs of the form in \eqref{eq:block_SDP_FD_flagged_generalized}, since after the block-diagonalization the structure is identical.

\subsection{Numerical results}
\label{sec:numerical_examples}

In this section we use the numerical toolbox introduced above to study lower bounds on the channel fidelity of multiple copies of two important noise models for quantum communication and quantum computation. 

\subsubsection{Qubit Amplitude Damping Channel}
\label{sec:AFC}
The \textit{amplitude damping channel} (ADC) is defined by:\begin{equation}
    \mathcal{A}_{\gamma}(\rho) \coloneqq A_1 \rho A_1^\dag + A_2 \rho A_2^\dag
\end{equation}
where:\begin{align}
    A_1 \coloneqq \begin{pmatrix}
        1 & 0\\
        0 & \sqrt{1-\gamma}
    \end{pmatrix} \qquad   A_2 \coloneqq \begin{pmatrix}
        0 & \sqrt{\gamma}\\
        0 & 0
    \end{pmatrix}.
\end{align}
where $\gamma$ denotes the transition probability from the state $\ket{1}$ to $\ket{0}$  (decay process). 
Its Choi matrix is:\begin{equation}
    \Gamma^{\mathcal{A}_\gamma} \coloneqq  \begin{pmatrix}
        1 & 0 & 0 & \sqrt{1-\gamma}\\
        0 & 0 & 0 & 0\\
        0 & 0 & \gamma & 0\\
        \sqrt{1-\gamma} & 0 & 0 & 1-\gamma\\
    \end{pmatrix},
\end{equation}
which is not Bell diagonal, and in fact the ADC is a non-Pauli channel. 

This channel is known to be degradable for $\gamma \in \left[0, \frac{1}{2}\right]$ and antidegradable for $\gamma \in \left[\frac{1}{2}, 1\right]$, so its quantum capacity can be computed explicitly as \cite{Giovannetti2005}:\begin{equation}
\label{eq:asymptotic_ADC}
 Q(\mathcal{A}_\gamma) =\begin{cases}
    \displaystyle\max_{\substack{p \in [0,1]}} [h((1-\gamma) p) - h(\gamma p)] & \gamma \in \left[0, \frac{1}{2}\right], \\[10pt]
    0 & \gamma \in \left[\frac{1}{2}, 1\right].
\end{cases}    
\end{equation} The ADC is an important noise model in superconducting circuits and a good model for spontaneous emission and energy dissipation, so approximate error correction strategies for finite $n$ are of interest in current platforms. Its (uncoded) entanglement fidelity can be explicitly computed from the Kraus representation as:\begin{equation}
\label{eq:ADC_uncoded}
    F_e(\mathcal{A}_{\gamma}) = \left( \frac{1+ \sqrt{1-\gamma}}{2}\right)^2.
\end{equation}
In \cite{Leung1997}, the authors introduced an explicit four-qubit quantum error-correcting code designed for amplitude damping errors, allowing to get the entanglement fidelity\begin{equation}
    \label{eq:ADC_Leung}
    F_e(\mathcal{D}^{(4)}_\gamma \circ \mathcal{A}_{\gamma} \circ \mathcal{E}^{(4)}) = \frac{1}{2} + \frac{\sqrt{1+ (\gamma -1)^4}}{2\sqrt{2}} + \gamma  -  \frac{\sqrt{1+ (\gamma -1)^4}}{2\sqrt{2}} \gamma - \frac{15}{4}\gamma^2 + \frac{7}{2}\gamma^3 - \gamma^4,
\end{equation}
where $\mathcal{E}^{(4)}$ is an isometric encoder and $\mathcal{D}^{(4)}_\gamma$ is a decoder that first performs a syndrome measurement and then an approximate error correction procedure \cite{Leung1997}. This formula shows that, similarly to stabilizer codes like the Shor's code \cite{Shor1995} or the five-qubit code \cite{Laflamme1996, Bennett1996}, this code is better than no coding only when the error parameter is small, namely for damping probabilities below $\gamma \approx 0.25$. In \cite[Section 3.3.3]{Reimpell2008} the authors showed that using the standard seesaw method to lower bound $F_c(\mathcal{A}^{\otimes 4}, 2)$, the fidelity of the resulting code $(\tilde{\mathcal{E}^{(4)}}, \tilde{\mathcal{D}^{(4)}})$ has a strict improvement over the four qubit code of \cite{Leung1997} for all $\gamma \in (0,1)$. The authors of \cite{Mao2025} also provided an explicit analytical expression for a code that beats Leung et al's code for $n = 4$ uses. For $n =5$ uses, the seesaw-optimized codes reach similar values of fidelity, and in general, all codes found by the seesaw iteration outperform the known error correction codes, for $0 < \gamma < 1$ \cite{Reimpell2008}.

Explicit PI codes for AD errors have been constructed by Ouyang and Chao~\cite{OuyangChao2020}, who designed codes satisfying the Knill--Laflamme conditions for a fixed number of amplitude damping errors, and by Aydin et al.~\cite{Aydin2024}, who further reduced the required code length.

In Figure~\ref{fig:ADC_results}, the codes arising from the symmetric seesaw  for $n \in [20]$ are compared with other codes from the literature. Specifically, results improve over the known lower bound on $F(\mathcal{A}_\gamma^{\otimes n}, 2)$ for all $\gamma \in [0, 0.4]$ and $\gamma \gtrsim 0.8$, achieved for $n = 5$ in \cite[Figure 3.13]{Reimpell2008}.

\begin{figure}[h]
    \centering
    \includegraphics[width=\linewidth]{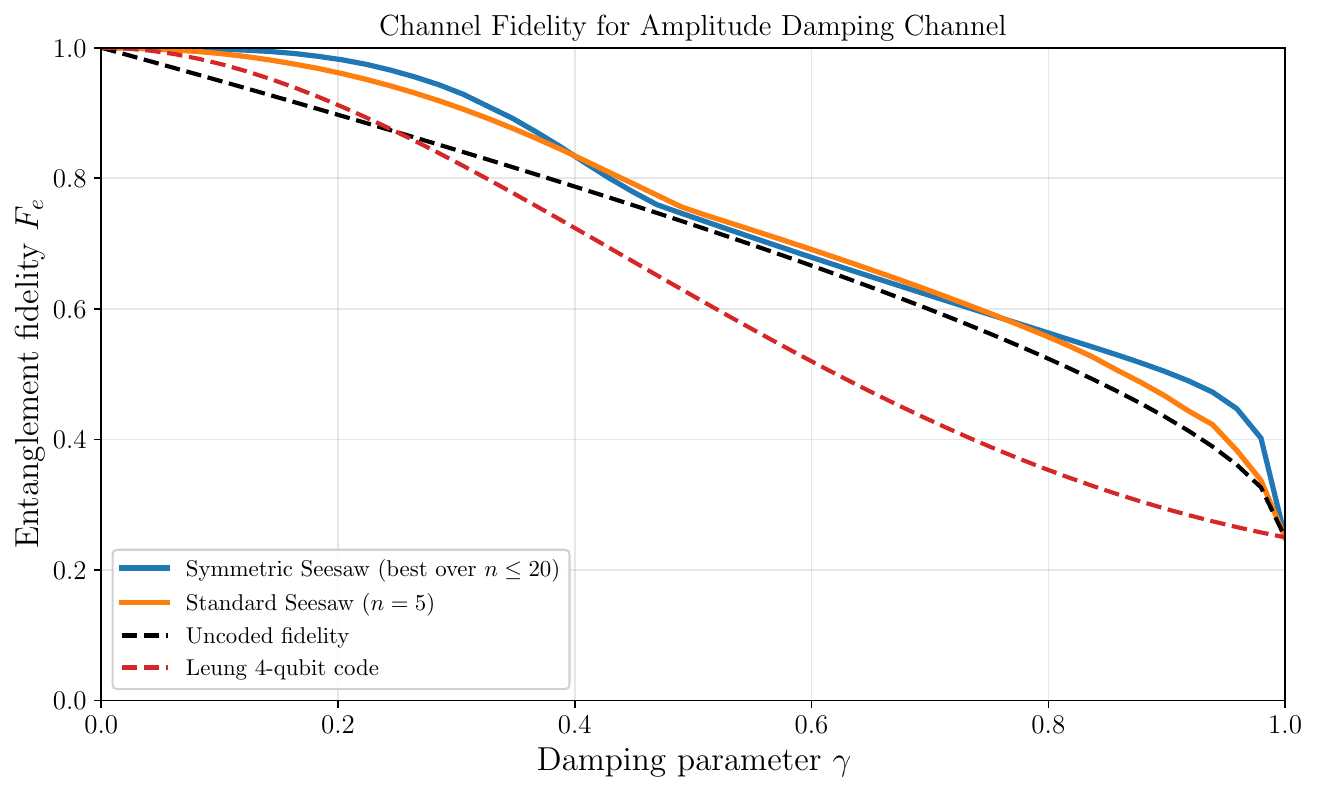}
    \caption{Comparison of error correction schemes for the amplitude damping channel. The uncoded curve displays the function in \eqref{eq:ADC_uncoded}, the four qubit code curve the function in \eqref{eq:ADC_Leung}. The orange curve of the complete optimization is the result of the seesaw iteration ($n = 5$), while the blue curve shows the result of the symmetric seesaw iteration (maximizing over $n \in \{1, ..., 20\}$).}
    \label{fig:ADC_results}
\end{figure}

Of particular interest is the behavior for low $\gamma$ since by the asymptotic result \eqref{eq:asymptotic_ADC} for $\gamma < 0.5$ we must get $F(\mathcal{A}_\gamma^{\otimes n}, 2) \xrightarrow{n \to \infty} 1$, i.e. reliable asymptotic qubit entanglement transmission. Hence, this result sheds new light on how to build codes for reliable communication over tens of channel copies of $\mathcal{A}_\gamma$, achieving an error probability as low as $\approx 1\%$ for $\gamma < 0.2$ and significantly improving over what was known before. These results offer a benchmark for future error correcting codes over the ADC in the finite blocklength setting.

\subsubsection{Qubit Depolarizing Channel}
\label{sec:depolarizing_channel}
The \textit{qubit depolarizing channel} is an important Pauli channel that with probability $p$ totally depolarizes the input system and with probability $1-p$ leaves the input system untouched:\begin{equation}
\label{eq:depolarizing_channel}
    \mathcal{D}_p(\rho) = (1-p)\rho + p \Tr[\rho] \pi
\end{equation}
This is a valid quantum channel for $p \in \left[0, \frac{4}{3}\right]$, and its Choi state is an \textit{isotropic state} \cite{Horodecki1999Reduction}, i.e.:\begin{align}
\label{eq:depolarizing_channel_choi}
    \Phi^{\mathcal{D}_p}_{AB} &= \rho_{AB}^{(1-p)} = (1-p) \Phi^2_{AB} +p\frac{\mathbbm{1}_{AB} - \Phi^2_{AB}}{3}\\
    &= \begin{pmatrix}
        \frac{1}{2} - \frac{p}{3} & 0 & 0 & \frac{1}{2} -\frac{2p}{3}\\
        0 & \frac{p}{3} & 0 & 0\\
        0 & 0 & \frac{p}{3} & 0\\
       \frac{1}{2}-\frac{2p}{3} & 0 & 0 &  \frac{1}{2} - \frac{p}{3}
    \end{pmatrix},
\end{align}
i.e., it is invariant under unitaries of the form $U \otimes \overline{U}$. This makes the depolarizing channel particularly suited for symmetry reductions; despite this, its quantum capacity remains unknown in closed form due to non-additivity of its coherent information. One finds~\cite {SmithSmolin2008, Sutter2017, Leditzky_2018}:\begin{equation}
    I_c(\mathcal{D}_p) = 1 + \left(1 - \frac{3p}{4}\right)\log  \left(1 - \frac{3p}{4}\right) + \frac{3p}{4} \log \frac{p}{4}
\end{equation}




and indeed, it is known that $ 5 I_c(\mathcal{D}^{p}) < I_c(\mathcal{D}^{{p}^{\otimes 5}})$ for some values of $p > 0.2524$, where a 5-qubit repetition code allows to get a strictly positive quantum data rate even if $I_c(\mathcal{D}^{(p)})$ is zero. For $p > \frac{1}{3}$, we know that $Q(\mathcal{D}^{p})=0$ because the channel is anti-degradable, but the exact threshold $q$ for which $Q(\mathcal{D}^{p})> 0$ is not known \cite{LeungWatrous2017}. A study of the $n$-shot entanglement transmission over $\mathcal{D}^{p}$ could shed new light to this, providing explicit codes over finitely many uses of the channel. 

The (uncoded) entanglement fidelity of the qubit depolarizing channel can be explicitly computed as:\begin{equation}
\label{eq:depolarizing_uncoded}
    F_e(\mathcal{D}_{p}) = 1 - \frac{3}{4}p.
\end{equation}
The standard benchmark for codes over the depolarizing channel is the the five bit code \cite{Laflamme1996, Bennett1996}, denoted by $(\mathcal{E}^{(5)}, \mathcal{D}^{(5)})$ and designed to correct all Kraus operators of $\mathcal{D}^{{p}^{\otimes 5}}$ that have at most one tensor factor different from the identity. Clearly, only few Kraus operators of $\mathcal{D}^{{p}^{\otimes 5}}$ are of this form, so perfect correction is not possible, but all terms of order $p$ can be corrected, and it can be shown that the fidelity using the five bit code is \cite[Equation 3.8]{Reimpell2008}:\begin{equation}
     \label{eq:depolarizing_five_qubit}
         F_e(\mathcal{D}^{(5)} \circ \mathcal{D}^{p, \otimes 5}\circ \mathcal{E}^{(5)}) = 1- \frac{45}{8} p^2 + \frac{75}{8}p^3- \frac{45}{8}p^4 + \frac{9}{8}p^5.
\end{equation}
For exact quantum error correction of Pauli errors, permutation-invariant codes are less efficient than stabilizer codes: the smallest PI code correcting a single arbitrary error uses $7$ qubits~\cite{Aydin2024} (or $9$ qubits for Ouyang's gnu codes~\cite{Ouyang2014}), compared to $5$ for the optimal stabilizer code~\cite{Laflamme1996, Bennett1996}.
Moreover, Bhalerao and Leditzky~\cite{Bhalerao2025} observed that the coherent information when optimized over permutation invariant states does not increase beyond what is achievable with standard (orthogonal) repetition codes for $\mathcal{D}^p$, unlike for other Pauli channels. Despite these limitations, the symmetric seesaw method allows to reach $n = 20$ channel uses, thereby accessing symmetric coding strategies that outperform known bounds for depolarizing parameter $p < 0.1$, as shown in Figure \ref{fig:depolarizing_results}.
\begin{figure}[h]
    \centering
    \includegraphics[width=0.8\linewidth]{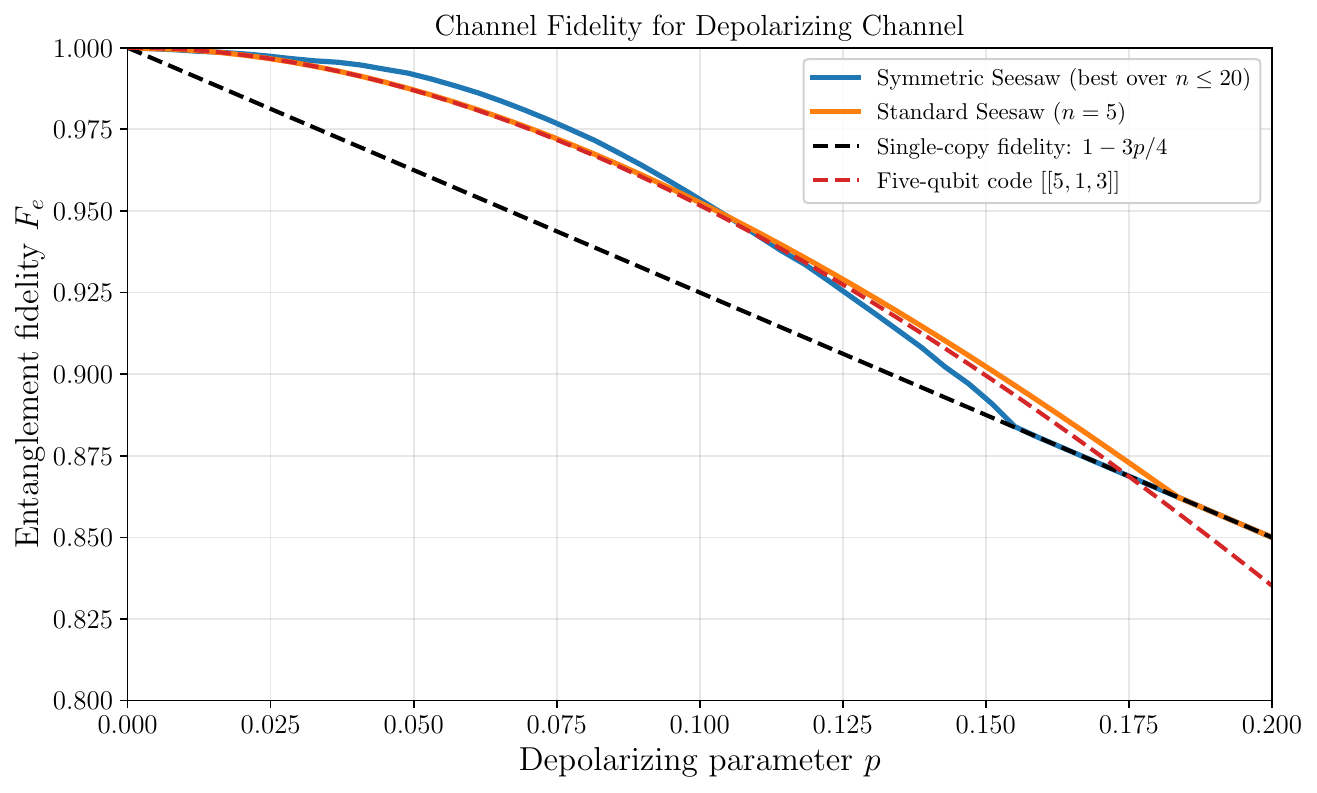}
    \caption{Comparison of error correction schemes for the depolarizing channel for $p \in [0, 0.20]$ . The orange curve is the result of the seesaw iteration ($n = 5$), while the blue curve shows the result of the symmetric seesaw iteration (maximizing over $n \in \{1, ..., 20\}$). These are compared with the no-coding curve of \eqref{eq:depolarizing_uncoded} and the five qubit code curve in \eqref{eq:depolarizing_five_qubit}.}
    \label{fig:depolarizing_results}
\end{figure}
We highlight that the symmetric seesaw simulation for the depolarizing channel, differently from the ADC, displayed a 'step-like' behavior: for all values of the depolarizing parameter $p \in [0,0.2]$, the estimated fidelity remains constant for several values of $n$ and then experiences sudden bumps, the first at $n = 7$ (as expected from \cite{Aydin2024}).

\section{Application: Quantum Relative Entropy Programs}
\label{sec:other_application_to_REE_programs}
\label{sec:PPT_REE} 
In this section we also apply the numerical toolbox introduced in Section \ref{sec:perm-inv} to study Quantum Relative Entropy programs.

For all $\rho \in \mathcal{D}(\mathcal{H})$ and $\sigma \in \mathcal{P}(\mathcal{H})$, the \textit{Umegaki relative entropy} is defined as \cite{Umegaki1954}\begin{equation}
    D(\rho \Vert \sigma) \coloneqq \begin{cases}
        \Tr[\rho(\log \rho - \log \sigma)]\qquad &\mathrm{if} \quad \mathrm{supp}(\rho) \subseteq \mathrm{supp}(\sigma),\\
        + \infty &\mathrm{otherwise},
    \end{cases}
\end{equation}(where $\log\equiv \log_2$) and crucially satisfies the \textit{data-processing inequality} \cite{Lindblad1975}\begin{equation}
    D(\mathcal{N}(\rho )||\mathcal{N}(\sigma)) \leq D(\rho||\sigma),
\end{equation}
for all quantum channels $\mathcal{N}$. A consequence of this is its \textit{joint convexity}: for all finite sets of states $\{\rho_i\}_{i=1}^n$, positive semidefinite operators $\{\sigma_i\}_{i=1}^n$ and probability distributions $\{\lambda_i\}_{i=1}^n$, we have
    \begin{equation}
    \label{eq:joint_convexity}
        D\left(\sum_{i=1}^n \lambda_i \rho_i \,\Big\|\, \sum_{i=1}^n \lambda_i \sigma_i\right) 
\leq \sum_{i=1}^n \lambda_i D(\rho_i\|\sigma_i).
    \end{equation} Finally, it satisfies a \textit{direct-sum property}: for all finite sets of states $\{\rho_i\}_{i=1}^n$, positive semidefinite operators $\{\sigma_i\}_{i=1}^n$ and probability distributions $\{\lambda_i\}_{i=1}^n$, we have
    \begin{equation}
    \label{eq:direct_sum_RE}
        D\left(\bigoplus_{i=1}^n \lambda_i \rho_i \,\Big\|\, \bigoplus_{i=1}^n \lambda_i \sigma_i\right) = \sum_{i=1}^n \lambda_i D(\rho_i\|\sigma_i).
    \end{equation} 
    
\subsection{Rains Relative Entropy}
The distillable entanglement of a bipartite state $\rho_{AB}$, denoted by $E_D(\rho_{AB})$ is the highest rate at which one can obtain maximally entangled states from $\rho_{AB}$ by local operations and classical communication (LOCC) \cite{Rains1999, Bennett1996MixedStateEntanglement}. Since it is extremely hard to calculate $E_D$ for general quantum states, in practice one relies on computable upper and lower bounds. The best known upper bound for the distillable entanglement of a bipartite state is its \textit{Rains relative entropy}, defined as \cite{Rains2001SemidefiniteDistillableEntanglement}\begin{equation}
\label{eq:Rains_Relative_entropy}
   R(\rho_{AB}) \coloneqq \inf_{\sigma_{AB} \in\text{PPT}' (A:B)} D(\rho_{AB} || \sigma_{AB}),
\end{equation}
where we defined the following set of positive semidefinite operators \cite{Audenaert2002}:\begin{equation}
\label{eq:rains_set}
    \mathrm{PPT}'(A;B) = \{\sigma_{AB} \in \mathcal{P}(AB): ||\sigma_{AB}^{T_B}||_1 \leq 1\}.
\end{equation}
This set is a relaxation of the set of PPT states (which satisfy $||\sigma_{AB}^{T_B}||_1 = 1$), and can be formulated as the optimization problem:\begin{equation}
\label{eq:rains_relative_entropy_program}
\begin{aligned}
      R(\rho_{AB}) &= \inf_{\sigma_{AB}, K_{AB}, L_{AB}}D(\rho_{AB}|| \sigma_{AB})\\
        & \ \ \  \text{subject to} \ \ \  (K_{AB}-L_{AB})^{T_{B}}  = \sigma_{AB}\\ 
          & \ \ \ \ \ \ \ \ \ \ \ \ \ \ \  \ \ \ \ \ \sigma_{AB}, K_{AB}, L_{AB} \geq 0\\ 
          & \ \ \ \ \ \ \ \ \ \ \  \ \ \ \ \ \ \ \ \ \Tr[K_{AB} + L_{AB}] \leq 1.
\end{aligned}
    \end{equation}
which is efficient because it consists of SDP constraints and a convex objective function (by \eqref{eq:joint_convexity}).    
The Rains relative entropy is in general \textit{subadditive}, i.e. $R(\rho_{AB} \otimes \omega_{AB}) \leq R(\rho_{AB}) + R(\omega_{AB})$, and the authors of \cite{Wang2017_Rains} exhibited a two-qubit state such that the inequality is strict, i.e.:
    \begin{equation}
       R(\rho_{AB}^{\otimes 2}) < 2R(\rho_{AB}).
    \end{equation}
Thus, if one defines the regularized Rains relative entropy \cite{Hayashi2006QuantumInformation} as\begin{equation}
\label{eq:regularized_rains}
    R^{\infty}(\rho_{AB}) \coloneqq \inf_{n \in \mathbb{N}} \frac{R(\rho_{AB}^{\otimes n})}{n},
\end{equation}
the latter is in general a better (i.e., tighter) upper bound for the distillable entanglement, because\begin{equation}
    E_D(\rho_{AB}) \leq  R^{\infty}(\rho_{AB}) \leq R(\rho_{AB}),
\end{equation}
and the second inequality can be strict \cite{Wang2017_Rains}. 

\subsection{Efficient Approximation of the Regularized Rains Relative Entropy}
\label{sec:efficient_Rains}

While quantum relative entropy programs of the form~\eqref{eq:rains_relative_entropy_program} can be efficiently solved by interior-point methods (see~\cite{HeSaundersonFawzi2024SCBarriers, HeSaundersonFawzi2024StructureQRE, Fang_2019}), e.g.\ via the toolkit \texttt{QICS} (Quantum Information Conic Solver)~\cite{HeSaundersonFawzi2024QICS}, the exponential growth of $d_A^n d_B^n$ with $n$ makes the direct computation of $R(\rho_{AB}^{\otimes n})$ intractable beyond a few copies. In this section, we exploit the permutation symmetry of the i.i.d.\ input $\rho_{AB}^{\otimes n}$ together with the representation-theoretic toolbox of Section~\ref{sec:perm-inv} to design a $\mathrm{poly}(n)$-time algorithm to compute $R(\rho_{AB}^{\otimes n})$ for all $n \in \mathbb{N}$. This in turn yields significantly improved approximations to the regularized Rains bound $R^\infty(\rho_{AB})$ defined in~\eqref{eq:regularized_rains}.

The key structural observation is that the optimization in~\eqref{eq:rains_relative_entropy_program} can be restricted to the permutation-invariant subspace at no loss of optimality. The next lemma formalizes this; the argument is standard and follows from twirling combined with the data-processing inequality (this is as a special case of \cite[Lemma 18]{Fang2025}).

\begin{lemma}
\label{lemma:rains_symmetric}
Let $n \in \mathbb{N}$ and let $\rho_{AB}$ be a bipartite state. The Rains relative entropy of $\rho_{AB}^{\otimes n}$ admits an optimal solution $(\sigma_{A^n B^n}^*, K_{A^n B^n}^*, L_{A^n B^n}^*)$ such that
\begin{equation}
\label{eq:rains_symmetric_constraint}
    \sigma_{A^nB^n}^*, K_{A^nB^n}^*, L_{A^nB^n}^* \in \mathrm{End}^{\mathfrak{S}_n}\!\bigl((AB)^{\otimes n}\bigr).
\end{equation}
Equivalently, the optimization in~\eqref{eq:rains_relative_entropy_program} for the input $\rho_{AB}^{\otimes n}$ can be restricted to permutation-invariant feasible triples without affecting the optimal value.
\end{lemma}
\begin{proof}
Let $(\sigma_{A^nB^n}, K_{A^nB^n}, L_{A^nB^n})$ be any feasible triple for the program~\eqref{eq:rains_relative_entropy_program} applied to $\rho_{AB}^{\otimes n}$, and let $\mathcal{S}_{A^nB^n}$ denote the symmetrizer channel under the diagonal action of $\mathfrak{S}_n$ on $(AB)^{\otimes n}$:
\begin{equation}
    \mathcal{S}_{A^nB^n}(X) \coloneqq \frac{1}{n!}\sum_{\pi \in \mathfrak{S}_n} P_{A^nB^n}(\pi)\, X\, P_{A^nB^n}(\pi)^\dagger,
\end{equation}
where $P_{A^nB^n}(\pi) = P_{A^n}(\pi) \otimes P_{B^n}(\pi)$. Define the twirled triple $(\widetilde{\sigma}, \widetilde{K}, \widetilde{L}) \coloneqq (\mathcal{S}_{A^nB^n}(\sigma), \mathcal{S}_{A^nB^n}(K), \mathcal{S}_{A^nB^n}(L))$. We verify that this triple is feasible. Positive semidefiniteness is preserved by $\mathcal{S}_{A^nB^n}$, since it is a convex combination of unitary conjugations. The trace constraint is also preserved by linearity:
\begin{equation}
    \Tr[\widetilde{K} + \widetilde{L}] = \Tr[\mathcal{S}_{A^nB^n}(K + L)] = \Tr[K + L] \leq 1.
\end{equation}
Since the permutation matrices $P_{A^nB^n}(\pi) = P_{A^n}(\pi) \otimes P_{B^n}(\pi)$ are real and orthogonal, the partial transpose $T_{B^n}$ commutes with the conjugation by $P_{A^nB^n}(\pi)$ for every $\pi \in \mathfrak{S}_n$, hence with the symmetrizer $\mathcal{S}_{A^nB^n}$. Therefore \begin{equation}
    \widetilde{\sigma}^{T_{B^n}} = \mathcal{S}_{A^nB^n}(\sigma^{T_{B^n}}) = \mathcal{S}_{A^nB^n}(K - L) = \widetilde{K} - \widetilde{L}.
\end{equation}
Finally, by the data-processing inequality applied to the channel $\mathcal{S}_{A^nB^n}$:
\begin{equation}
    D\bigl(\rho_{AB}^{\otimes n} \,\big\|\, \widetilde{\sigma}\bigr) = D\bigl(\mathcal{S}_{A^nB^n}(\rho_{AB}^{\otimes n}) \,\big\|\, \mathcal{S}_{A^nB^n}(\sigma)\bigr) \leq D\bigl(\rho_{AB}^{\otimes n} \,\big\|\, \sigma\bigr),
\end{equation}
where we used that $\rho_{AB}^{\otimes n}$ is $\mathfrak{S}_n$-invariant.


We have shown that for any feasible triple, the symmetrized triple is also feasible and achieves an objective no larger than the original. Therefore, the infimum in~\eqref{eq:rains_relative_entropy_program} can be restricted to permutation-invariant triples without loss of generality, proving~\eqref{eq:rains_symmetric_constraint}.
\end{proof}

The following Theorem turns Lemma~\ref{lemma:rains_symmetric} into an efficient algorithm. In the following, we will denote $X_n \in \{\sigma_n, K_n, L_n\}$ as \begin{equation}
\label{eq:orbit_basis_rains}
    X_{n} = \sum_{s = 1}^{m_{AB}} c_s^X\, C_s^{AB} = \bigoplus_{\lambda \in \mathrm{Par}(d_A d_B,\, n)}   \widetilde{[X_n]}_\lambda,
\end{equation}
where $\widetilde{[X_n]}_\lambda \equiv [\widetilde{\psi}_{A^n B^n}(X_n)]_\lambda \in \mathbb{C}^{m_\lambda^{AB} \times m_\lambda^{AB}}$ and $\widetilde{\psi}_{A^n B^n}$ defined in \eqref{eq:full_isomorphism}.

\begin{theorem}
\label{thm:rains_polytime}
Let $\rho_{AB}$ be a bipartite state with $d_A, d_B \in \mathbb{N}$, and let $n \in \mathbb{N}$. Then the Rains relative entropy $R(\rho_{AB}^{\otimes n})$ can be computed via a quantum relative entropy program with $\mathrm{poly}(n)$ variables and constraints of the form\begin{equation}
\label{eq:rains_block_SDP}
\begin{aligned}
    R(\rho^{\otimes n}) = \min_{\{\widetilde{[\sigma]}_\lambda, \widetilde{[K]}_\lambda, \widetilde{[L]}_\lambda\}_\lambda} \quad & \sum_{\lambda \in \mathrm{Par}(d_Ad_B, n)} f_\lambda \, D\!\left(\widetilde{[\rho^{\otimes n}]}_\lambda \,\big\|\, \widetilde{[\sigma_n]}_\lambda\right) \\
    \text{s.t.} \quad & \widetilde{[\sigma_n]}_\lambda \geq 0, \;\; \widetilde{[K_n]}_\lambda \geq 0, \;\; \widetilde{[L_n]}_\lambda \geq 0 &\forall\, \lambda\in \mathrm{Par}(d_Ad_B, n), \\
    & \left[\widetilde{T_{B}}\!\left(\oplus_{\lambda'}\widetilde{[\sigma_n]}_{\lambda'}\right)\right]_\lambda = \widetilde{[K]}_\lambda - \widetilde{[L]}_\lambda &\forall\, \lambda \in \mathrm{Par}(d_Ad_B, n), \\
    & \sum_{\lambda \in \mathrm{Par}(d_Ad_B, n)} f_\lambda\, \Tr\!\left[\widetilde{[K]}_\lambda + \widetilde{[L]}_\lambda\right] \leq 1,
\end{aligned}
\end{equation}
where $\widetilde{T_{B}}$ was defined in \eqref{eq:partial_transpose_block}. This program has:
\begin{equation}
\label{eq:rains_polynomial_constraints}
\begin{aligned}
\#\;\text{variables} &= O\bigl((n+1)^{(d_A d_B)^2}\bigr), \\
\#\;\text{constraints} &= O\bigl((n+1)^{d_A d_B}\bigr), \\
\text{size of constraints} &= O\bigl((n+1)^{d_A d_B (d_A d_B - 1)/2}\bigr).
\end{aligned}
\end{equation}
\end{theorem}

\begin{proof}
By Lemma~\ref{lemma:rains_symmetric}, the optimization can be restricted to permutation-invariant triples $(\sigma_{n}, K_{n}, L_{n})$. Setting $S = AB$ in the formalism of Section~\ref{sec:perm-inv}, the coefficients $\{c_s^{\rho^{\otimes n}}\}_{s \in [m_{AB}]}$ of the input state can be computed by applying~\eqref{eq:tensor_product_orbit} to $\rho_{AB}$. Using the $*$-isomorphism~\eqref{eq:full_isomorphism}, the input state's blocks $\widetilde{[\rho^{\otimes n}]}_\lambda$ are obtained from the orbit-basis coefficients via Theorem~\ref{thm:Efficient_construction_Isomorphism} in time $\mathrm{poly}(n)$.

The PSD constraint and the trace inequality constraints of \eqref{eq:rains_relative_entropy_program} decompose block-by-block. To enforce the partial-transpose equality constraint $\sigma_n^{T_{B^n}} = K_n - L_n$ in the block basis, we need go through the orbit basis (recall the discussion in Section \ref{sec:efficient-operations-in-block-basis}), using the linear map $\widetilde{T}_{\mathrm{B}}$ as in \eqref{eq:partial_transpose_block}. Using the blocks $\widetilde{[\sigma_n]}_\lambda$, $\widetilde{[K_n]}_\lambda$ and $\widetilde{[L_n]}_\lambda$ as optimizing variables, the objective function in \eqref{eq:rains_relative_entropy_program} splits, by the direct-sum property \eqref{eq:direct_sum_RE}, as \begin{equation}
    D(\rho^{\otimes n} \| \sigma_n) = \sum_\lambda f_\lambda\, D\!\left(\widetilde{[\rho^{\otimes n}]}_\lambda \,\big\|\, \widetilde{[\sigma_n]}_\lambda\right).
\end{equation}

The number of partitions is $|\mathrm{Par}(d_Ad_B, n)| \leq (n+1)^{d_Ad_B}$, and each block has size $m_\lambda^{AB} \leq (n+1)^{d_Ad_B(d_Ad_B-1)/2}$. The total number of variables across the three triples $(\widetilde{[\sigma]}_\lambda, \widetilde{[K]}_\lambda, \widetilde{[L]}_\lambda)$ is therefore $O\bigl((n+1)^{(d_Ad_B)^2}\bigr)$, and the constraint size is dominated by $m_\lambda^{AB}$, yielding the polynomial bounds in~\eqref{eq:rains_polynomial_constraints}.
\end{proof}

\subsection{Numerical Results}
Building upon the results of \cite{MiranowiczIshizaka2008ClosedREE}, Wang and Duan explicitly found a two-qubit state for which the Rains relative entropy is non-additive. The construction defines a closest separable state (CSS) $\sigma$ and a corresponding state $\rho$:
\begin{equation}
\label{eq:rho_wang_duan}
    \sigma(r) \coloneqq \begin{pmatrix}
        \frac{1}{4} & 0 & 0 & 0\\
        0 & r & \frac{1}{4\sqrt{2}} & 0\\
        0 & \frac{1}{4\sqrt{2}} & \frac{5}{8} - r & 0\\
        0 & 0 & 0 & \frac{1}{8}
    \end{pmatrix}, \qquad \rho (r) \coloneqq \begin{pmatrix}
        \frac{1}{8} & 0 & 0 & 0\\
        0 & x(r) & z(r) & 0\\
        0 & z(r) & \frac{7 - 8x(r)}{8} & 0\\
        0 & 0 & 0 & 0
    \end{pmatrix}
\end{equation}
where $r \in [0.3125, 0.5468]$, and:
\begin{align}
    x(r) &\coloneqq  r + \frac{32r^2 - 10r + 1}{256r^2 - 160r + 33} + \frac{(16r - 5) y^{-1}(r)}{32\ln(5/8 - y(r)) - 32\ln(5/8 + y(r))},\\
    y(r) &\coloneqq  \sqrt{4r^2 - \frac{5r}{8} + \frac{33}{64}},\\
    z(r) &\coloneqq  \frac{32r^2 - (6 + 32x(r))r + 10x(r) + 1}{4\sqrt{2}}.
\end{align}

The range $r \in [0.3125, 0.5468]$ ensures positivity of the states. Using the numerical algorithm from Theorem \ref{thm:rains_polytime}, we can increase the number of copies of $\rho$ and see even stronger violations of additivity, as depicted in Figure \ref{fig:rains_wang_duan} (compare with \cite[Fig. 1]{Wang2017_Rains}).

\begin{figure}[ht]
    \centering
    \includegraphics[width=0.8\linewidth]{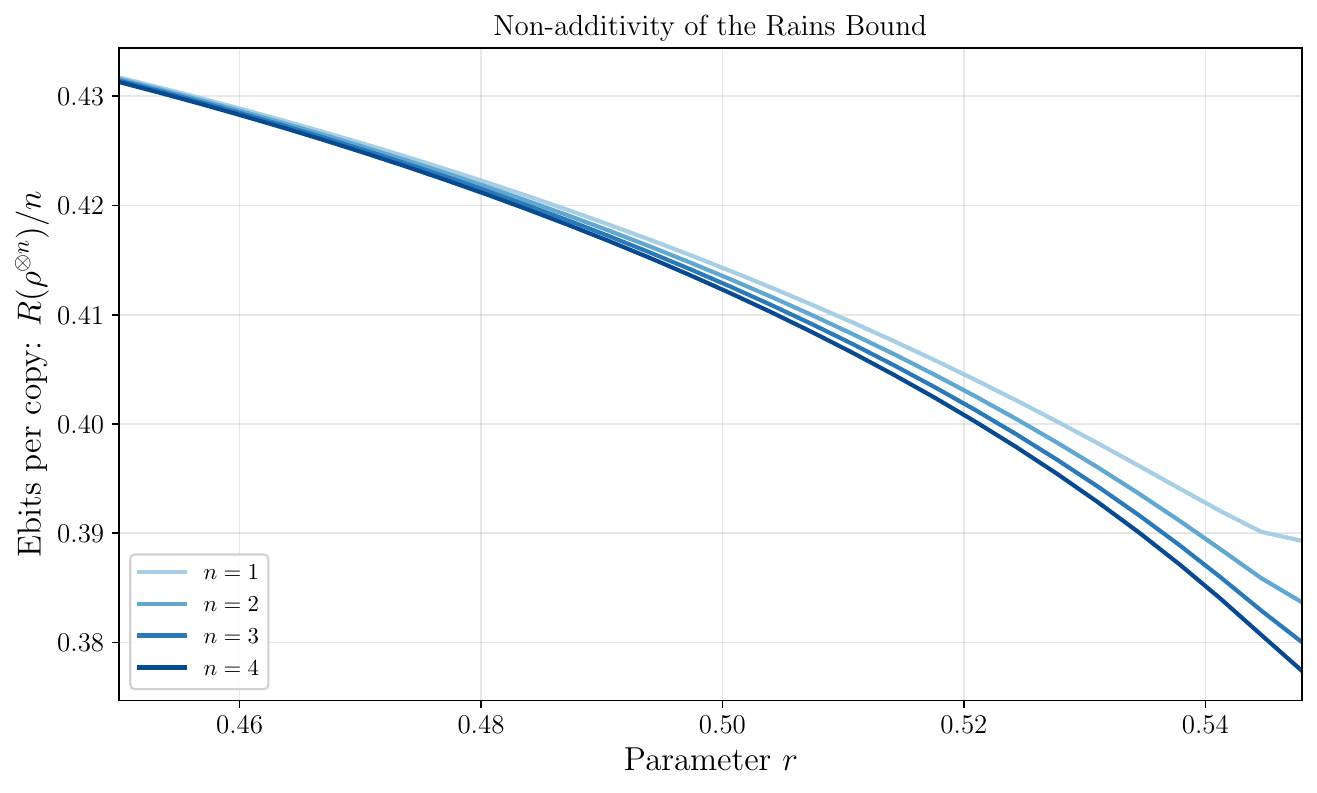}
    \caption{This plot shows the Rains Relative Entropy of $\rho^{\otimes n}(r)$ divided by $n$ for $r \in [0.45, 0.548]$, for $n = 1,2,3,4$.}
    \label{fig:rains_wang_duan}
\end{figure}

Using the results of \cite{Fang2024, Fang2025}, it is possible to combine our explicit algorithm in \eqref{eq:rains_block_SDP} with their lower bound based on the \textit{measured} relative entropy to efficiently approximate $R^{\infty}(\rho(r))$, since the set in \eqref{eq:rains_set} satisfies the structural assumptions of their analysis (see \cite[Section 4.3]{Fang2025}).

\section{Python Implementation} \label{sec:Python_package}
\newcommand{\heading}[1]{\medskip\par\textbf{{#1}:}}

We provide a Python package called \texttt{permqit} (\url{https://github.com/bbbergh/permqit}) which implements the algorithms and ideas described in this paper. Detailed usage instructions and documentation of individual classes can be found in the GitHub repository, what follows is a brief summary of what we is contained in this package:
\heading{Bases and Basis Operations}
We provide Python classes which represent the set of orbit basis matrices $C_s^{\mathcal{H}}$ as described in Section \ref{sec:orbit_basis_operations}, and allow for many operations such as: obtaining coefficients of tensor product matrices in this basis, taking the transpose, trace or HS norm of coefficients in this basis and much more. 

We also implement these operations for orbit bases in the general matrix algebra setting as described in Section \ref{sec:cq-algebras}, via the representation of Lemma \ref{lem:general_algebra_end_sn_trivial_isomorphism}. 

Additionally, the framework allows for coefficient vectors on only a subset of these bases and will try to limit operations (such as the channel link products described below) only to this subset wherever possible.
This is particularly useful for the coefficients of a tensor product channel $\mathcal{N}^{\otimes n}$ when the Choi matrix of $\mathcal{N}$ is sparse. 

\heading{Block Diagonalizations of $\esn(\HS\n)$}
We implement both of the methods explained in Appendix \ref{app:two_solutions_representation_theory} to obtain the block-diagonalizing maps $\psi$ and $\widetilde{\psi}$ described in Section \ref{sec:endsn_block_diagonalization}. In particular, this allows constructing block-diagonal representations (of size $\operatorname{poly}(n)$) of elements in $\esn(\HS\n)$  without ever constructing any (exponentially large) matrices on $\HS\n$.

\heading{Channel Link Product and Partial Traces}
We implement the partial-trace maps of Proposition \ref{prop:serial_concatenation} and Proposition \ref{prop:general_serial_contenation}, additionally also in a setting where one only needs to consider a subspace of $\esn(\A_A\n \otimes \A_B\n)$. 

\heading{Symmetric Seesaw Method}
We implement our symmetric seesaw method for quantum channel fidelity described in Section \ref{sec:seesaw} for a generic permutation-covariant channel $\mathcal{N}_n$ with Choi matrix $\Gamma^{\mathcal{N}_n} \in \esn((\A_A \otimes \A_B)\n)$, including optimizations if the output space is classical-quantum, or if the Choi matrix is sparse. We provide routines where the underlying optimization problems are either solved as an SDP, or solved through the power iteration method described in Section \ref{sec:powermethod}, the latter will also run on a GPU if available. 

\heading{Generic SDP optimization problems}
We provide helper methods to generate permutation covariant SDP variables for the \href{https://picos-api.gitlab.io/picos/}{\texttt{picos}} interface to solve permutation in-/covariant problems.

\heading{Numerical Backends}
Many sections of the code are written as to be compatible with multiple numerical backends, in particular if they implement the \href{https://data-apis.org/array-api/latest/}{Python Array API}. Concretely, we provide tested code paths using \href{https://numpy.org/}{\texttt{numpy}} and \href{https://cupy.dev/}{\texttt{cupy}} arrays, and, for sections that use sparse matrices, using \href{https://docs.scipy.org/doc/scipy/reference/sparse.html}{\texttt{scipy.sparse}}, \href{https://sparse.pydata.org/en/stable/}{\texttt{sparse} (pydata)} and \href{https://docs.cupy.dev/en/stable/reference/scipy_sparse.html}{\texttt{cupyx.scipy.sparse}} matrices. Many parts of the code use GPUs if available.

\section*{Acknowledgements}
B.B. is supported by the Engineering and Physical Sciences Research Council [Grant Ref: EP/Y028732/1].

\bibliography{bibliography}

\appendix

\section{Efficient Computation of the Polynomial}
\label{app:two_solutions_representation_theory}
In this appendix, we present two methods from the literature for efficiently computing the polynomial $f_{\tau,\gamma}$ in defined in~\eqref{eq:encoding_polynomial_semistandard-tab}. Both produce the same polynomial but via different approaches. We use throughout the notation established in Section~\ref{sec:perm-inv}: partitions $\lambda \vdash_{d_S} n$, semistandard Young tableaux $\tau, \gamma \in \mathcal{T}_{\lambda,d_S}$, row stabilizer $R_\lambda$, column stabilizer $C_\lambda$, row-equivalence $\tau' \sim \tau$, Young-symmetrized vectors $|u_\tau\rangle$, count matrices $E$, orbit matrices $C_E$, and monomials $x_E(X)$ as defined in~\eqref{eq:monomial_count_matrix}.

\subsection{Method 1: Count Functions}
\label{app:method1}

This method, due to~\cite{Litjens2016} (see also~\cite{Fawzi2022}), derives a direct formula for $f_{\tau,\gamma}$ valid for arbitrary $\tau, \gamma \in \mathcal{T}_{\lambda,d_S}$ using the structure of the Young basis:
\begin{proposition}
\label{prop:f_explicit}
The polynomial $f_{\tau,\gamma}$ can be written as:
\begin{equation}
\label{eq:f_explicit}
    f_{\tau, \gamma}(X) = \sum_{\tau' \sim \tau}\;
    \sum_{\gamma' \sim \gamma}\;
    \sum_{c,c' \in C_{\lambda}}
    \mathrm{sgn}(c)\,\mathrm{sgn}(c')
    \prod_{k=1}^{n} x_{\tau'(c(k)),\, \gamma'(c'(k))}
\end{equation}
where $\tau' \sim \tau$ and $\gamma' \sim \gamma$ denote summation over row-equivalent tableaux, $C_\lambda$ is the column stabilizer,and the product runs over all $n$ box positions of the Young diagram. Here, $x_{\tau'(c(k)),\, \gamma'(c'(k))} \coloneqq x_{E'}$ where $E'$ is the count matrix for the multi-indices corresponding to the pair of tableaux $(\tau'(c(k)),\, \gamma'(c'(k)))$ via \eqref{eq:multi-indices-tableaux}.
\end{proposition}
\begin{proof}
Expanding the Young-symmetrized vectors in $\langle u_\tau | C_E  u_\gamma \rangle$:
\begin{equation*}
    \langle u_\tau | C_E | u_\gamma \rangle =
    \sum_{\substack{\tau' \sim \tau,\; \gamma' \sim \gamma \\
    c, c' \in C_\lambda}}
    \mathrm{sgn}(c)\,\mathrm{sgn}(c')\;
    \left\langle \bigotimes_{k=1}^n \tau'(c(k))
    \;\middle|\; C_E \;\middle|\;
    \bigotimes_{k=1}^n \gamma'(c'(k)) \right\rangle
\end{equation*}
Since $\langle \underline{i} | C_E |  \underline{j} \rangle = 1$ if the pair $(\underline{i}, \underline{j})$ has count matrix $E$, and $0$ otherwise, each nonzero term contributes to the monomial $x_{E'}$ where $E'$ is the count matrix of the pair $(\tau'(c(\cdot)), \gamma'(c'(\cdot)))$. Summing over all $E$ with coefficient $x_E$ as in~\eqref{eq:encoding_polynomial}  gives~\eqref{eq:f_explicit}.
\end{proof}

The naive evaluation of~\eqref{eq:f_explicit} involves exponentially many terms (over all row-equivalent tableaux). The key to efficiency is grouping terms by their \textit{column structure}.

First, we eliminate one sum over $C_\lambda$ via the substitution $c' \mapsto c^{-1}c'$, using the fact that $C_\lambda$ is a group and $\mathrm{sgn}$ is a homomorphism:
\begin{equation}
\label{eq:f_determinantal}
    f_{\tau, \gamma}(X) = |C_\lambda| \sum_{\tau' \sim \tau}\;
    \sum_{\gamma' \sim \gamma}
    \prod_{j=1}^{\lambda_1}
    \det\!\left(
    (x_{\tau'(i,j),\, \gamma'(i',j)})_{i,i'=1}^{\lambda^*_j}
    \right)
\end{equation}
where $\lambda^*_j$ is the height of the $j$-th column of the Young diagram of $\lambda$, and $\tau'(i,j)$ denotes the entry in row $i$, column $j$ of the tableau $\tau'$.

The remaining sums over row-equivalent tableaux $\tau' \sim \tau$ and $\gamma' \sim \gamma$ can be grouped by recording, for each column, which entry vectors appear:

\begin{definition}
\label{def:count_function}
For tableaux $\tau' \sim \tau$ and $\gamma' \sim \gamma$, the \textit{count function} $\kappa: \bigcup_{t=1}^h ([d_S]^t \times [d_S]^t) \to \mathbb{N}_0$ assigns to each pair $(v,w) \in [d_S]^t \times [d_S]^t$ the number of columns $j$ of height $\lambda^*_j = t$ such that the column-entry vectors of $\tau'$ and $\gamma'$ in column $j$ equal $v$ and $w$, respectively:
\begin{equation}
    \kappa(v,w) = \bigl|\bigl\{j : \lambda^*_j = t,\;
    (\tau'(1,j), \ldots, \tau'(t,j)) = v,\;
    (\gamma'(1,j), \ldots, \gamma'(t,j)) = w\bigr\}\bigr|
\end{equation}
\end{definition}

The count function $\kappa$ must satisfy consistency conditions with the row contents of $\tau$ and $\gamma$. We denote by $\mathcal{K}(\tau, \gamma)$ the set of valid count functions for a given pair of SSYTs.

\begin{theorem}[Proposition 3, \cite{Litjens2016}]
\label{thm:kappa_formula}
The polynomial $f_{\tau,\gamma}$ can be written as:
\begin{equation}
\label{eq:f_kappa}
    f_{\tau, \gamma}(X) = |C_\lambda|\sum_{\kappa \in \mathcal{K}(\tau,\gamma)}\prod_{t=1}^{h} (\lambda_t - \lambda_{t+1})!\;\prod_{v,w \in [d_S]^t}\frac{\left[\det\!\left((x_{v_i, w_{i'}})_{i,i'=1}^{t}\right)\right]^{\kappa(v,w)}}{\kappa(v,w)!}
\end{equation}
where $h$ is the height of $\lambda$ (i.e., the number of nonzero parts) and $\lambda_{h+1} = 0$ by convention.
\end{theorem}

\begin{proof}
Group the sum over pairs $(\tau', \gamma')$ in~\eqref{eq:f_determinantal} by their count function $\kappa$. For columns of height $t$, there are $\lambda_t - \lambda_{t+1}$ columns to be assigned entry-vector pairs; the number of assignments yielding a given $\kappa$ is the multinomial coefficient:
\begin{equation*}
    \prod_{t=1}^h \frac{(\lambda_t - \lambda_{t+1})!}{
    \prod_{v,w \in [d_S]^t} \kappa(v,w)!}
\end{equation*}
Each assignment with the same $\kappa$ contributes the same product of determinants (since the determinant depends only on the column-entry vectors, not on which column they occupy). Combining gives~\eqref{eq:f_kappa}.
\end{proof}

As a consequence of this theorem, for fixed $d_S$, the polynomial $f_{\tau,\gamma}(X)$ can be computed in $\mathrm{poly}(n)$ time. Indeed, the determinants have size at most $d_S \times d_S$, the number of valid count functions $|\mathcal{K}(\tau,\gamma)|$ is bounded independently of $n$, and the multinomial products involve $\mathrm{poly}(n)$ arithmetic. Hence each matrix element $\langle u_\tau | C_E  u_\gamma \rangle = [x_E]\, f_{\tau,\gamma}$ can be extracted in $\mathrm{poly}(n)$ time.

\subsection{Method 2: Gijswijt Formula}
\label{app:method2}

This method, due to~\cite{Gijswijt2009}, gives a simple product formula for the special case $\tau = \gamma = \tau_\lambda$ (the constant tableau), and then extends to general tableaux via transition and differential operators.

We first define the following set of polynomials $\{Q_k\}_{k=0}^{d_S}$, with $Q_k \in O_k(S)$ defined by:
\begin{equation}
    Q_k(X) \coloneqq k!\, \det\!\begin{pmatrix}
        x_{1,1} & \cdots & x_{1,k}\\
        \vdots & \ddots & \vdots \\
        x_{k,1} & \cdots & x_{k,k}
    \end{pmatrix}
\end{equation}
with $Q_0 \coloneqq 1$.

We can express conveniently the polynomial $f_{\tau_\lambda, \tau_\lambda}$ using the above polynomial thanks to the following Theorem.

\begin{theorem}[Proposition 3,~\cite{Gijswijt2009}]
\label{thm:gijswijt}
For the constant tableau $\tau_\lambda$:
\begin{equation}
\label{eq:P_lambda}
    f_{\tau_\lambda, \tau_\lambda}(X) = P_\lambda(X)\coloneqq \prod_{k=1}^{d_S} Q_k(X)^{\lambda_k - \lambda_{k+1}}
\end{equation}
where $\lambda_{d_S+1} = 0$ by convention.
\end{theorem}
\begin{proof}
This follows by a similar argument as Proposition \eqref{prop:f_explicit}. Recalling the definition of $\ket{u_{\lambda}}$ in \eqref{eq:young_vector_trivial}, we can express the polynomial $f_{\tau_\lambda, \tau_\lambda}$ as:\begin{align}
    f_{\tau_\lambda, \tau_\lambda} &= \sum_E \braket{u_{\lambda}}{C_Eu_{\lambda}} \cdot x_E\\
    &= \sum_E \left(\sum_{c, c' \in C_\lambda }\mathrm{sgn}(c) \mathrm{sgn}(c')    \left\langle \bigotimes_{k=1}^n \tau_\lambda (c(k))
    \;\middle|\; C_E \;\middle|\;
    \bigotimes_{k=1}^n \tau_\lambda(c'(k)) \right\rangle \right) \cdot x_E\\
       &= \sum_E |C_\lambda| \left(\sum_{\pi \in C_\lambda }\mathrm{sgn}(\pi) \left\langle \bigotimes_{k=1}^n \tau_\lambda (\pi(k))
    \;\middle|\; C_E \;\middle|\;
    \bigotimes_{k=1}^n \tau_\lambda(k) \right\rangle \right) \cdot x_E
\end{align}
where we used the substitution $\pi \coloneqq c'^{-1} c$.  Since $\langle \underline{i} | C_E |  \underline{j} \rangle = 1$ if the pair $(\underline{i}, \underline{j})$ has count matrix $E$, and $0$ otherwise, each nonzero term contributes to the monomial $x_{E'}$ where $E'$ is the count matrix of the pair $(\tau_\lambda(\pi(\cdot)), \tau_\lambda(\cdot))$. Thus, we get:\begin{equation}
    f_{\tau_\lambda, \tau_\lambda} = |C_\lambda| \sum_{\pi \in C_\lambda} \mathrm{sgn}(\pi) \prod_{k=1}^n x_{\tau_\lambda(\pi(k)), \tau_\lambda(k)}
\end{equation}
which is a special case of \eqref{eq:f_explicit} where $\tau = \gamma = \tau_\lambda$, as expected. 

We now exploit the structure of the tableau $\tau_\lambda$, namely that $\tau_{\lambda} (i,j) = i$ for all rows $i \in [h]$ and columns $j \in [\lambda_1]$. Every permutation $\pi \in C_\lambda$ can be thought of as a tuple  $(\pi_1, \cdots, \pi_{\lambda_1})$, where $\pi_j \in \mathfrak{G}_{\lambda_j^*}$ permutes column $j$. Hence, $\pi(\tau_\lambda)(i,j) = \pi_j^{-1} (i)$, where $\pi_j^{-1} (i)$ denotes the action of the permutation $\pi_j^{-1}$ on symbol $i \in [\lambda_j^*]$, in line with \eqref{eq:symmetric_group_action}. Therefore:\begin{equation}
    \prod_{k=1}^n x_{\tau_\lambda(\pi(k)), \tau_\lambda(k)} = \prod_{j=1}^{\lambda_1} \prod_{i = 1}^{\lambda_j^*} x_{\pi_j(i),i}
\end{equation}
Using this, we obtain:\begin{align}
    f_{\tau_\lambda, \tau_\lambda} &=  |C_\lambda| \prod_{j = 1}^{\lambda_1} \sum_{\pi_j \in \mathfrak{G}_{\lambda_j^*}} \mathrm{sgn}(\pi_j) \prod_{i = 1}^{\lambda_j^*} x_{\pi_j(i),i}\\
    &=  |C_\lambda| \prod_{j = 1}^{\lambda_1}  \frac{1}{\lambda_j^*!} Q_{\lambda_{j}^*}\\
    &=  \prod_{k=1}^{d_S} Q_k(X)^{\lambda_k - \lambda_{k+1}}
\end{align}
\end{proof}

For any SSYT $\tau \in \mathcal{T}_{\lambda,d}$, the orbit matrices relate the Young-symmetrized vectors via:\begin{equation}
\label{eq:link_SSYT_constant_tab}
    |u_\tau\rangle = C_{E_\tau} |u_\lambda\rangle,
\end{equation}
where $E_\tau \coloneqq E^{(\tau, \tau_\lambda)}$ is the count matrix with\begin{equation}
(E_\tau)_{(a,b)} = |\{k : \tau(k) = a,\; \tau_\lambda(k) = b\}|.  
\end{equation}
This reduces general matrix elements to the constant-tableau case, but requires an efficient method for multiplying orbit matrices, i.e. computing terms of the form:\begin{equation}
    \bra{u_\lambda} C^T_{E_\gamma} C_E  C_{E_\tau} \ket{u_\lambda}
\end{equation}
for all $\gamma, \tau \in \mathcal{T}_{\lambda, d}$ and count matrix $E$ satisfying \eqref{eq:count_matrix}.

The approach of~\cite{Gijswijt2009} allows to efficiently perform these multiplications defining transition and differential operators.

For $a, b \in [d]_S$, define $T_{a \to b} \in \mathrm{End}^{\mathfrak{G}_n}(\mathcal{H}^{\otimes n})$ by:
\begin{equation}
    (T_{a \to b})_{(\underline{i}, \underline{j})} = \begin{cases}
        1 & \text{if } \exists!\, k : i_k = a,\; j_k = b,
            \text{ and } i_l = j_l\ \forall\, l \neq k \\
        0 & \text{otherwise}
    \end{cases}
\end{equation}
The action of \textit{transition operators} on orbit matrices can be computed efficiently as follows.

\begin{proposition}[Proposition 4, \cite{Gijswijt2009}]
\label{prop:transition_action}
For every orbit matrix $C_E$ and $a \neq b$:
\begin{align}
\label{eq:first_transition}
    C_E \cdot T_{a \to b} &= \sum_{c:\, E_{(c,a)} > 0} (E_{(c,b)} + 1) \cdot C_{E - E^{c,a} + E^{c,b}} \\
    \label{eq:second_transition}
    T_{a \to b} \cdot C_E &= \sum_{c:\, E_{(b,c)} > 0} (E_{(a,c)} + 1) \cdot C_{E - E^{b,c} + E^{a,c}}
\end{align}
where $E^{a,b}$ denotes the elementary count matrix with a single $1$ at position $(a,b)$, and the sums range only over values of $c$ that yield valid count matrices.
\end{proposition}
\begin{proof}
    Consider a multi-index pair $(\underline{i}, \underline{j})$. Since the product involves incidence matrices, the generic element $(C_E \cdot T_{a \to b})_{(\underline{i}, \underline{j})}$ equals the number of $\underline{l} \in [d_S]^{\times n}$ such that $E^{(\underline{i}, \underline{l})} = E$ and $\underline{j}$ is obtained from $\underline{l}$ by changing an $a$ into a $b$ in some position $k \in [n]$. If $i_k = c,$ then $E^{(\underline{i}, \underline{j})} = E - E^{c,a} + E^{c,b}$; in that case, the number of possible $\underline{l}$ equals $(E^{(\underline{i}, \underline{j)}})_{(c,b)} = E_{(c,b)} + 1$, concluding the proof of the first line. The see the second line, observe that:\begin{equation}
        T_{a \to b} \cdot C_E = (C_{E^T} \cdot T_{b\to a})^T
    \end{equation}
    from \eqref{eq:count-orbit-transpose-relation}.
\end{proof}

For $a, b \in [d_S]$ with $a \neq b$, define:
\begin{align}
    d_{a \to b} &\coloneqq \sum_{c=1}^{d_S} x_{c,a}\frac{\partial}{\partial x_{c,b}} \\
    d^*_{a \to b} &\coloneqq \sum_{c=1}^{d_S} x_{a,c} \frac{\partial}{\partial x_{b,c}}
\end{align}

These \textit{differential operators} are the polynomial counterparts of the transition operators. More precisely, we have the following result. 

\begin{proposition}[Proposition 6, \cite{Gijswijt2009}]
\label{prop:diff_transition}
For any $A \in \mathrm{End}^{\mathfrak{G}_n}(\mathcal{H}^{\otimes n})$ and $a \neq b$:
\begin{align}
    f_{A \cdot T_{a \to b}} &= d_{b \to a}\, f_A \\
    f_{T_{a \to b} \cdot A} &= d^*_{a \to b}\, f_A
\end{align}
\end{proposition}
\begin{proof}
Using $T_{a \to b}^\dagger = T_{b \to a}$ and cyclicity of the trace:
\begin{equation}
    f_{A \cdot T_{a \to b}} = \sum_E x_E\, \mathrm{Tr}[A^\dagger C_E T_{b \to a}]
\end{equation}
Applying~\eqref{eq:first_transition} to $C_E T_{b \to a}$ (i.e.,~\eqref{eq:first_transition} with $a$ and $b$ exchanged):
\begin{equation}
    C_E T_{b \to a} = \sum_{c:\, E_{(c,b)} > 0} (E_{(c,a)} + 1)\, C_{E - E^{c,b} + E^{c,a}}
\end{equation}
Substituting $E' = E - E^{c,b} + E^{c,a}$ (so $E = E' + E^{c,b} - E^{c,a}$ and $E_{(c,a)} + 1 = E'_{(c,a)}$ since $a \neq b$):
\begin{equation}
    f_{A \cdot T_{a \to b}} = \sum_c \sum_{E'}E'_{(c,a)}\, x_{E' + E^{c,b} - E^{c,a}}\,\mathrm{Tr}[A^\dagger C_{E'}]
\end{equation}
On the other hand:
\begin{equation}
    d_{b \to a}\, f_A = \sum_c x_{c,b}
    \frac{\partial}{\partial x_{c,a}} \sum_{E'}
    \mathrm{Tr}[A^\dagger C_{E'}]\, x_{E'}
    = \sum_c \sum_{E'} E'_{(c,a)}\,
    x_{E' - E^{c,a} + E^{c,b}}\,
    \mathrm{Tr}[A^\dagger C_{E'}]
\end{equation}
Since $E' + E^{c,b} - E^{c,a} = E' - E^{c,a} + E^{c,b}$, the two expressions coincide. A completely analogous proof gives the second line, using \eqref{eq:second_transition} instead.
\end{proof}

Combining the product formula for $P_\lambda$ with the differential--transition correspondence, we obtain the general formula.

\begin{theorem}[Theorem 7, \cite{Gijswijt2009}]
\label{thm:general_f}
Let $\tau, \gamma \in \mathcal{T}_{\lambda, d}$ with associated count matrices $E_\tau = E^{(\tau, \tau_\lambda)}$ and $E_\gamma = E^{(\gamma, \tau_\lambda)}$. Then:
\begin{equation}
\label{eq:f_general_gijswijt}
    f_{\tau, \gamma}(X) =\prod_{a > b}\left[(E_\tau)_{(a,b)}!\, (E_\gamma)_{(a,b)}!\right]^{-1} \cdot \prod_{b=1}^{d_S-1} \prod_{a=b+1}^{d_S} (d_{a \to b})^{(E_\gamma)_{(a,b)}}\,(d^*_{a \to b})^{(E_\tau)_{(a,b)}}\circ P_\lambda(X)
\end{equation}
\end{theorem}
\begin{proof}
The proof follows Theorem 7 of~\cite{Gijswijt2009}. We proceed in three
steps.

First, by Proposition~5 of~\cite{Gijswijt2009}, for any SSYT
$\tau \in \mathcal{T}_{\lambda, d_S}$ with lower-triangular count matrix
$E_\tau = E^{(\tau, \tau_\lambda)}$:
\begin{equation}
\label{eq:transition_decomposition_vector}
    \left(\prod_{a > b} (E_\tau)_{(a,b)}!\right) |u_\tau\rangle = \left(\prod_{b=1}^{d_S-1} \prod_{a=b+1}^{d_S} T_{a \to b}^{(E_\tau)_{a,b}}\right) |u_\lambda\rangle
\end{equation}
This follows from the operator identity~$\left(\prod_{a>b} (E_\tau)_{a,b}!\right) C_{E_\tau} = \left(\prod T_{a \to b}^{(E_\tau)_{a,b}}\right) C_{\mathrm{diag}(\lambda)}$, combined with the fact that the diagonal orbit matrix $C_{\mathrm{diag}(\lambda)}$ acts as the identity on $|u_\lambda\rangle$, since all multi-indices in the expansion of $|u_\lambda\rangle$ have composition $\lambda = w(\tau_\lambda)$.

We then express $f_{\tau,\gamma}$ via left and right multiplications. From~\eqref{eq:link_SSYT_constant_tab} and
\eqref{eq:transition_decomposition_vector}:
\begin{align}
    |u_\tau\rangle &= \frac{1}{\prod_{a>b}(E_\tau)_{a,b}!}
    \left(\prod_{b,a>b} T_{a \to b}^{(E_\tau)_{a,b}}\right)
    |u_\lambda\rangle \\[4pt]
    \langle u_\gamma| &= \frac{1}{\prod_{a>b}(E_\gamma)_{a,b}!}
    \langle u_\lambda|
    \left(\prod_{b,a>b} T_{a \to b}^{(E_\gamma)_{a,b}}\right)^\dagger
    = \frac{1}{\prod_{a>b}(E_\gamma)_{a,b}!}
    \langle u_\lambda|
    \left(\prod_{b,a>b} T_{b \to a}^{(E_\gamma)_{a,b}}\right)
\end{align}
where the second line uses $T_{a \to b}^\dagger = T_{b \to a}$. Therefore:
\begin{equation}
    \prod_{a>b}\!\left[(E_\tau)_{a,b}!(E_\gamma)_{a,b}!\right]
    \cdot |u_\tau\rangle\langle u_\gamma|
    = \left(\prod_{b,a>b} T_{a \to b}^{(E_\tau)_{a,b}}\right)
    |u_\lambda\rangle\langle u_\lambda|
    \left(\prod_{b,a>b} T_{b \to a}^{(E_\gamma)_{a,b}}\right)
\end{equation}

Finally, applying the polynomial $f$ to both sides, and using Proposition~\ref{prop:diff_transition} iteratively:
\begin{itemize}
    \item Each left multiplication by $T_{a \to b}$ translates to
    $d^*_{a \to b}$ (second line of Proposition~\ref{prop:diff_transition}).
    \item Each right multiplication by $T_{b \to a}$ translates to $d_{a \to b}$ (first line of Proposition~\ref{prop:diff_transition}, with $a$ and $b$ exchanged, $f_{A \cdot T_{b \to a}} = d_{a \to b} f_A$).
\end{itemize}
Moreover, the operators $d_{a \to b}$ and $d^*_{a' \to b'}$ commute for all index pairs. Indeed:
\begin{equation}
    [d^*_{a_1 \to b_1},\, d_{a_2 \to b_2}]
    = \left[\sum_c x_{a_1,c}\frac{\partial}{\partial x_{b_1,c}},\;
    \sum_{c'} x_{c',a_2}\frac{\partial}{\partial x_{c',b_2}}\right]
    = x_{a_1,a_2}\frac{\partial^2}{\partial x_{b_1,b_2}}
    - x_{a_1,a_2}\frac{\partial^2}{\partial x_{b_1,b_2}} = 0
\end{equation}
where the two terms arise from $[\partial_{b_1,c}, x_{c',a_2}] = \delta_{b_1,c'}\delta_{c,a_2}$ and $[x_{a_1,c}, \partial_{c',b_2}] = -\delta_{a_1,c'}\delta_{c,b_2}$ respectively. This commutativity
allows us to freely interleave the $d$ and $d^*$ operators. Combining
everything and using $f_{|u_\lambda\rangle\langle u_\lambda|} = P_\lambda$
(Theorem~\ref{thm:gijswijt}):
\begin{equation}
    \prod_{a>b}\left[(E_\tau)_{a,b}!(E_\gamma)_{a,b}!\right]
    \cdot f_{\tau,\gamma}
    = \prod_{b=1}^{d_S-1} \prod_{a=b+1}^{d_S}
    (d_{a \to b})^{(E_\gamma)_{a,b}}\,
    (d^*_{a \to b})^{(E_\tau)_{a,b}} \circ P_\lambda
\end{equation}
Dividing both sides by $\prod_{a>b}[(E_\tau)_{a,b}!(E_\gamma)_{a,b}!]$
yields~\eqref{eq:f_general_gijswijt}.
\end{proof}

Since $[d_{a \to b}, d^*_{a' \to b'}] = 0$ for all index pairs, the ordering of the operators in~\eqref{eq:f_general_gijswijt} is immaterial: the $d$ and $d^*$ factors can be applied in any order. In particular, one may equivalently group all $d^*$ operators together (bra side) and all $d$ operators together (ket side):
\begin{equation}
    f_{\tau,\gamma} = \prod_{a>b}[(E_\tau)_{a,b}!(E_\gamma)_{a,b}!]^{-1}
    \cdot \left(\prod_{a>b} (d^*_{a \to b})^{(E_\tau)_{a,b}}\right)
    \left(\prod_{a>b} (d_{a \to b})^{(E_\gamma)_{a,b}}\right)
    P_\lambda
\end{equation}
However, operators of the same type (e.g., $d^*_{a_1 \to b_1}$ and $d^*_{a_2 \to b_2}$) do not commute in general, and their relative ordering within each group must be consistent with the derivation (descending in the pair indices $(a,b)$, i.e., operators with larger indices applied first).

\subsection{Efficient Computation of the Gram Matrix}

\label{sec:gram_matrix}

The Gram matrix $G_\lambda = (\langle u_\tau | u_\gamma \rangle)_{\tau, \gamma \in \mathcal{T}_{\lambda, d_S}}$ admits a block-diagonal structure arising from weight orthogonality, and can be computed efficiently using the polynomial.

The \textit{weight} of a tableau $\tau \in \mathcal{T}_{\lambda, d_S}$ is $w(\tau) \coloneqq (\mu_1, \ldots, \mu_{d_S})$, where $\mu_a$ counts the occurrences of symbol $a \in [d_S]$ in $\tau$.

\begin{lemma}
\label{lemma:weight_orthogonality}
Let $\lambda \in \mathrm{Par}(d_S, n)$ and $\tau, \gamma \in \mathcal{T}_{\lambda, d_S}$. Then:
\begin{equation}
\label{eq:weight_condition}
    w(\tau) \neq w(\gamma) \implies \langle u_\tau | u_\gamma \rangle = 0
\end{equation}
\end{lemma}

\begin{proof}
The Young-symmetrized vector $|u_\tau\rangle$ lies entirely in the weight-$w(\tau)$ subspace of $S^{\otimes n}$. Since distinct weight subspaces are mutually orthogonal (they are eigenspaces of the diagonal torus action with distinct eigenvalues), the claim follows.
\end{proof}

For fixed $\lambda$, the SSYTs partition by weight: $\mathcal{T}_{\lambda, d_S} = \bigcup_\mu \mathcal{T}_{\lambda, d_S}(\mu)$, where $|\mathcal{T}_{\lambda, d_S}(\mu)| = K_{\lambda, \mu}$ is the \textit{Kostka number}~\cite{Kostka1882}. Consequently, $G_\lambda$
decomposes as:
\begin{equation}
    G_\lambda = \bigoplus_{\mu} G_\lambda(\mu)
\end{equation}
where $G_\lambda(\mu)$ is a $K_{\lambda,\mu} \times K_{\lambda,\mu}$ block.

For $d_S = 2$ (qubit systems), each weight class contains exactly one SSYT, since the SSYT constraint forces all $2$s into the second row, leaving the tableau uniquely determined by the number of $1$s in the first row. Therefore $G_\lambda$ is always diagonal.

\begin{corollary}
\label{cor:gram_matrix}
For all $\lambda \in \mathrm{Par}(d_S, n)$, the Gram matrix $G_\lambda$ is block-diagonal with blocks indexed by weight and can be computed in $\mathrm{poly}(n)$ time for fixed $d_S$. Its entries admit the combinatorial formula:
\begin{equation}
\label{eq:gram_combinatorial}
    (G_\lambda)_{\tau,\gamma} = \langle u_\tau | u_\gamma \rangle = |C_\lambda| \sum_{r_1, r_2 \in R_\lambda}\; \sum_{\substack{c \in C_\lambda \\[2pt] c(r_2(\gamma)) = r_1(\tau)}} \mathrm{sgn}(c)
\end{equation}
\end{corollary}

\begin{proof}
The block-diagonal structure follows from Lemma~\ref{lemma:weight_orthogonality}. For the combinatorial formula, start from the determinantal expression~\eqref{eq:f_determinantal} evaluated at $X = \mathbbm{1}_{d_S}$:
\begin{equation}
    \langle u_\tau | u_\gamma \rangle = f_{\tau,\gamma}(\mathbbm{1}_{d_S}) = |C_\lambda| \sum_{\tau' \sim \tau}\; \sum_{\gamma' \sim \gamma} \prod_{j=1}^{\lambda_1} \det\!\left((\delta_{\tau'(i,j),\, \gamma'(i',j)})_{i,i'=1}^{\lambda^*_j} \right)
\end{equation}
The determinant $\det(\delta_{v_i, w_{i'}})_{i,i'=1}^t$ of a $t \times t$ matrix of Kronecker deltas vanishes unless $v$ and $w$ are permutations of the same $t$ distinct symbols, in which case it equals the sign of the unique permutation mapping $v$ to $w$. Therefore, the product over columns is nonzero only when, for every column $j$, the column-entry vectors of $\tau'$ and $\gamma'$ are related by a column permutation. This means precisely that there exists $c \in C_\lambda$ with $c(\gamma') = \tau'$, and the product evaluates to $\mathrm{sgn}(c)$. Writing $\tau' = r_1(\tau)$ and $\gamma' = r_2(\gamma)$ for $r_1, r_2 \in R_\lambda$ gives~\eqref{eq:gram_combinatorial}.

The weight-orthogonality is immediate from this formula: if $w(\tau) \neq w(\gamma)$, no element of $R_\lambda C_\lambda R_\lambda$ can map $\gamma$ to $\tau$ (since row permutations, column permutations, and their compositions all preserve weight), so the sum vanishes.

The computation is efficient: evaluating $\langle u_\tau | u_\gamma \rangle = f_{\tau,\gamma}(\mathbbm{1}_{d_S})$ takes $\mathrm{poly}(n)$ time by either Method~1 or Method~2 above.
\end{proof}
\section{Maximal Fidelity of Recovery}
\label{sec:entropy_channel}
\subsection{Properties of the Maximal Fidelity of Recovery}
The next lemma states and proves the main properties of the maximal fidelity of recovery defined in \eqref{eq:recovery_as_singlet_fraction}, based on the correspondence with the maximal singlet fraction.
\begin{lemma}
\label{lem:max_recovery_properties}
Let $\mathcal{M}_{R'\to B}$ be a quantum channel with $d_{R'} = d$. The maximal recovery coefficient $F_{\mathrm{D}}(\mathcal{M})$ satisfies the following properties:
\begin{enumerate}
    \item \textbf{Bounds}. \begin{equation}
    \label{eq:bounds_F_M}
        \frac{1}{d^2} \leq F_{\mathrm{D}}(\mathcal{M}) \leq 1
    \end{equation}Moreover, $F_{\mathrm{D}}(\mathcal{M}) = 1$ if and only if $\mathcal{M}$ is an isometric channel with $d_{B} \geq d$, while $F_{\mathrm{D}}(\mathcal{M}) = \frac{1}{d^2}$ if and only if $\mathcal{M} = \mathcal{R}^{\sigma_B}$ is a replacement channel.
        \item \textbf{Convexity}.
    \begin{equation}
        F_{\mathrm{D}}\left(\sum_{x \in \mathcal{X}} p_X(x) \mathcal{M}^x\right) \leq \sum_{x \in \mathcal{X}} p_X(x) F_{\mathrm{D}}(\mathcal{M}^x)
    \end{equation}
    \item \textbf{Entanglement breaking channels and PPT channels}. For every entanglement breaking or (more generally) positive-partial-transpose (PPT) channel \cite{Rains1999} $\mathcal{M}$:
    \begin{equation}
        F_{\mathrm{D}}(\mathcal{M}) \leq \frac{1}{d}
\end{equation}
\item $k-$\textbf{Extendibile Channels}.  For every $k-$extendible channel $\mathcal{M}$:\begin{equation}
\label{eq:Fidelity_recovery_k_extendible}
    F_D(\mathcal{M}) \leq \frac{k+d-1}{dk}
\end{equation}
    \item \textbf{Classical-quantum channels}. For a classical-quantum channel $\mathcal{M}^{cq}_{X\to B}$:
    \begin{equation}
        F_{\mathrm{D}}(\mathcal{M}^{cq}) = \frac{1}{d} P_{\mathrm{opt}}(\mathcal{M}^{cq})
    \end{equation}
    where $P_{\mathrm{opt}}(\mathcal{M}^{cq})$ is the optimal guessing probability \cite{Tomamichel2015} with quantum side information $\rho_B^x = \mathcal{M}^{cq}(\ket{x}\bra{x})$.    
    \item \textbf{Multiplicativity}.
    \begin{equation}
    \label{eq:multiplicativity_maximal_fidelity_recovery}
        F_{\mathrm{D}}(\mathcal{M}_1 \otimes \mathcal{M}_2) = F_{\mathrm{D}}(\mathcal{M}_1) \cdot F_{\mathrm{D}}(\mathcal{M}_2).
    \end{equation}
\item \textbf{Flagged channels}. Let $\mathcal{N}_{A \to ZB}$ be a flagged channel of the form
\begin{equation}
\label{eq:flagged_QCQ_channel}
    \mathcal{N}_{A \to ZB}(\rho_A) = \sum_{i=1}^\ell p_i \, \ket{i}\!\bra{i}_Z \otimes \mathcal{N}^i_{A \to B}(\rho_A),
\end{equation}
where $\{\mathcal{N}^i\}_{i \in [\ell]} \subset \mathrm{CPTP}(A \to B)$ and $\{p_i\}_{i \in [\ell]}$ is a probability distribution. Then
\begin{equation}
\label{eq:fidelity_recovery_flagged}
    F_{\mathrm{D}}(\mathcal{N}) = \sum_{i=1}^\ell p_i\, F_{\mathrm{D}}(\mathcal{N}^i),
\end{equation}
and the optimum is attained by a recovery map of the form $\mathcal{D}_{ZB \to R'} = \mathcal{D}^Z_{B \to R'} \circ \Pi_Z$, where $\Pi_Z$ is the measurement on $Z$ in the flag basis $\{\ket{i}\}_{i \in [\ell]}$ and $\{\mathcal{D}^i_{B \to R'}\}_{i \in [\ell]}$ is a family of conditional decoding channels.
\item \textbf{Group-Invariance}. Given a finite group $G$, if $\mathcal{M}$ is $G$-invariant with respect to the input, i.e. $\mathcal{M} (\rho) = U(g) \mathcal{M}(\rho) U(g)^\dag$ for all input states $\rho$ and $g \in G$, then the maximal fidelity of recovery \eqref{eq:SDP_Maximal_Recovery_Coefficient} is achieved by a $G-$invariant channel on its input, i.e.:\begin{equation}
\label{eq:G-invariant-decoder}
    \mathcal{D}(U(g)\, \sigma\, V(g)^\dagger)
    = \mathcal{D}(\sigma)
    \qquad \forall\, g \in G,\; \forall\, \sigma \in \mathcal{D}(\mathcal{H})
\end{equation}
In the special case $G = \mathfrak{S}_n$, if $\mathcal{M}$ is symmetric in the sense of \eqref{eq:condition_symmetry_encoder}, then $\mathcal{D}$ is symmetric in the sense of \eqref{eq:condition_symmetry_decoder}.
\end{enumerate}
\end{lemma}

\begin{proof}
These properties are all immediate consequences of the analogous ones of the maximal singlet fraction (see e.g. \cite[Theorem 4.1.2]{Parentin2025Thesis}). 
\begin{enumerate}
    \item The primal in~\eqref{eq:SDP_Maximal_Recovery_Coefficient} always admits $\Phi^{\mathcal{D}^*}_{RB} = \pi_R \otimes \pi_{B}$ as a feasible point (corresponding to the completely depolarizing channel), yielding $F_D(\mathcal{M}) \geq 1/d^2$. Similarly, the dual in~\eqref{eq:SDP_Maximal_Recovery_Coefficient} admits the feasible $Y_{B} = \mathbbm{1}_{B}$, yielding $F_D(\mathcal{M}) \leq 1$. This proves \eqref{eq:bounds_F_M}.  By the compactness of $\mathcal{C}^*(R;B)$ and the faithfulness property of fidelity (see, e.g., \cite[theorem 6.7]{Khatri2020Principles}), the upper bound is saturated iff $\Phi^\mathcal{M}_{RB}$ is pure and equal to the Choi state of a CPU map, i.e.~iff $\mathcal{M}$ is an isometric channel (with $d_{B} \geq d$). Saturation of the lower bound follows from \eqref{eq:maximal_singlet_fraction_as_generalized_roboustness}, as\begin{equation}
     \min_{r}\frac{1+r}{d^2} =  \frac{1}{d^2} \iff r = 0 \iff \Phi_{RB}^{\mathcal{M}} = \pi_{R} \otimes \sigma_{B} \iff \mathcal{M} = \mathcal{R}^{\sigma_B}, 
\end{equation}
where $ \mathcal{R}^{\sigma_B}$ denotes a replacement channel for some $\sigma_B \in \mathcal{D}(B)$.
   \item Let $\mathcal{M} = \sum_{x} p(x) \mathcal{M}^x$. By linearity, its Choi state is\begin{equation}
       \Phi^{\mathcal{M}}_{RB} = \sum_x p(x)\,\Phi^{\mathcal{M}^x}_{RB}.
   \end{equation}Convexity of $F_D$ then follows immediately because its defining SDP is a maximization, and the maximum of a convex combination is no greater than the convex combination of the maxima.
   \item If $\mathcal{M}$ is PPT (in particular entanglement-breaking), $\Phi^\mathcal{M}_{RB}$ has positive partial transpose. Then~\cite[Lemma~2]{Rains1999}:
\begin{equation}
\label{eq:PPT_singlet_fidelity}
     \Tr[\Phi_{RB}^d \Phi^\mathcal{M}_{RB}] = \frac{1}{d}\Tr[\mathbb{F} (\Phi^\mathcal{M}_{RB})^{T_B}] \leq \frac{1}{d} \left\|(\Phi^\mathcal{M}_{RB})^{T_B}\right\|_1 = \frac{1}{d},
\end{equation}
where we used the self-adjointness of $T_B$, the standard property $d \ket{\Phi^d}\!\bra{\Phi^d}^{T_B} = \mathbb{F}$, where $\mathbb{F}$ is the (unitary) flip operator, the variational characterization of the trace norm, and in the last step, the PPT property of $\Phi^\mathcal{M}_{RB}$.
\item A $k$-extendible $\mathcal{M}$ has a Choi state $\Phi^\mathcal{M}_{RB} \in \mathrm{Ext}_k(R:B)$. Then, using the bound in~\cite{Johnson2013} gives:
\begin{equation}
\label{eq:singlet_fraction_k_extendible}
   F_D(\mathcal{M})  = F_{\Phi^d}(\Phi^\mathcal{M}_{RB} ) \leq \frac{k + d - 1}{dk}.
\end{equation}
\item Let $\mathcal{M}^{cq}_{X \to B}$ be a classical-quantum (CQ) channel, whose Choi state is the CQ state  \begin{equation}
     \Phi_{RB}^{\mathcal{M}^{cq}} = \frac{1}{|\mathcal{X}|} \sum_{x\in \mathcal{X}} \ket{x}\bra{x}_R \otimes \rho_{B}^x,
\end{equation}
where $\rho_{B}^x \equiv \mathcal{M}^{cq}_{R' \to B}(\ket{x}\!\bra{x}_{R'})$. Expanding \eqref{eq:MES} and using linearity, the objective function in \eqref{eq:recovery_as_singlet_fraction} becomes\begin{align}
            \Tr[\Phi^d_{RR'} \cdot (\mathrm{id}_R \otimes (\mathcal{D}_{B \to R'}\circ \mathcal{M}^{cq}_{R' \to B}))(\Phi^d_{RR'})] &= \frac{1}{d}\sum_{i=0}^{d-1}\bra{ii}_{RB}\left(\sum_{x\in \mathcal{X}}p(x) \ket{x}\bra{x}_R \otimes  \mathcal{D}_{B\to R'}(\rho_B^x)\right) \sum_{j=0}^{d-1}\ket{jj}_{RB}\\
            &= \frac{1}{d} \sum_{x\in \mathcal{X}}  p(x) \bra{x}\mathcal{D}_{B\to R'}(\rho_B^x)\ket{x}\\
            &= \frac{1}{d} \sum_{x\in \mathcal{X}}  p(x) \Tr[\rho_B^x\mathcal{D}^*_{R' \to B}(\ket{x}\bra{x}_{R'})],
        \end{align}
where in the last line we used the definition of adjoint map. Hence\begin{align}
    F_D(\mathcal{M}^{cq}) &=\frac{1}{d} \max_{\mathcal{D}^*_{R'\to B} \in \mathrm{CPU}(R'\to B)}\sum_{x\in \mathcal{X}}  p_X(x) \Tr[\rho_B^x\mathcal{D}^*_{R' \to B}(\ket{x}\bra{x}_{R'})]\\
    &= \frac{1}{d}\max_{\Lambda_B^x \geq 0, \sum_{x\in \mathcal{X}}\Lambda_B^x = \mathbbm{1}}\sum_{x\in \mathcal{X}}p_{X}(x) \Tr[\rho_{B}^{x}\Lambda_B^x]\\
     &\eqqcolon \frac{1}{d} P_{\mathrm{opt}}(\mathcal{M}^{cq}),
\end{align} where we defined $\Lambda_x \coloneqq \mathcal{D}^*_{R' \to B}(\ket{x}\bra{x}_{R'})$ and we used the fact that any $d-$-outcome (POVM) measurement $\{\Lambda_B^x\}_{x \in \mathcal{X}}$ can be formed by applying a CPU map to an ideal projective measurement $\{\ket{x}\!\bra{x}\}_{x \in \mathcal{X}}$. In the last, we defined the guessing probability in a quantum state discrimination with side information $\rho_{B}^x \equiv \mathcal{M}^{cq}_{R' \to B}(\ket{x}\!\bra{x}_{R'})$.
\item  Let $\Phi^{\mathcal{D}_i^*}_{RB_i}$ be optimal primal solutions for $F_D(\mathcal{M}_i)$, $i = 1, 2$. Their tensor product $\Phi^{\mathcal{D}_1^*}_{RB_1} \otimes \Phi^{\mathcal{D}_2^*}_{RB_2}$ is positive semidefinite and satisfies the primal constraint $\Tr_{R_1 R_2}[\Phi^{\mathcal{D}_1^*}_{RB_1} \otimes \Phi^{\mathcal{D}_2^*}_{RB_2}] = \mathbbm{1}_{B_1 B_2}/d^2$, so it is feasible for $F_D(\mathcal{M}_1 \otimes \mathcal{M}_2)$; evaluated on the objective, it achieves\begin{equation}
    \Tr[(\Phi^{\mathcal{M}_1}_{R_1 B_1} \otimes \Phi^{\mathcal{M}_2}_{R_2 B_2})(\Phi^{\mathcal{D}_1^*}_{RB_1} \otimes \Phi^{\mathcal{D}_2^*}_{RB_2})] = F_D(\mathcal{M}_1)\, F_D(\mathcal{M}_2),
\end{equation} proving $F_D(\mathcal{M}_1 \otimes \mathcal{M}_2) \geq F_D(\mathcal{M}_1)\, F_D(\mathcal{M}_2)$. Conversely, let $Y_{B_i}$ be optimal dual solutions satisfying $\Phi^{\mathcal{M}_i}_{R_i B_i} \leq \mathbbm{1}_{R_i} \otimes Y_{B_i}$. Taking the tensor product yields\begin{equation}
    \Phi^{\mathcal{M}_1}_{R_1 B_1} \otimes \Phi^{\mathcal{M}_2}_{R_2 B_2} \leq \mathbbm{1}_{R_1 R_2} \otimes Y_{B_1} \otimes Y_{B_2},
\end{equation} so $Y_{B_1} \otimes Y_{B_2}$ is dual-feasible for $F_D(\mathcal{M}_1 \otimes \mathcal{M}_2)$, giving the reverse inequality.
\item The Choi state of $\mathcal{N}$ has classical-quantum form on the flag register $Z$:
\begin{equation}
    \Phi^{\mathcal{N}}_{RZB} = \sum_{i=1}^\ell p_i\, \ket{i}\!\bra{i}_Z \otimes \Phi^{\mathcal{N}^i}_{RB},
\end{equation}
so that, via the link \eqref{eq:recovery_as_singlet_fraction} between the maximal recovery coefficient and the maximal singlet fraction, $F_{\mathrm{D}}(\mathcal{N}) = F_\Phi(\Phi^{\mathcal{N}}_{RZB})$. We prove the equality \eqref{eq:fidelity_recovery_flagged} by establishing matching upper and lower bounds.

\emph{Upper bound.} By convexity of $F_\Phi$ in its argument (item 2 of this lemma applied to the singlet fraction), 
\begin{equation}
    F_{\mathrm{D}}(\mathcal{N}) = F_\Phi\!\left(\sum_{i=1}^\ell p_i\, \ket{i}\!\bra{i}_Z \otimes \Phi^{\mathcal{N}^i}_{RB}\right) \leq \sum_{i=1}^\ell p_i\, F_\Phi\!\left(\ket{i}\!\bra{i}_Z \otimes \Phi^{\mathcal{N}^i}_{RB}\right).
\end{equation}
For each fixed $i$, since the $Z$ register is in the pure product state $\ket{i}\!\bra{i}$, any joint recovery map $\mathcal{D}_{ZB \to R'}$ acts effectively as the conditional channel $\mathcal{D}^i_{B \to R'}(\cdot) \coloneqq \mathcal{D}_{ZB \to R'}(\ket{i}\!\bra{i}_Z \otimes \cdot)$ on system $B$ alone, hence $F_\Phi(\ket{i}\!\bra{i}_Z \otimes \Phi^{\mathcal{N}^i}_{RB}) = F_\Phi(\Phi^{\mathcal{N}^i}_{RB}) = F_{\mathrm{D}}(\mathcal{N}^i)$.

\emph{Lower bound.} Conversely, for any family of recovery channels $\{\mathcal{D}^i_{B \to R'}\}_{i \in [\ell]}$, the map $\mathcal{D}^Z_{B \to R'} \circ \Pi_Z$ --- where $\Pi_Z$ measures $Z$ in the flag basis and $\mathcal{D}^Z$ applies $\mathcal{D}^i$ conditioned on the outcome $i$ --- is a valid feasible recovery map, hence
\begin{equation}
    F_{\mathrm{D}}(\mathcal{N}) \geq \sup_{\{\mathcal{D}^i\}_i} \sum_{i=1}^\ell p_i\, F_\Phi\!\bigl((\mathrm{id}_R \otimes \mathcal{D}^i)(\Phi^{\mathcal{N}^i}_{RB})\bigr) = \sum_{i=1}^\ell p_i\, F_{\mathrm{D}}(\mathcal{N}^i),
\end{equation}
where in the last step the supremum exchanges with the sum since the conditional channels $\mathcal{D}^i$ are unconstrained from each other. Combining the bounds gives \eqref{eq:fidelity_recovery_flagged}, and the lower-bound construction shows that the optimum is attained by a recovery map of the form $\mathcal{D}^Z_{B \to R'} \circ \Pi_Z$.

\item Direct consequence of Lemma \ref{lem:covariant_singlet_fraction} applied to \eqref{eq:recovery_as_singlet_fraction}.

\end{enumerate}
\end{proof}
\begin{remark}
    The maximal fidelity of recovery of a channel defined in \eqref{eq:SDP_Maximal_Recovery_Coefficient} directly compares with its quantum Doeblin coefficient \cite{Doeblin1937, George2025}, which finds operational interpretation in the 'opposite' task of quantum recovery of Figure \ref{fig:quantumrecoverytask_mes}, called \textit{quantum exclusion task}. Specifically, up to a normalization constant, the quantum Doeblin coefficient coincides with its minimal fidelity of recovery, and it can be written as an SDP analogous to \eqref{eq:SDP_Maximal_Recovery_Coefficient}, where the maximization becomes a minimization.
\end{remark}
\subsection{Conditional Min Entropy of a Channel}
In \cite{GourWilde2021} (see also \cite{Gour2019}) the authors defined the notion of conditional min-entropy of a quantum channel, starting from the max-channel divergence, defined as the SDP\cite{Datta2009, Cooney2016, Leditzky2018}:
\begin{align}
     D_\text{max}(\mathcal{N}\|\mathcal{M}) &\coloneqq \sup_{\rho_{RA}} D_\text{max}(\mathcal{N}_{A \to B}(\rho_{RA})\|\mathcal{M}_{A \to B}(\rho_{RA}) )\\
    &= D_\text{max}(\mathcal{N}_{A \to B}(\Phi_{RA})\|\mathcal{M}_{A \to B}(\Phi_{RA}) )\\
    &= \inf_{\lambda \in \mathbb{R}} \{\lambda: \mathcal{N} \leq 2^\lambda \mathcal{M}\}
\end{align}
where the last two equalities follow from \cite[Lemma 12]{Wilde2019} and \cite[Remark 30]{Wilde2019} and imply that such quantity can be computed efficiently as an SDP. Then, the \textit{channel's conditional min-entropy} was defined as \cite{GourWilde2021}:
\begin{align}
\label{eq:channel_conditional_min_entropy}
   H^\downarrow_{\text{min}}(\mathcal{N}) \coloneqq  - \log_2 \inf_{\lambda \geq 0} \{\lambda : \Phi^{\mathcal{N}}_{RB} \leq \lambda  \Phi^{\mathcal{N}}_R \otimes \mathbbm{1}_B \} \\
   =  - \log_2 \inf_{\lambda \geq 0} \{\lambda : \Phi^{\mathcal{N}}_{RB} \leq \lambda  \pi_R \otimes \mathbbm{1}_{B} \}.
\end{align}
Note that this definition uses a different underlying notion of conditional min-entropy with respect to \eqref{eq:conditionalMinEntropy}; namely, instead of taking the minimum over all $\sigma_B$ in the second argument it considers the reduced state of the first argument. This quantity sometimes appears in the literature \cite{Tomamichel2015} as $H_\mathrm{min}^\downarrow(B|R)_{\Phi^{\mathcal{N}}}$, hence the notation above. It remains unknown whether the quantity in \eqref{eq:channel_conditional_min_entropy} coincides with the \textit{extended channel's min entropy} \cite{GourWilde2021}:
\begin{equation}
\label{eq:ext_channel_min}
    H^{\uparrow}_{\mathrm{min}}(\mathcal{N}) \coloneqq H_{\mathrm{min}}(B|R)_{\Phi^{\mathcal{N}}}  = -\log_2 (d_B F^{(R)}_\Phi(\Phi^{\mathcal{N}}_{RB})),
\end{equation}
which uses \eqref{eq:conditionalMinEntropy} as underlying conditional min-entropy.

An alternative definition, motivated by the operational interpretation in terms of
recovery \cite{Koenig2009} and by our analysis in Section \ref{subsec:maximal_fidelity_preparation}, is
\begin{equation}
\label{eq:conditional_min_entropy_recovery}
    \tilde{H}_{\mathrm{min}}(\mathcal{N}) \coloneqq H_{\mathrm{min}}(R|B)_{\Phi^{\mathcal{N}}}  = -\log_2 (d\, F^{(B)}_\Phi(\Phi^{\mathcal{N}}_{RB})),
\end{equation}
or equivalently, by \eqref{eq:recovery_as_singlet_fraction}:
\begin{equation}
     \tilde{H}_{\mathrm{min}}(\mathcal{N})  = - \log_2 (d\, F_D(\mathcal{N})).
\end{equation}

Notice that by \eqref{eq:multiplicativity_maximal_fidelity_recovery}, this quantity is additive:
\begin{equation}
     \tilde{H}_{\mathrm{min}}(\mathcal{M}_1\otimes \mathcal{M}_2) =  \tilde{H}_{\mathrm{min}}(\mathcal{M}_1) + \tilde{H}_{\mathrm{min}}(\mathcal{M}_2),
\end{equation}
and by \eqref{eq:bounds_F_M} it satisfies the bounds
\begin{equation}
    -\log_2 d \leq \tilde{H}_{\mathrm{min}}(\mathcal{N}) \leq \log_2 d,
\end{equation}
which are saturated respectively, by isometric channels and replacement channel.
\begin{remark}
The quantities in \eqref{eq:ext_channel_min} and \eqref{eq:conditional_min_entropy_recovery} differ only in what is the conditioning system of $\Phi^{\mathcal{N}}_{RB}$ on which the underlying conditional min-entropy is evaluated, i.e. the optimizing system in the maximal singlet fraction. To illustrate this difference, consider the replacement channel $\mathcal{R}^{|0\rangle\langle 0|}(\rho) = |0\rangle\langle 0|$ for all input states $\rho$. Then:
\begin{align}
\label{eq:replacement_channel_1}
    H^{\uparrow}_{\mathrm{min}}(\mathcal{R}^{|0\rangle\langle 0|})
    &= H_{\mathrm{min}}(B|R)_{\Phi^{\mathcal{R}}} = 0, \\
\label{eq:replacement_channel_2}
    \tilde{H}_{\mathrm{min}}(\mathcal{R}^{|0\rangle\langle 0|})
    &= H_{\mathrm{min}}(R|B)_{\Phi^{\mathcal{R}}} = \log_2 d.
\end{align}
This difference can be physically interpreted by saying that \eqref{eq:ext_channel_min} is a measure of uncertainty of the channel output (\eqref{eq:replacement_channel_1} vanishes since the output state is pure), while \eqref{eq:conditional_min_entropy_recovery} is a measure of the channel's ability to preserve correlations with a reference system (\eqref{eq:replacement_channel_2} is maximum because $\mathcal{R}^{|0\rangle\langle 0|}$ breaks any correlation).
\end{remark}

Since it is a measure of correlation, the quantity $\tilde{H}_{\mathrm{min}}$ admits a natural comparison with the
channel's coherent information, defined as:
\begin{equation}
\label{eq:coherentInfoChannel}
    I_c(\mathcal{N}) \coloneqq \max_{\psi_{RA}}
    I_c(R\rangle B)_\omega = -\min_{\psi_{RA}} H(R|B)_\omega
\end{equation}
where $\omega_{RB} = (\mathrm{id}_R \otimes \mathcal{N})(\psi_{RA})$ and
$H(R|B)_\omega = H(RB) - H(B)$ is the von Neumann conditional entropy. Specifically, if one defines the \textit{channel's min-coherent information} as
\begin{equation}
\label{eq:channel_min_coherent_info}
    I_{c, \mathrm{min}}(\mathcal{N}) \coloneqq - \tilde{H}_{\mathrm{min}}(\mathcal{N}),
\end{equation}
then using the standard inequality $H_{\mathrm{min}}(R|B) \leq H(R|B)$ (see, e.g., \cite{Tomamichel2015}), we have:
\begin{equation}
     I_{c, \mathrm{min}}(\mathcal{N}) 
        = -H_{\mathrm{min}}(R|B)_{\Phi^{\mathcal{N}}}
        \geq -H(R|B)_{\Phi^{\mathcal{N}}}
        \eqqcolon I_c(\pi, \mathcal{N}),
\end{equation}
while
\begin{equation}
    \label{eq:Ic_vs_HminR}
        I_c(\mathcal{N}) 
        \geq I_c(\pi, \mathcal{N}),
    \end{equation}
since \eqref{eq:coherentInfoChannel} involves an optimization over all input states $\psi_{RA}$. There is a large class of covariant channels \eqref{eq:G_covariant_channel} for which $\psi_{RA} = \Phi_{RA}$, i.e. the MES is the optimal input in \eqref{eq:coherentInfoChannel}. Examples of these include the depolarizing channel in \eqref{eq:depolarizing_channel} and the quantum erasure channel \cite{Grassl1997}, defined as
\begin{equation}
\label{eq:erasure_channel}
    \mathcal{A}^{q}(\rho) = (1-q)\,\rho + q\,\mathrm{Tr}(\rho)\,|e\rangle\!\langle e|
\end{equation}
where $|e\rangle \in \mathcal{H}_B$ is an erasure flag orthogonal to all input states. For these channels, the chain of inequalities tightens to:
\begin{equation}
\label{eq:chain_covariant}
   I_{c,\mathrm{min}}(\mathcal{N}) \geq I_c(\mathcal{N}),
\end{equation}
and by the LSD theorem \cite{Lloyd1997, Shor2002, Devetak2005} we have:
\begin{equation}
        Q(\mathcal{N}) = \lim_{n \to \infty} \frac{1}{n}I_c(\mathcal{N}^{\otimes n})  \leq \lim_{n \to \infty}  \frac{1}{n}I_{c,\text{min}}(\mathcal{N}^{\otimes n}) = I_{c,\text{min}}(\mathcal{N}), 
    \end{equation}
so \eqref{eq:channel_min_coherent_info} provides a single-letter upper bound on their quantum capacity. While tight for the extreme cases (noiseless isometric channel and useless replacement channel), in general this bound will be very loose. This is because $ \tilde{H}_{\mathrm{min}}(\mathcal{N})$ is a one-shot quantity measure of recoverability determined entirely by its Choi state. For instance, it is easy to see that
\begin{equation}
        \tilde{H}_{\mathrm{min}}(\mathcal{D}^{q}) = \tilde{H}_{\mathrm{min}}(\mathcal{A}^{q}), \quad \forall q \in [0,1]
    \end{equation}
    where $\mathcal{D}^{q}$ is the depolarizing channel in \eqref{eq:depolarizing_channel} and $\mathcal{A}^q$ is the erasure channel in \eqref{eq:erasure_channel}. This is because the maximal fidelity of recovery is $F_D = 1 - 3q/4$ in both cases, even though we know that the noise acts according to markedly different mechanisms and the asymptotic transmission abilities of the two channels are very different.

\end{document}